\documentclass[aoas]{imsart}
\RequirePackage[OT1]{fontenc}
\RequirePackage{amsthm,amsmath,amsfonts,amssymb,amsbsy}
\RequirePackage[colorlinks,citecolor=blue,urlcolor=blue]{hyperref}
\RequirePackage{graphics}
\RequirePackage{graphicx}
\RequirePackage{appendix}

\RequirePackage{natbib}

\startlocaldefs
\numberwithin{equation}{section}
\theoremstyle{plain}
\endlocaldefs
\newcommand\ignore[1]{}

\newcommand{\hsp}{\hspace{0.1in} }
\newcommand{\hspp}{\hspace{0.05in} }
\newcommand{\hsppp}{\hspace{0.02in} }
\newcommand{\mb}[1]{\boldsymbol{#1}} %% use \mb{symb}
\newcommand{\Sb}{{\mb S}}
\newcommand{\indic}{\mbox{$1\!\!1$}}

\newtheorem{defn}{Definition}
\newtheorem{prop}{Proposition}
\newtheorem{thm}{Theorem}
\newtheorem{cor}{Corollary}

\begin{document}

\begin{frontmatter}
\title{Tracking Changes in Resilience and Level of Coordination in Terrorist
Groups} %\protect\thanksref{T1}}
\runtitle{Tracking Changes}
%\thankstext{T1}{Footnote to the title with the ``thankstext'' command.}
%\thankstext{T1}{This work has been supported in part by the U.S.\
%Defense Threat Reduction Agency (DTRA) under grant HDTRA-1-10-1-0086,
%the U.S.\ Defense Advanced Research Projects Agency under grant
%W911NF-12-1-0034, the U.S.\ National Science Foundation under grant
%DMS-1221888, and the U.S.\ Air Force Office of Scientific Research (AFOSR)
%via the MURI grant FA9550-10-1-0569 at the University of Southern California.}

\begin{aug}
\author{\fnms{Vasanthan} \snm{Raghavan,}\thanksref{m1}
\ead[label=e1]{vasanthan\_raghavan@ieee.org}}
\and
\author{\fnms{Alexander} \snm{G. Tartakovsky}\thanksref{m2}
\ead[label=e3]{alexg.tartakovsky@gmail.com}
%\ead[label=u1,url]{http://cams.usc.edu/usr/facmemb/tartakov/} %http://www.foo.com}
}

%\thankstext{t1}{Some comment}
%\thankstext{t2}{First supporter of the project}
%\thankstext{t3}{Second supporter of the project}
\runauthor{Raghavan and Tartakovsky}

\affiliation{Qualcomm Flarion Technologies\thanksmark{m1}
and University of Connecticut\thanksmark{m2}}

\address{Address of First author \\
2 Petunia Drive, Apt.\ 2H \\
North Brunswick, NJ 08902 \\
%\phantom{E-mail:\ }
E-mail: \printead*{e1}
%\\ \printead{u1}
}

{\vspace{0.1in}}
\address{Address of Second author \\
%Department of Statistics \\
University of Connecticut \\
Storrs, CT 06071 \\
%\phantom{E-mail:\ }
E-mail: \printead*{e3} }

\end{aug}

\begin{abstract}
\noindent
Activity profiles of terrorist groups show frequent spurts and downfalls corresponding to
changes in the underlying organizational dynamics. In particular, it is of interest in
understanding changes in attributes such as intentions/ideology, tactics/strategies,
capabilities/resources, etc., that influence and impact the activity. The goal of this
work is the quick detection of such changes and in general, tracking of macroscopic as
well as microscopic trends in group dynamics. Prior work in this area are based on
parametric approaches and rely on time-series analysis techniques, self-exciting hurdle
models (SEHM), or hidden Markov models (HMM). While these approaches detect spurts and
downfalls reasonably accurately, they are all based on model learning --- a task that
is difficult in practice because of the ``rare'' nature of terrorist attacks from a
model learning perspective. In this paper, we pursue an alternate non-parametric
approach for spurt detection in activity profiles. Our approach is based on binning
the count data of terrorist activity to form observation vectors that can be compared
with each other. Motivated by a {\em majorization theory} framework, these vectors are
then transformed via certain functionals and used in spurt classification. While the
parametric approaches often result in either a large number of missed detections of
real changes or false alarms of unoccurred changes, the proposed approach is shown to
result in a small number of missed detections {\em and} false alarms. Further, the
non-parametric nature of the approach makes it attractive for ready applications in a
practical context.
\end{abstract}

%The first approach is based on the Exponential Weighted Moving-Average (EWMA) filter
%applied to the count data of terrorist activity followed by binning the resultant
%in-control/out-of-control decisions for further spurt classification. The second
%approach is based on binning the count data


\end{frontmatter}

\section{Introduction}
\label{sec1}
Changes in the organizational dynamics of terrorist groups lead to either spurts
or downfalls in their activity profiles. It is of interest in detecting such
changes, associating these changes to specific macroscopic changes in group
dynamics, and in tracking these dynamics over time. Prior work in this area has primarily
been of a parametric nature.

Initial work on monitoring terrorist network activity profiles follows the
interrupted time-series framework where the main goal is to study whether certain
strategic policy interventions lead to statistically significant reduction in
certain types of attacks and/or if different types of attacks act as substitutes
for/complements of each other. This is achieved by isolating the time %before and after a
of intervention, %is introduced,
fitting (potentially different) threshold vector auto-regression %(TAR)
models to the time-series data before and after the intervention is introduced, and
inferencing on %the efficacy of the policy intervention
its efficacy; see, works
by~\citet{Landes_1978,Cauley_Im_1988,enders_sandler1993,enders_sandler2000}, %Brophy-Baermann_Conybeare_1994},
%and~\citet{Enders_Sandler_Cauley_1990},
for example. In another parametric direction, group-based trajectory analyses
are adopted by~\citet{dugan_lafree_piquero} and~\citet{lafree_morris_dugan} to identify
regional terrorism trends %with similar developmental paths
via the use of Cox proportional hazards model or zero-inflated Poisson model. A similar
philosophy of identifying common trends across multiple terrorist groups is also
adopted by~\citet{breiger,melamed} and~\citet{bakker_dark_network}.

%In terms of parametric modeling of activity profiles of terrorist groups, two recent
Two recent approaches in modeling activity profiles have been along
the directions of: i) self-exciting hurdle models
(SEHM)~\citep{hawkes_1971,mohler,porter_white2012} and ii) hidden Markov
models (HMM)~\citep{vasanth_aoas2013}. While both approaches leverage the
sparsity of activity profiles and account for clustering of attacks,
they do so via different mechanisms. In the SEHM approach, the hurdle component
creates data sparsity by ensuring a pre-specified density of zero counts, while the
self-exciting component induces clustering of data. In the HMM approach, an
increase or decrease in the attack intensity is attributed to switching between
internal states that captures the dynamics of the group's evolution.

\begin{figure}[tbh!]
\begin{center}
\begin{tabular}{c}
\includegraphics[height=2.5in,width=3.5in] {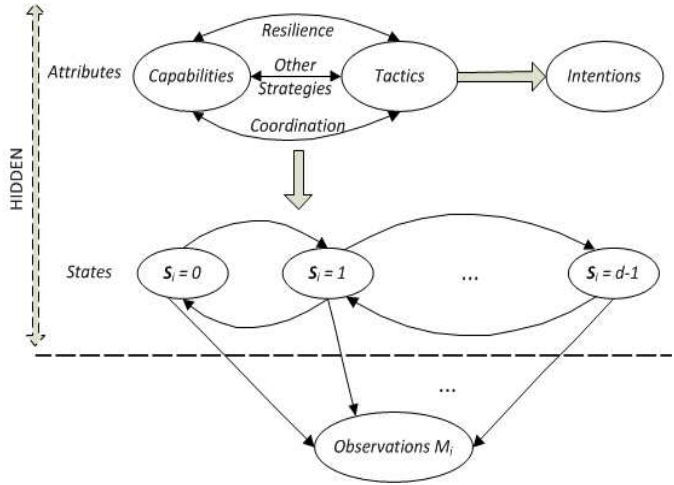}
\end{tabular}
\caption{\label{fig_prob_form}
A typical mechanistic model %from~\cite{rand}
capturing the dynamics of a terrorist group with connections between the
underlying attributes, states and observations.}
\end{center}
\end{figure}

The HMM framework~\citep{vasanth_aoas2013,vasanth_compstudy2015} provides good
explanation/prediction capability of past/future activity
%and good prediction capability of future activity
across a large set of terrorist groups with different ideological
attributes and is thus of main
focus in this work. This framework is motivated by a typical mechanistic model
proposed by~\citet{rand} and illustrated in Fig.~\ref{fig_prob_form} that
captures the complex correlations in time and network structure of the activity
profile. Some of the terrorist group attributes capturing the dynamics in this
model include its {\em Intentions}, its {\em Capabilities}, and the underlying
{\em Tactics} deployed by it to utilize its {\em Capabilities} in realizing its
{\em Intentions}. Of particular interest in this work are {\em Tactics} that
reflect a strong {\em Resilience} and/or {\em Coordination} in the group since they
capture counter-terrorism dynamics; see Sec.~\ref{sec2} for precise definitions
and works
by~\citet{sageman,rand,santos,bakker_dark_network,lindberg,sandler_2011,breiger,melamed}
for motivations in tracking these particular {\em Tactics}.

The work by~\citet{vasanth_aoas2013} bins the activity data into decision interval blocks,
leverages the underlying HMM structure, and proposes a state estimation strategy for spurt
detection under certain assumptions. While useful, technical difficulties %in deriving
%updates equations for the HMM framework
ensured that a complementary view of this strategy
where the states are estimated over the entire interval and are then binned into decision
interval blocks was not considered by~\citet{vasanth_aoas2013}. We overcome these technical
difficulties and provide this missing link in this work. We note that both approaches lead
to acceptable inferencing performance on the terrorist group and allow attribution to specific
{\em Tactics} that could have induced a change in the overall activity profile. However, it is
unclear if either approach is optimal from an inferencing perspective for the underlying
{\em Tactic}. Further, both
approaches suffer from a fundamental issue in that they are retrospective or non-causal (models
are first learned over the data followed by inferencing). Thus, such approaches are difficult
to implement in practice since terrorist activity data is sparse from a model learning perspective
and latencies in model learning could render assumptions on model stability questionable.

Given this backdrop, the major contributions of this work are as follows.
\begin{itemize}
\item We propose a non-parametric approach to detect changes in the activity
profile and to attribute them to specific changes in the underlying {\em Tactics}
deployed by the group. For this, we develop an application of majorization theory
(commonly used as a partial ordering to compare probability vectors) in terrorism
analysis. Motivated by a theory of reverse majorization, we build on the partial
ordering via the use of certain functionals that serve as a proxy for complete
ordering of attack frequency vectors. We then identify a subset of these functionals
to capture changes in {\em Resilience} (or {\em Coordination}) better and use these
associations to track changes in group dynamics.

\item
We conduct extensive numerical studies with both data generated from the mechanistic
model in Fig.~\ref{fig_prob_form} as well as real data from the {\em Fuerzas Armadas
Revolucionarias de Colombia} (FARC) terrorist group from Colombia~\citep{rdwti}. We
show that the proposed non-parametric
approach %competes well with both parametric approaches based on the HMM framework by
results in both a small number of false alarms declaring changes in {\em Tactics} when
there are none {\em and} a small number of missed detections of real changes in {\em Tactics}.
On the other hand, the parametric approaches based on the HMM framework either result in low
probability of false alarm or low probability of missed detection, but not both.

\item
These observations suggest
that the non-parametric approach provides a suitable compromise not only in terms of its
practical utility, but also in terms of performance with real terrorist data.
\end{itemize}

%Leveraging this
%model,~\citet{vasanth_aoas2013} propose a $d = 2$-state HMM framework with a hurdle-based
%geometric density for the number of attacks (observations) in either state.

%From a class of many one- and two-parameter observation models, it is shown
%that the hurdle-based geometric model fits the FARC dataset from the RAND Database
%on Worldwide Terrorism Incidents~(\citeauthor{rdwti}) best. From a
%comparative study between the SEHM and the HMM frameworks for the FARC dataset and
%the Indonesia/Timor-Leste dataset from the Global Terrorism Database (GTD)
%(see the work by~\citet{lafree1} for details),
%it is shown that the HMM framework has a better generalization capability (prediction of
%time to next attack) relative to the SEHM framework for both datasets, whereas in terms
%of explanatory power, both HMM and SEHM perform reasonably well, with neither framework
%outperforming the other.
%The HMM framework is seen to be better for the FARC dataset, whereas
%the SEHM is seen to be better for the Indonesia/Timor-Leste dataset.
%Further,~\citet{vasanth_compstudy2015} performs a comparative analysis of the %TAR model,
%SEHM and HMM frameworks across a large set of terrorist groups with different ideological
%attributes and establishes the efficacy of the HMM framework for both explanative and
%predictive capacities for a vast majority of the groups.

\section{Problem Setup}
\label{sec2}
The observations capturing the dynamics of a terrorist group
%(also denoted as its activity profile)
are multivariate and are of a mixed (categorical, ordinal and interval variables)
type, %. Specifically, observations in
%terrorism modeling are made of categorical, ordinal and interval variables,
e.g.,
time and location of attacks, type of ammunition used, (apparent) sub-group of the
group involved, intensity and impact of the attacks, etc. In addition, the
observations can suffer from impairments such as missing data, mislabeled data,
temporal and attributional ambiguity, transcribing errors, etc.
%Toward addressing the spurt detection problem, w
We start by developing a temporal model for the activity profile by
discarding %\footnote{Capturing the discarded variables in a model for terrorist
%activity is an important problem, but is far more complicated than the scope of
%this work. This issue will be addressed in future work.}
the categorical and ordinal variables.

\subsection{Temporal Modeling of Activity Profiles}
\label{sec_2a}
Let the first and last day of the time-period of interest be denoted as Day $1$ and
Day ${N}$, respectively. Let ${\mb M}_i$ denote the number of terrorism incidents on the
$i$th day of observation, $i = 1, \cdots, {N}$. Note that ${\mb M}_i \in \{0, 1,
2, \cdots \}$ %can take values from the set $\{ 0, 1, 2, \cdots \}$
with ${\mb M}_i = 0$ corresponding to no terrorist activity
on the $i$th day of observation. %On the other hand, there could be multiple terrorism
%incidents corresponding to independent attacks on a given day reflecting a high level
%of coordination between different sub-groups of the group.
Let ${\mb H}_i$ denote the
history of the group's activity till (and including) day $i$. That is, ${\mb H}_i =
\left\{ {\mb M}_1, \cdots, {\mb M}_i \right\}, \hsppp i = 1,2, \cdots, N$ with ${\mb H}_0
\triangleq \varnothing$ (denoting the null set). The temporal point process model is
completely specified when ${\sf P}\left( {\mb M}_i = r | {\mb H}_{i-1} \right)$ is known
as a function of ${\mb H}_{i - 1}$ for all $i = 1, \cdots, N$ and $r = 0,1,2, \cdots$.

With the HMM framework,~\citet{vasanth_aoas2013} hypothesize that the observations ${\mb M}_i$
depend only on certain hidden states $\Sb_i$ (such as {\em Intentions}, {\em Tactics}, or
{\em Capabilities}) in the sense that ${\mb M}_i$ is conditionally independent of ${\mb H}_{i-1}$
and $\Sb_{i-1}$ given $\Sb_i$. Further, %~\citet{vasanth_aoas2013}
they assume a time-homogenous
one-step Markovian evolution for $\Sb_i$ with a $d$-state model to capture the dynamics of
the group over time. That is, $\Sb_{i} \in \{ 0, 1 , \cdots , d-1 \}$ with each distinct
value corresponding to a different level in the underlying attribute of the group. Using
these two hypotheses, the temporal point process model can be written as
\begin{align}
& {\hspace{0in}} {\sf P} \left( {\mb M}_i = r | {\mb H}_{i-1} \right)
%\nonumber \\ & \hspp \hspp
%= \sum_{j = 0}^{d-1} {\sf P} \left( {\mb M}_i = r, \hsppp \Sb_{i} = j %, \hsppp \Sb_{i-1} = k
%| {\mb H}_{i-1} \right) \nonumber \\
%& \hspp \hspp {\hspace{1.05in}}
%= \sum_{j = 0}^{d-1}
%{\sf P} \left( {\mb M}_i = r | \Sb_{i} = j , \hsppp {\mb H}_{i-1} \right) \cdot
%{\sf P} \left( \Sb_{i} = j | {H}_{i-1} \right) \nonumber \\
%& \hspp \hspp {\hspace{1.05in}}
%= \sum_{j = 0}^{d-1} \sum_{k = 0}^{d-1}
%{\sf P} \left( {\mb M}_i = r | \Sb_{i} = j , \hsppp {\mb H}_{i-1} \right) \cdot
%{\sf P} \left( \Sb_{i} = j, \hsppp \Sb_{i-1} = k | {\mb H}_{i-1} \right) \nonumber \\
%& \hspp \hspp {\hspace{1.05in}}
= \sum_{j = 0}^{d-1} \sum_{k = 0}^{d-1}
{\sf P} \left( {\mb M}_i = r | \Sb_{i} = j  \right) \cdot
{\sf P} \left( \Sb_{i} = j, \hsppp \Sb_{i-1} = k \right).
\nonumber
\end{align}
The trade-off between accurate modeling of the group's attributes (larger $d$ is
better for this goal) versus estimating more model parameters\footnote{The number of
model parameters in the HMM framework is $d(d - 1 + \ell)$ where $\ell$ is the (common)
number of observation density parameters in each state.} (smaller $d$ is better for
this goal) is resolved by~\citet{vasanth_aoas2013} by focussing on mature terrorist
groups (where the {\em Intentions} %and {\em Tactics}
attribute remains stable) and by considering a $d = 2$ setting. This trade-off
corresponds to a binary quantization of the group's {\em Tactics} and {\em Capabilities}
into {\em Active} and {\em Inactive} states.

For the observations, a simple two-parameter model such as the {\em hurdle-based geometric}
density, defined as,
\begin{eqnarray}
{\sf P}( {\mb M}_i = r | \Sb_i = j) \triangleq  {\sf HBG} \left( \mu_j, \gamma_j \right)
= \left\{
\begin{array}{cc}
1 - \gamma_j, & r = 0 \\
\gamma_j (1 - \mu_j) \cdot (\mu_j)^{r-1}, & r \geq 1
\end{array}
\right.
%\nonumber
\label{eq_hbg}
\end{eqnarray}
can be hypothesized. The intuition behind the hurdle-based geometric model is that
the terrorist group remains {\em oblivious} of its past activity and continues to
attack with the same {\em Tactics} as before, as long as its short-term objective
is met, provided a certain group resistance/hurdle has been overcome. The special
case where there is no group resistance to this aforementioned strategy is obtained
by setting $\mu_j = \gamma_j$, resulting in a geometric observation density.

\subsection{Underlying Assumptions and Problem Statements}
\label{sec_2b}
We make the following assumptions in this work.

\noindent {\bf \em \underline{Assumption 1:}} Motivated by the efforts
in~\citep{vasanth_aoas2013,vasanth_compstudy2015}, %in the sequel,
we assume that %the observations corresponding to
terrorist activity can be accurately described by a $d = 2$-state HMM with
observations following the hurdle-based geometric density in~(\ref{eq_hbg}).
Specifically, let ${\cal H}_j$ denote
the hypothesis that $\Sb_i = j$ and the observation model is given as %can then be written as
\begin{eqnarray}
{\cal H}_j \hsppp : \hsppp {\mb M}_i \sim {\sf HBG} \left( \mu_j, \gamma_j \right),
\hsppp j \in \{ 0, 1 \}, \nonumber
\end{eqnarray}
with a state transition probability matrix
\begin{eqnarray}
{\bf T}( {\sf p}_0, \hsppp {\sf p}_1) = \left[ \begin{array}{cc}
1-{\sf p}_0 & {\sf p}_0 \\
{\sf q}_0 & 1-{\sf q}_0
\end{array} \right]
\nonumber
\end{eqnarray}
capturing the dynamics of evolution from $\Sb_{i-1}$ to $\Sb_{i}$.

\noindent {\bf \em \underline{Assumption 2:}} With the mechanistic model
of~\cite{rand} as the backdrop, we are primarily interested in two specific
types of {\em Tactics} deployed by the group: Those %Of interest here are those
{\em Tactics} reflecting i) {\em Resilience} and ii) a high level of {\em Coordination}
in the group. These {\em Tactics} are important since they determine the broad
outline of counter-terrorism policies and measures %initiated and
sustained by the establishment~(see related works
by~\citet{sageman,rand,santos,lindberg,sandler_2011,bakker_dark_network,breiger,melamed}
for motivations on the focus on these {\em Tactics}). To be
specific, resilience is defined as the ability of the group to sustain terrorist
activity over a number of days and this ability reflects the group's capacity to
rejuvenate itself from asset (manpower, material, and skill-sets) losses. On the
other hand, coordination is defined as the ability of the group to launch multiple
attacks over a given time-period and this ability reflects its capacity to coordinate
the group's assets necessary for simultaneous action over a wide geography.

\noindent {\bf \em \underline{Assumption 3:}} Assuming a stable set of {\em Intentions}
for the group (e.g., mature groups) and with a focus on its {\em Tactics} and
{\em Capabilities}, the first problem of interest in this work is to quickly arrive at
specific/microscopic inferencing decisions on disruptions in the group's activity profile
(with typical interest on spurts and downfalls). The second problem of interest is of
a broad/macroscopic nature: whether these disruptions could be attributed either
to a change in the group's resilience, or a change in the level of coordination between
different sub-groups of the group, or both of these attributes.

\noindent {\bf \em \underline{Assumption 4:}} Inferencing %on the group's
%%{\em Capabilities} and {\em Tactics} %with $\{ {\mb M}_i \}$
%resilience and coordination
with $\{ {\mb M}_i \}$ on a daily basis could lead to a performance mirroring
the potential rapid fluctuations in the observations. This is particularly
disadvantageous in making global policy decisions on the group.
%based on the group's attributes such as resilience and coordination.
To overcome this problem, we propose inferencing over a $\delta > 1$ day
{\em disjoint} time-window. For this, we decompose the time-period of interest
into disjoint time-windows, $\Delta_n, \hspp n = 1, 2, \cdots, K$, where
$\Delta_n = \{ (n-1) \delta+1, \cdots , n \delta \}$ and
$K = \lfloor N/\delta \rfloor$. The appropriate choice of $\delta$ is determined
by the group dynamics and the timelines for inferencing decisions with typical
choices being $7$ or $14$ days corresponding to a weekly or a bi-weekly decision
process. %An alternate motivation for $\delta$ is provided in Supplementary A.

\begin{figure}[tbh]
\begin{center}
\begin{tabular}{c}
%\begin{minipage}{5.5in}
%\centerline{
\includegraphics[height=0.47in,width=4.9in] {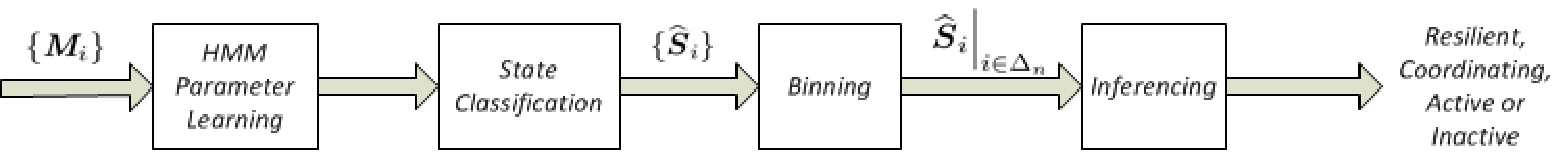}
\\
(a)
\\
\includegraphics[height=0.9in,width=5.4in] {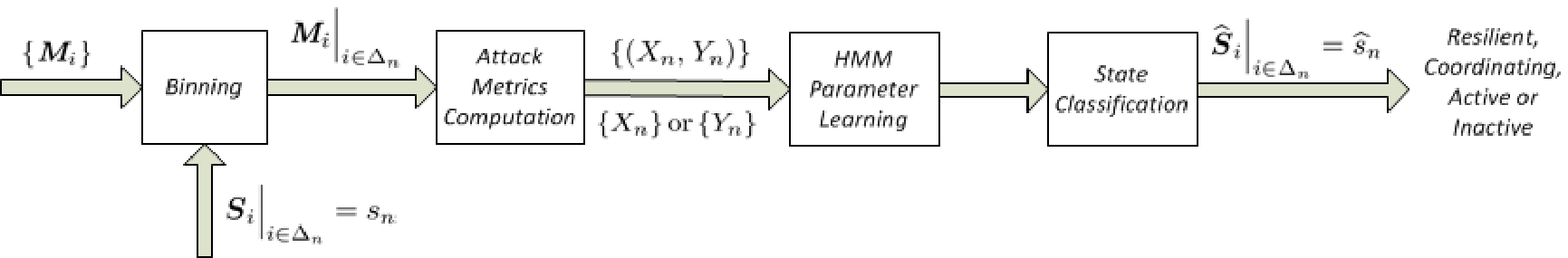}
\\
(b)
 %&
%\includegraphics[height=2.2in,width=2.6in] {farc_res_coord.eps}
%\\ (a) & (b)
%}
%\end{minipage}
\end{tabular}
\caption{\label{fig_hmm}
%(a)
Broad outline of HMM-based parametric approaches: (a) Based on state
classification with $\{ {\mb M}_i \}$ as input followed by binning, and (b) Based on
binned sets of $\{ {\mb M}_i \}$ as input followed by state classification.
 }
\end{center}
\end{figure}

\section{Parametric Approaches for Spurt Detection Based on HMM Structure}
\label{sec3}
The spurt detection problem is similar in %spirit and
objective to the changepoint problem
of detecting sudden/abrupt changes in the statistical nature of observations. The theory of
changepoint detection has matured significantly and different procedures %such as Shewhart chart,
%Cumulative Sum (CUSUM) chart, Exponentially Weighted Moving Average (EWMA) control chart,
%Shiryaev-Roberts (SR) procedure, %generalized SR procedure,
%etc.,
have been developed %over the last few decades
(see books by~\citet{BasNIk,TartNikBas} %poor_olympia
for details). Fundamentally speaking, a changepoint procedure is equivalent to an update equation
for the test statistic based on the likelihood ratio of the observations. This test statistic
is tested against a threshold (which is chosen to meet appropriate false alarm constraints) to
lead to a change decision. %(spurt/downfall detection).
Developing the structure of the update equation as well as setting the threshold require
knowledge of the pre-change and post-change parameters.

A na\"{i}ve approach to leverage this theory in the context of this work is to ignore the
contextual connections between the hidden states and the observations from a terrorist group in
developing the update equation for the test statistic of any chosen changepoint procedure,
apply it to detect change, and restart it once a change decision (spurt/downfall detection)
has been made, so that the next disruption can be monitored.
%In this sense, spurt detection is primarily concerned with
%selection of the appropriate procedure and its repeated application, but not with the design
%of the algorithm itself.
Ignoring the connections between the states and the observations could potentially
lead to poor performance in detecting the change %as well as in not providing
%specific attributions to changes in the resilience or coordination structure of the group.
as well as poor decisions on changes in group dynamics (such as resilience and
coordination) and is thus best avoided.

We now develop approaches that leverage these connections and are tailored to
terrorist group dynamics. Toward our goal, we consider the following set of
{\em attack metrics} that capture
%resilience and coordination in
different attributes of the group: i) $X_n$, the number of days of terrorist
activity, and ii) $Y_n$, the total number of attacks, both within the $\Delta_n$
time-window:
%a $\delta$-day time-window leading up to the $i$th day. That is, %$\Delta_n$:
\begin{eqnarray}
X_n  = \sum_{i \in \Delta_n}   %{j = (n-1)\delta+1}^{n \delta}
\indic \left( %\{
{\mb M}_i > 0 %\}
\right); \quad Y_n  = \sum_{i \in \Delta_n}   %{ j = (n-1)\delta+1}^{n \delta}
{\mb M}_i , \quad
n = 1,2,\cdots, K %i = \delta, \delta + 1, \cdots ,
\label{eq_XnYn}
\end{eqnarray}
where $\indic(\cdot)$ denotes the indicator function of the set under consideration.
Note that $Y_n/\delta$ is the average number of attacks per day and thus $Y_n$ is a
reflection of the intensity of attacks launched by the group. In general, $X_n$ is
more indicative of resilience in the group, whereas $Y_n$ captures the level of
coordination better.

%\noindent{\bf \em \underline{Inferencing with $\{ M_i \}$:}}
\subsection{Inferencing with $\{ {\mb M}_i \}$}
\label{subsec_3a}
The simplest method to leverage the underlying HMM structure and develop a parametric
scheme to classify the hidden states %as {\em Active} or {\em Inactive}
is illustrated in Fig.~\ref{fig_hmm}(a) and is described as follows.
The %HMM structure corresponding to the
observation sequence $\{ {\mb M}_i \}$, with density as
given in~(\ref{eq_hbg}), is used with the classical Baum--Welch algorithm~\citep{rabiner}
to learn the observation density parameters $\mu_j$ and $\gamma_j$ ($j = 0,1$) as well as
the initial probability density ($\pi_0$ and $\pi_1$) and state transition probability
matrix parameters (${\sf p}_0$ and ${\sf q}_0$); see update equations in Supplementary A.
With the converged Baum--Welch parameter estimates as initialization, the Viterbi
algorithm~\citep{rabiner} is then used to estimate the most probable state sequence
$\{ \widehat{\Sb}_i \}$ given the observations. At this stage, the states can only be
classified as {\em Active} or {\em Inactive} ($\widehat{\Sb}_i \in \{0,1\}$), with no
attribution to any specific mechanism that could result in the corresponding observations.
%Some of the main difficulties with the above method are as follows.
To ensure inferencing on spurts/downfalls and the source of such disruptions
(resilience and/or coordination) over disjoint time-windows,
one approach is to accumulate the binned state classifications $\widehat{\Sb}_i
\Big|_{ i \in \Delta_n}$ and infer on the specific mechanism leading to the
observations:
\begin{eqnarray}
\label{eq_dec0}
f\left(  \widehat{\Sb}_i \Big|_{ i \in \Delta_n} \right) > \widehat{\eta}
\hspp {\sf and} \hspp g \left(X_n, \hsppp Y_n \right)
> \overline{\eta}
%\nonumber
\end{eqnarray}
for an appropriate choice of $f(\cdot)$, $g(\cdot)$, $\widehat{\eta}$ and $\overline{\eta}$
(see Supplementary A for details).

%\noindent {\bf \em \underline{Inferencing with $\{ (X_n, Y_n) \}$:}}
\subsection{Inferencing with $\{ (X_n, Y_n) \}$}
\label{subsec_3b}
In an alternate approach built on $\{ X_n \}$ or $\{ Y_n \}$ as observation
sequence, pictorially illustrated in Fig.~\ref{fig_hmm}(b), we assume
that the hidden state remains fixed over $\Delta_n$: %the time-window:
%\begin{eqnarray}
%\Sb_{ i} \Big|_{ i \in \Delta_n } = s_n,
%\hspp s_n \in \{ 0 , 1 \}. \nonumber
%%1, \cdots , d-1 \} . \nonumber
%%\hspp \hspp {\rm for} \hspp
%% = (n-1)\delta + 1, \cdots, n\delta, \hspp n = 1, 2 , \cdots . \nonumber
%\end{eqnarray}
$\Sb_{ i} \Big|_{ i \in \Delta_n } = s_n,
\hspp s_n \in \{ 0 , 1 \}$.
The reason for this binning assumption (on the state) is to provide an
explicit attribution to macroscopic dynamics in the group unlike the case
in~Sec.~\ref{subsec_3a}. Further, the density function of $X_n$ or
$Y_n$ becomes difficult to write in closed-form without this assumption. Our goal is to infer $s_n$
with the aid of the appropriate attack metrics %an appropriate set of observations
corresponding to $\Delta_n$.

Since $X_n$ captures resilience, inferencing\footnote{Note that only inferencing with
$\{ (X_n , \hsppp Y_n) \}$ is performed in the hurdle-based geometric setting
by~\citet{vasanth_aoas2013}. Inferencing with $\{ X_n \}$ or $\{ Y_n \}$ is not
performed due to technical difficulties in deriving update equations. We overcome this
difficulty in Supplementary A.} on resilience in the group
is performed with $\{ X_n \}$ as the observation sequence. %and the density function in~(\ref{eq_Xn}).
Similarly, inferencing on coordination is performed with $\{ Y_n \}$,
%and the density function in~(\ref{eq_Yn}),
whereas a joint inferencing on the group's activity is performed with the joint sequence
$\{ (X_n , \hsppp Y_n) \}$.
%and the density function in~(\ref{eq_jointseq_density}), as described in~\cite{vasanth_aoas2013}.
With the hurdle-based geometric model in~(\ref{eq_hbg}), the joint density of
$(X_n , \hsppp Y_n)$ is given as~(see~\citep[Supplementary A]{vasanth_aoas2013} for
details)
\begin{align}
%& {\hspace{0.0in}}
{\sf P} \left( X_n = k, \hsppp Y_n = r \Big|  \Sb_{ i} |_{ i \in \Delta_n } = j
\right) = {\delta \choose k} { r -1 \choose r - k} \cdot
%\nonumber \\ & {\hspace{2in}}
(1 - \gamma_j)^{\delta  - k} (\gamma_j)^{k} \cdot (1 - \mu_j)^k
(\mu_j)^{r - k}, \hspp r \geq k.
%\nonumber
\label{eq_jointseq_density}
\end{align}
An important property of the density function in~(\ref{eq_jointseq_density}) is that it
decomposes into a product of two terms, each depending on only one of the two
parameters, $\gamma_j$ and $\mu_j$. This product structure leads to a simplified update
equation for parameter estimation (see~Supplementary A for details).
%~(see~\cite[Supplementary A]{vasanth_aoas2013} for details).
On the other hand, since $X_n$ counts the number of days of activity, it is a binomial
random variable with parameters $\delta$ and $\gamma_j$:
%(that is, $X_n \sim {\rm Binomial}(\delta, \hsppp \gamma_0)$).
\begin{eqnarray}
{\sf P} \left(X_n = k \Big|  \Sb_{ i} |_{ i \in \Delta_n } = j \right) =
{\delta \choose k} \cdot (\gamma_j )^k \cdot \left( 1 - \gamma_j \right)^{\delta - k}.
%\nonumber
\label{eq_Xn}
\end{eqnarray}
It is to be observed that the above density function depends only on $\gamma_j$ and
hence, inferring of $\mu_j$ is impossible with $\{ X_n \}$ as the observation sequence.
This difficulty is overcome in Supplementary A by using an estimate of $\mu_j$
based on all the observations $\{ {\mb M}_i \}$. Finally, the density function of $Y_n$ is
obtained from~(\ref{eq_jointseq_density}) by summing over the $k$ variable: %to result in:
\begin{eqnarray}
{\sf P} \left(Y_n = r \Big|  \Sb_{ i} |_{ i \in \Delta_n } = j \right) =
(1 - \gamma_j)^{\delta} \cdot (\mu_j)^{r }\cdot
\sum_{ k = 1}^{\min(r, \hsppp \delta)} {\delta \choose k} { r -1 \choose r - k} \cdot
{\sf A}^k
\label{eq_Yn}
\end{eqnarray}
where ${\sf A} = \frac{ (1 - \mu_j) \hspp \gamma_j } { (1 - \gamma_j) \hspp \mu_j }$.
While the above expression can be rewritten in terms of Gauss hypergeometric
functions, it appears to be not easily amenable to a closed-form expression rendering
a product decomposition (for parameter estimation) in the two parameters difficult, if not
impossible; see~Supplementary A for details. Under the assumption that $\mu_j > \gamma_j$
(or equivalently, ${\sf A} < 1$), update equations for the observation density parameter
estimates are obtained in Supplementary A in the $\delta \gg 1$ regime.

State classification is performed %retrospectively
using the model parameter estimates from the use of Baum--Welch algorithm on the
appropriate observation sequence ($\{ X_n \}$, $\{ Y_n \}$, or $\{ (X_n, \hsppp Y_n) \}$).
The output of the Viterbi algorithm is a state estimate for the period of interest
\begin{eqnarray}
\Big\{ \widehat{\Sb}_i \Big|_{ i \in \Delta_n}
= \widehat{s}_n \hsppp \in \hsppp \{ 0, 1 \}
\hspp {\rm for} \hspp {\rm all} \hspp
%i \in \Delta_n \hspp {\rm and} \hspp
n = 1, \cdots, K \Big \}. \nonumber
\end{eqnarray}
A state estimate of $1$ with $\{ X_n \}$ as the observation sequence (correspondingly
with $\{ Y_n\}$ and the joint sequence $\{ (X_n, \hsppp Y_n) \}$) indicates that the
group is resilient (coordinating and both, respectively) over the period of
interest, whereas an estimate of $0$ indicates that the group is non-resilient
(non-coordinating or neither, respectively). Transition between states
indicates spurt/downfall in the activity corresponding to the appropriate attribute.
%This approach is pictorially illustrated in Fig.~\ref{fig1}(b).
%This state sequence estimate
%$\{ \widehat{s}_n \}$ is used as the ground truth to compare the performance of the
%non-parametric methodologies studied next.

%Fundamentally speaking, both the approaches in Fig.~\ref{fig_hmm}(a) and~\ref{fig_hmm}(b)
%require knowledge of the underlying parameter estimates in state classification and in this
%sense, they are {\em retrospective (non-causal)}. The development of causal (non-parametric)
%approaches is considered next.

%\noindent {\bf \em \underline{Advantages and Disadvantages:}}
%\subsection{Tradeoffs between Inferencing with $\{ M_i \}$ and $\{ (X_n, Y_n) \}$}
%\label{subsec_3c}

\subsection{Difficulties with Parametric Approaches}
\label{subsec_3c}
While both approaches in Sec.~\ref{subsec_3a} and Sec.~\ref{subsec_3b} exploit the HMM
structure in different ways, they also require a reasonable knowledge of the underlying
parameter estimates for state classification. Acquiring such knowledge %could
leads to a %significant
latency in inferencing. %due to model learning.
%This latency in inferencing could be expected to be more
%acceptable from an application standpoint with the approach in Sec.~\ref{subsec_3a}
%than with the approach in Sec.~\ref{subsec_3b}. %$\{ {\mb M}_i \}$ than with $\{ (X_n, Y_n) \}$.
%This is because model learning with a similar number of observations leads to the
%requirement of a larger number of data points (training period) and a larger model
%learning latency for the latter approach (due to binning of observations) than with
%the former approach.
%However, more s
Specific to the context of this work, %despite the recent surge in media attention
%on trans-national terrorist activities and insurgencies,
terrorism incidents are ``rare'' from
the perspective of model learning, even for some of the most active\footnote{While a case can be
made that these datasets report only a representative subset of the true activity, the fact that
significant amount of resources have to be invested by the terrorist group for every new incident
acts as a natural dampener toward more attacks.} terrorist groups. For example, the FARC dataset
considered by~\citet{vasanth_aoas2013} (also in this work) corresponds to $641$ incidents over
a ten-year period leading to an average of approx.\ $1.23$ incidents per week. Similar trends
can be seen across a %large family
number of terrorist groups,
see~\citep{lafree1,breiger,melamed,vasanth_compstudy2015} for examples. As a crude illustration,
%of this constraint,
learning a $4$ parameter model with $100$ observation points (on average)
per parameter leads to a model learning latency of $\frac{ 4 \times 100}{1.23}
\approx 325$ weeks or $\approx 6 \frac{1}{4}$ years.

A closely related and more challenging problem is the fact that most models capture some
underlying attribute of the group dynamics, which in itself can change dramatically over a
long time-period (such as that incurred in model learning). This fact renders assumptions of
model stability over such periods questionable. %To simplify
%matters, most work in terrorism modeling make stationarity assumptions (as is done
%here). Future work will consider a broader palette of time-varying models for
%terrorist activity.
The use of the proposed approaches over a long time-horizon (with time-varying parameter
estimates) opens up an array of issues on the stability of inferencing decisions in the
short time-horizon. %Thus, {\em online} approaches based on model learning
%%(however good the model-fit for the data is) is
%are inherently difficult to utilize in practice because of uncertainty in the applicability
%of the learned model (based on sparse data from the past) to the current.
The two approaches in Sec.~\ref{subsec_3a} and Sec.~\ref{subsec_3b} differ
as follows. Inferencing with $\{ {\mb M}_i \}$ is primarily determined by hard
decisions ($\{ \widehat{\Sb}_i \}$) that could potentially lead to information loss,
unlike inferencing with soft metrics such as $\{ (X_n, Y_n) \}$.
%On the other hand, inferencing with the latter approach is based on soft metrics such
%as $\{ (X_n, Y_n) \}$.
Nevertheless, both approaches suffer from a common problem which
make it unattractive as an {\em online} approach: a {\em retrospective (non-causal)}
state classification process after model learning.

\section{Non-Parametric Approach for Spurt Detection}
\label{sec4}
The above difficulties motivate the
development of causal (non-parametric) approaches which is the focus of this section.
Non-parametric changepoint procedures based on signs or signed rank statistics of
observations with median or Wilcoxon scores have been studied for a long time (see
the book by~\citet{subha} for a survey). While the utility of such procedures in a
%complex
terrorist network setting %such as a terrorist network
has not been addressed before, we
follow along different lines in this work to develop a non-parametric approach that
exploits the HMM structure and still competes well with its parametric counterparts in
detecting spurts and downfalls. In addition to spurt detection, the proposed approach
also identifies the source of disruption behind a spurt/downfall in the activity of a
terrorist group. In this approach, instead of using only the summary statistics ($X_n$ and
$Y_n$) of the vector ${\mb M}_i\Big|_{i \in \Delta_n}$ as in Sec.~\ref{subsec_3b}, we
consider the entire vector to study resilience and coordination signatures in the group.
To develop this approach, recall that resilience and coordination are captured by the
group's ability to perpetrate multiple attacks over successive days and the same day,
respectively. Thus, a metric that measures the degree of ``well-spreadness'' of attacks
over $\Delta_n$ (or its lack thereof) can be used as an indicator and measure of high
resilience (or coordination).

\subsection{Majorization Theory}
\label{subsec_4a}
With this backdrop, majorization theory provides a theoretical framework~\citep{olkin}
to compare two vectors on the basis of their ``well-spreadness.'' For the sake of
self-containment of this paper, a brief introduction to majorization theory,
Schur-convex and -concave functions, catalytic majorization, and equivalent conditions
for verifying a catalytic majorization relationship between two vectors are provided in
%Appendix~\ref{app_majtheory}.
Supplementary B. The main conclusion from Supplementary B is provided next.

Let ${\mathbb P}_{\delta}$ denote the space of probability vectors of length $\delta$
(where $\delta > 1$) with $\underline{\boldsymbol P} =
%\left[ {\sf M}_1, \cdots , {\sf M}_{\delta} \right]
\left[ {\mb P}(1), \cdots, {\mb P}(\delta) \right] \in %{\cal P}(\delta)
{\mathbb P}_{\delta}$ %\Longleftrightarrow
if and only if ${\mb P}(i) \geq 0$ for all $i = 1, \cdots, \delta$ and
$\sum_{i = 1}^{\delta} {\mb P}(i) = 1$. Without loss in generality, we can assume
that the entries of $\underline {\boldsymbol P}$ are arranged in non-increasing order
%(${\sf M}_1 \geq \cdots \geq {\sf M}_{\delta}$).
(that is, ${\mb P}(1) \geq \cdots \geq {\mb P}(\delta)$).
\begin{thm}
\label{main_theorem_main_text}
Let $\{ \underline {\boldsymbol P}, \hsppp \underline {\boldsymbol Q} \} \in %{\cal P}(\delta)
{\mathbb P}_{\delta}$.
In one of two possibilities, $\underline {\boldsymbol P}$ and $\underline {\boldsymbol Q}$ are
not comparable with each other in the form of a catalytic majorization relationship.
In the %complementary
other possibility, their comparability is verified by checking an equivalent
%(simultaneous)
set of conditions over only two types of functions:
\begin{eqnarray}
{\rm i)} & & {\sf PM}( \underline{\boldsymbol P}, \hsppp \alpha) <
{\sf PM}( \underline{\boldsymbol Q}, \hsppp \alpha) \hspp \hspp
{\rm if} \hspp \alpha > 1 %\hspp {\rm or} \hspp \alpha < 0
, \nonumber \\
{\rm ii)} & & {\sf PM}( \underline{\boldsymbol P}, \hsppp \alpha) >
{\sf PM}( \underline{\boldsymbol Q}, \hsppp \alpha) \hspp \hspp
{\rm if} \hspp \alpha < 1, {\hspace{0.05in}} {\sf and}
\nonumber \\
{\rm iii)} & & {\sf SE}( \underline{\boldsymbol P}) >
{\sf SE}( \underline{\boldsymbol Q}). \nonumber
%\\ {\rm iv)} & & {\sf GM}( \underline{\boldsymbol M}) >
%{\sf GM}( \underline{\boldsymbol N}). \nonumber
\end{eqnarray}
In the above equations, ${\sf SE}(\cdot)$ and ${\sf PM}(\cdot, \hsppp \alpha)$ stand
for the Shannon entropy %\footnote{All logarithms in this paper are to base $e$.}
%unless specified otherwise.
function and the power mean function corresponding to an
index $\alpha$, and are defined as,
\begin{eqnarray}
{\sf SE}( \underline{\boldsymbol P} ) %& \triangleq &
\triangleq
- \sum_{i = 1}^{\delta}
{\mb P}(i) \log \left( {\mb P}(i) \right) ,
%\nonumber \\
& &
{\sf PM}( \underline{\boldsymbol P}, \hsppp \alpha) \triangleq %& \triangleq &
\left( %\frac{
\frac{ \sum _{i = 1}^{\delta} {\mb P}(i)^{\alpha} }{\sum_{i=1}^{\delta}
\indic( {\mb P}(i) > 0)}
\right)^{1/\alpha}. \nonumber
\end{eqnarray}
\qed
\end{thm}

\subsection{Computational Reduction by a Single Function Search}
\label{subsec_4b}
While Theorem~\ref{main_theorem_main_text} establishes the importance of evaluating
the power mean function over the continuous parameter $\alpha$ in comparing two different
%attack frequency
vectors, for computational reasons, we propose the search over a single function as proxy.
%for  comparison purposes.
This single function is the normalized power mean corresponding to a fixed index
$\alpha^{\star} \geq 1$ (see the definition in Supplementary B as well as a motivation
for this functional form as a candidate), which is given as,
\begin{eqnarray}
{\sf NPM}( \underline {\boldsymbol P}, \hsppp \alpha^{\star}) %& \triangleq &
%\frac{ {\sf PM}( \underline {\boldsymbol P}, \hsppp \alpha^{\star}) }
%{ {\sf NZ}( \underline {\boldsymbol P} ) }
%%\cdot \indic \left( \sum \nolimits_{i \in \Delta_n} M_i > 0 \right)
%%\nonumber \\  & = &
= \frac{  {\sf PM}( \underline {\boldsymbol P}, \hsppp \alpha^{\star}) }
{ \sum_{i = 1}^{\delta} \indic \left( {\boldsymbol P}(i) > 0 \right) }
= \frac{ \left( \sum _{i = 1}^{\delta} {\boldsymbol P}(i) ^{\alpha^{\star}} \right)^{1/\alpha^{\star}} }
{ \left( \sum_{i = 1}^{\delta} \indic \left( {\boldsymbol P}(i) > 0 \right) \right)^{1 +
1/\alpha^{\star}  }}.
%\cdot\indic \left( \sum \nolimits_{i \in \Delta_n} M_i > 0 \right)
\nonumber
%\\ & = &
\end{eqnarray}

To explain the reason for this specific choice,
%of ${\sf NPM}(  \underline {\boldsymbol M}, \hsppp \alpha^{\star})$ is as follows. For this,
we define $\alpha_{\sf max}({\sf PM})$ and $\alpha_{\sf max}({\sf NPM})$ corresponding to
$\underline {\boldsymbol P}$ and $\underline {\boldsymbol Q}$
(suitably permuted) as follows:
\begin{eqnarray}
\alpha_{\sf max}({\sf PM}) & \triangleq & \arg\sup_{ \alpha \hsppp \in \hsppp [1, \hsppp \infty)  }
\Big\{
%\hsppp : \hsppp
%\frac{
{\sf PM}( \underline {\boldsymbol P}, \hsppp \alpha )
%} {
> {\sf PM}( \underline {\boldsymbol Q}, \hsppp \alpha) %} > 1
\nonumber \\
& & {\hspace{1.2in}}
\hspp {\sf and} \hspp
%\frac{
{\sf PM}( \underline {\boldsymbol P}, \hsppp \widetilde{\alpha} )
%}{
\leq
{\sf PM}( \underline {\boldsymbol Q}, \hsppp \widetilde{\alpha} ) %} \leq 1
\hspp {\sf for} \hspp {\sf all} \hspp \widetilde{\alpha} > \alpha \Big\}
\nonumber \\
%%%%
\alpha_{\sf max}({\sf NPM}) & \triangleq &
\arg\sup_{ \alpha \hsppp \in \hsppp [1, \hsppp \infty)  }
\Big\{
%\frac{
{\sf NPM}( \underline {\boldsymbol P}, \hsppp \alpha ) %} {
> {\sf NPM}( \underline {\boldsymbol Q}, \hsppp \alpha) %} > 1
\nonumber \\
&& {\hspace{1.2in}}
\hspp {\sf and} \hspp
%\frac{
{\sf NPM}( \underline {\boldsymbol P}, \hsppp \widetilde{\alpha} )
%} {
\leq {\sf NPM}( \underline {\boldsymbol Q}, \hsppp \widetilde{\alpha} ) %} \leq 1
\hspp {\sf for} \hspp {\sf all} \hspp \widetilde{\alpha} > \alpha \Big\}.
\nonumber
\end{eqnarray}
In other words, $\alpha_{\sf max}(\cdot)$ is the largest choice (supremum) of
$\alpha$ at which the inequality relationship desired in Theorem~\ref{main_theorem_main_text}
fails to hold (for the corresponding function) with $\underline {\boldsymbol P}$ and
$\underline {\boldsymbol Q}$ as inputs.
%\end{eqnarray}
Similarly, we define $\alpha_{\sf min}({\sf PM})$ and $\alpha_{\sf min}({\sf NPM})$,
corresponding to the smallest choice (infimum) of $\alpha$ at which the inequality
relationship of Theorem~\ref{main_theorem_main_text} fails to hold, as:
\begin{eqnarray}
\alpha_{\sf min}({\sf PM}) & \triangleq &
\arg\inf_{ \alpha \hsppp \in \hsppp (-\infty, \hsppp 1]  }
\Big\{
%\hsppp : \hsppp
%\frac{
{\sf PM}( \underline {\boldsymbol P}, \hsppp \alpha )
%} {
< {\sf PM}( \underline {\boldsymbol Q}, \hsppp \alpha) %} > 1
\nonumber \\
& & {\hspace{1.2in}}
\hspp {\sf and} \hspp
%\frac{
{\sf PM}( \underline {\boldsymbol P}, \hsppp \widetilde{\alpha} )
%}{
\geq
{\sf PM}( \underline {\boldsymbol Q}, \hsppp \widetilde{\alpha} ) %} \leq 1
\hspp {\sf for} \hspp {\sf all} \hspp \widetilde{\alpha} < \alpha \Big\}
\nonumber \\
%%%%
\alpha_{\sf min}({\sf NPM}) & \triangleq &
\arg\inf_{ \alpha \hsppp \in \hsppp (-\infty, \hsppp 1]  }
\Big\{
%\frac{
{\sf NPM}( \underline {\boldsymbol P}, \hsppp \alpha ) %} {
< {\sf NPM}( \underline {\boldsymbol Q}, \hsppp \alpha) %} > 1
\nonumber \\
&& {\hspace{1.2in}}
\hspp {\sf and} \hspp
%\frac{
{\sf NPM}( \underline {\boldsymbol P}, \hsppp \widetilde{\alpha} )
%} {
\geq {\sf NPM}( \underline {\boldsymbol Q}, \hsppp \widetilde{\alpha} ) %} \leq 1
\hspp {\sf for} \hspp {\sf all} \hspp \widetilde{\alpha} < \alpha \Big\}.
\nonumber
\end{eqnarray}

Note that $\alpha_{\sf max}({\sf PM})$ %and $\alpha_{\sf max}({\sf NPM})$ are
is well-defined since
\begin{eqnarray}
\lim_{\alpha \hsppp \rightarrow \hsppp \infty}
{\sf PM}( \underline {\boldsymbol P}, \hsppp \alpha )
 - {\sf PM}( \underline {\boldsymbol Q}, \hsppp \alpha)
= {\mb P}(1) - {\mb Q}(1)
\nonumber
\end{eqnarray}
and this quantity can be ensured to be upper bounded by $0$ after an appropriate
permutation of $\underline {\boldsymbol P}$ and $\underline {\boldsymbol Q}$. Similarly,
$\alpha_{\sf max}({\sf NPM})$ is well-defined since
\begin{eqnarray}
\lim_{\alpha \hsppp \rightarrow \hsppp \infty}
{\sf NPM}( \underline {\boldsymbol P}, \hsppp \alpha )
- {\sf NPM}( \underline {\boldsymbol Q}, \hsppp \alpha) =
\frac{ {\mb P}(1) } { \sum_{i = 1}^{\delta} \indic \left( {\boldsymbol P}(i) > 0 \right) }
-
\frac{ {\mb Q}(1) } { \sum_{i = 1}^{\delta} \indic \left( {\boldsymbol Q}(i) > 0 \right) } .
\nonumber
\end{eqnarray}
As before, the quantity above can also be ensured to be upper bounded by $0$ after
an appropriate permutation (but not perhaps the same permutation as in the previous case).
Similarly, $\alpha_{\sf min}({\sf PM})$ and $\alpha_{\sf min}({\sf NPM})$ are also well-defined
since we have
\begin{eqnarray}
\lim_{\alpha \hsppp \rightarrow \hsppp -\infty}
{\sf PM}( \underline {\boldsymbol P}, \hsppp \alpha )
 - {\sf PM}( \underline {\boldsymbol Q}, \hsppp \alpha)
& = & {\mb P}(\delta) - {\mb Q}(\delta)
\nonumber \\
\lim_{\alpha \hsppp \rightarrow \hsppp -\infty}
{\sf NPM}( \underline {\boldsymbol P}, \hsppp \alpha )
- {\sf NPM}( \underline {\boldsymbol Q}, \hsppp \alpha) & = &
\frac{ {\mb P}(\delta) } { \sum_{i = 1}^{\delta} \indic \left( {\boldsymbol P}(i) > 0 \right) }
-
\frac{ {\mb Q}(\delta) } { \sum_{i = 1}^{\delta} \indic \left( {\boldsymbol Q}(i) > 0 \right) } ,
\nonumber
\end{eqnarray}
both of which can be lower bounded by $0$ with appropriate permutations of $\underline {\boldsymbol P}$
and $\underline {\boldsymbol Q}$.

\begin{figure}[tbh!]
\begin{center}
\begin{tabular}{cc}
%\begin{minipage}{5.5in}
%\centerline{
\includegraphics[height=2.45in,width=2.7in] {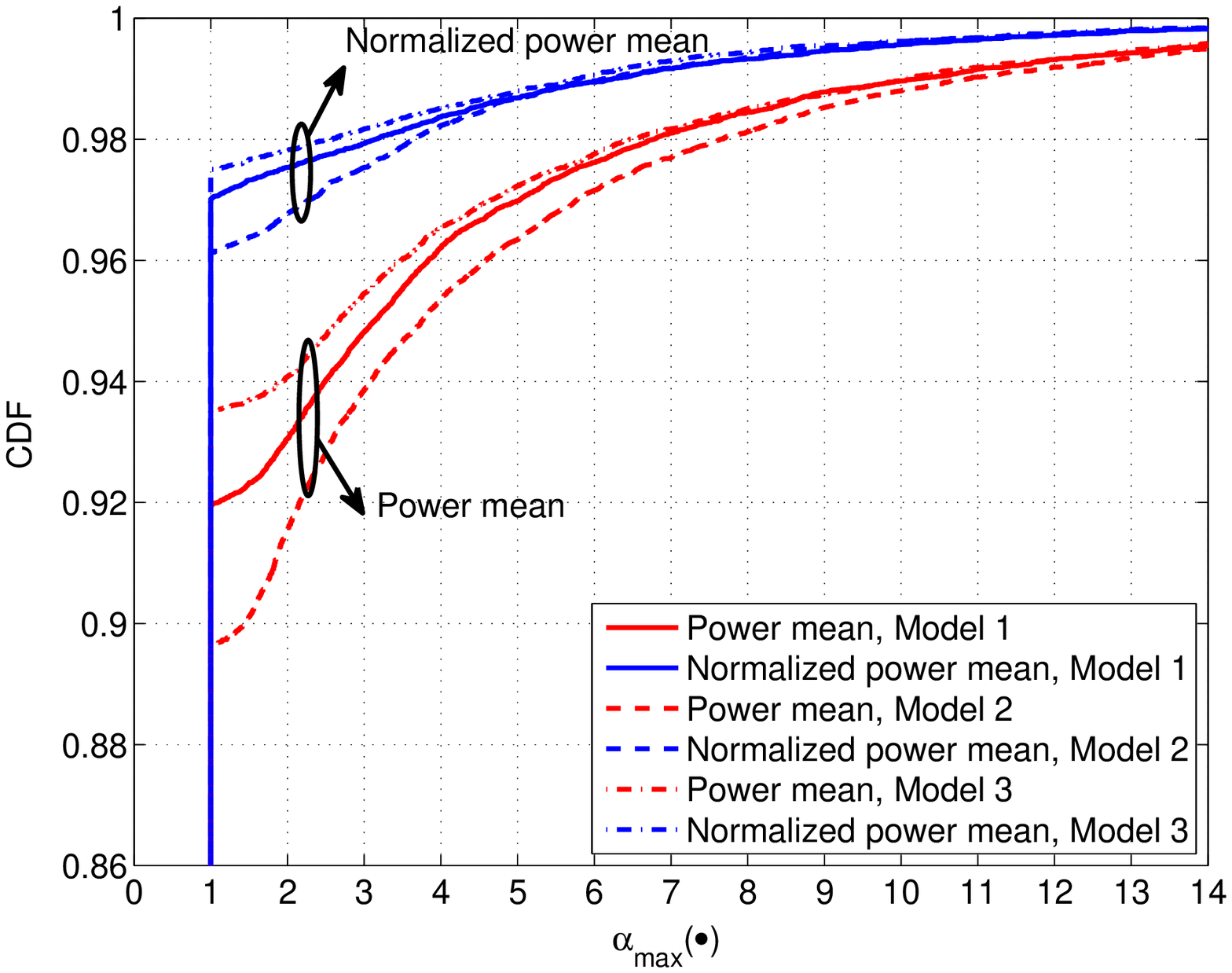}
%{figs/fig_catalytic_maj_theory_index_finder2.eps}
&
\includegraphics[height=2.45in,width=2.7in] {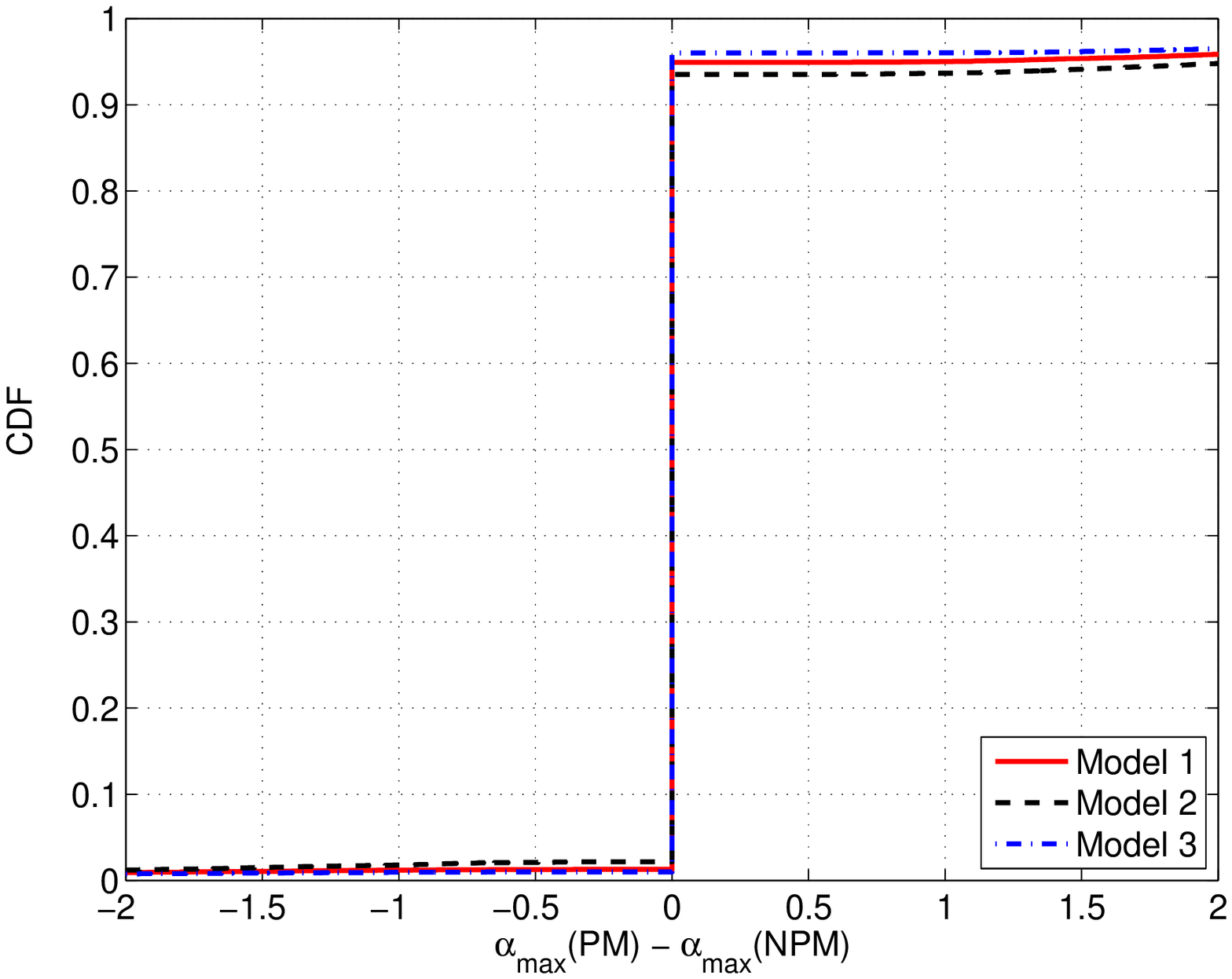}
%{figs/fig_catalytic_maj_theory_index_diff_finder.eps}
\\
(a) & (b)
\\
\includegraphics[height=2.45in,width=2.7in] {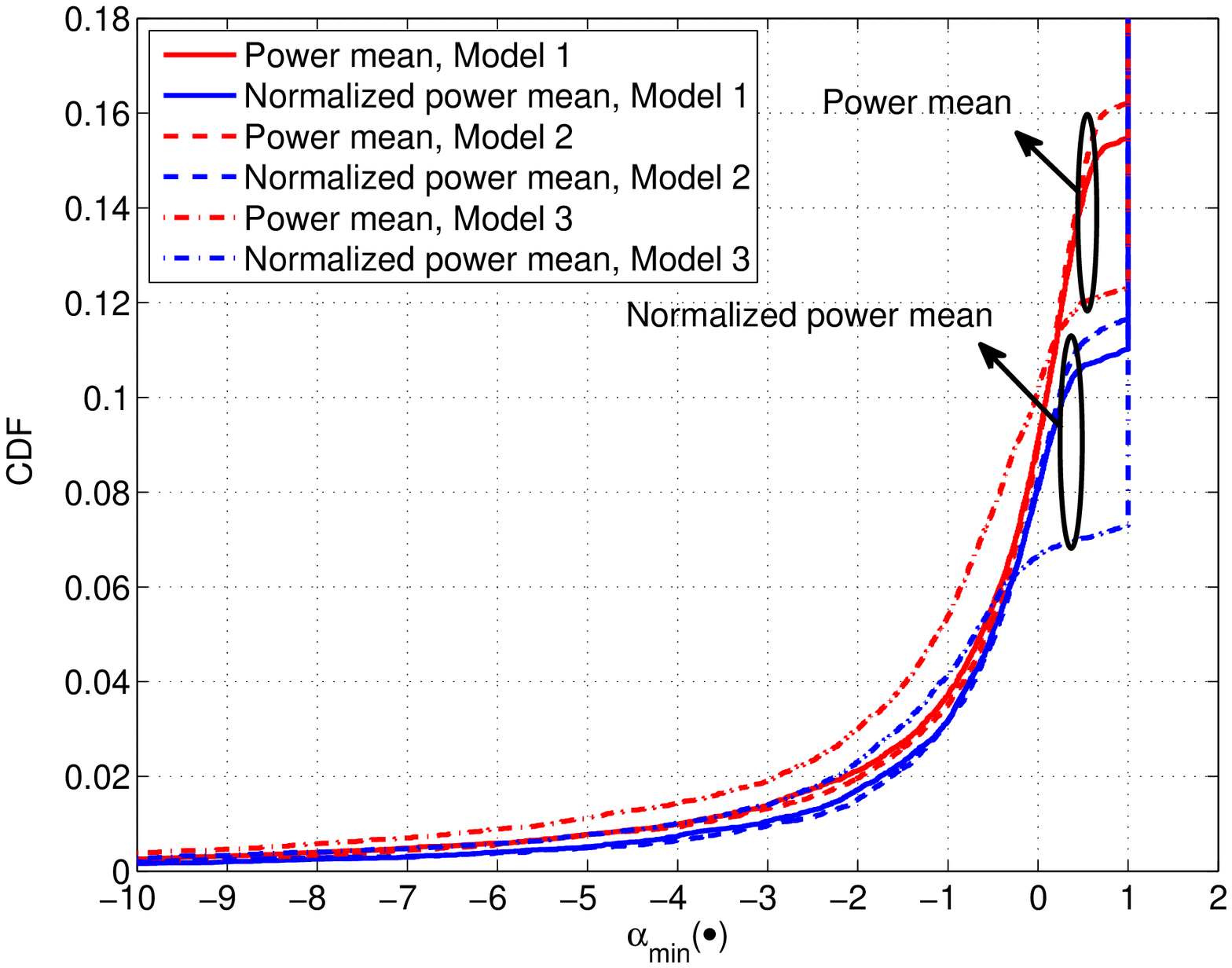}
%{figs/fig_catalytic_maj_theory_index_neg_finder2.eps}
&
\includegraphics[height=2.45in,width=2.7in] {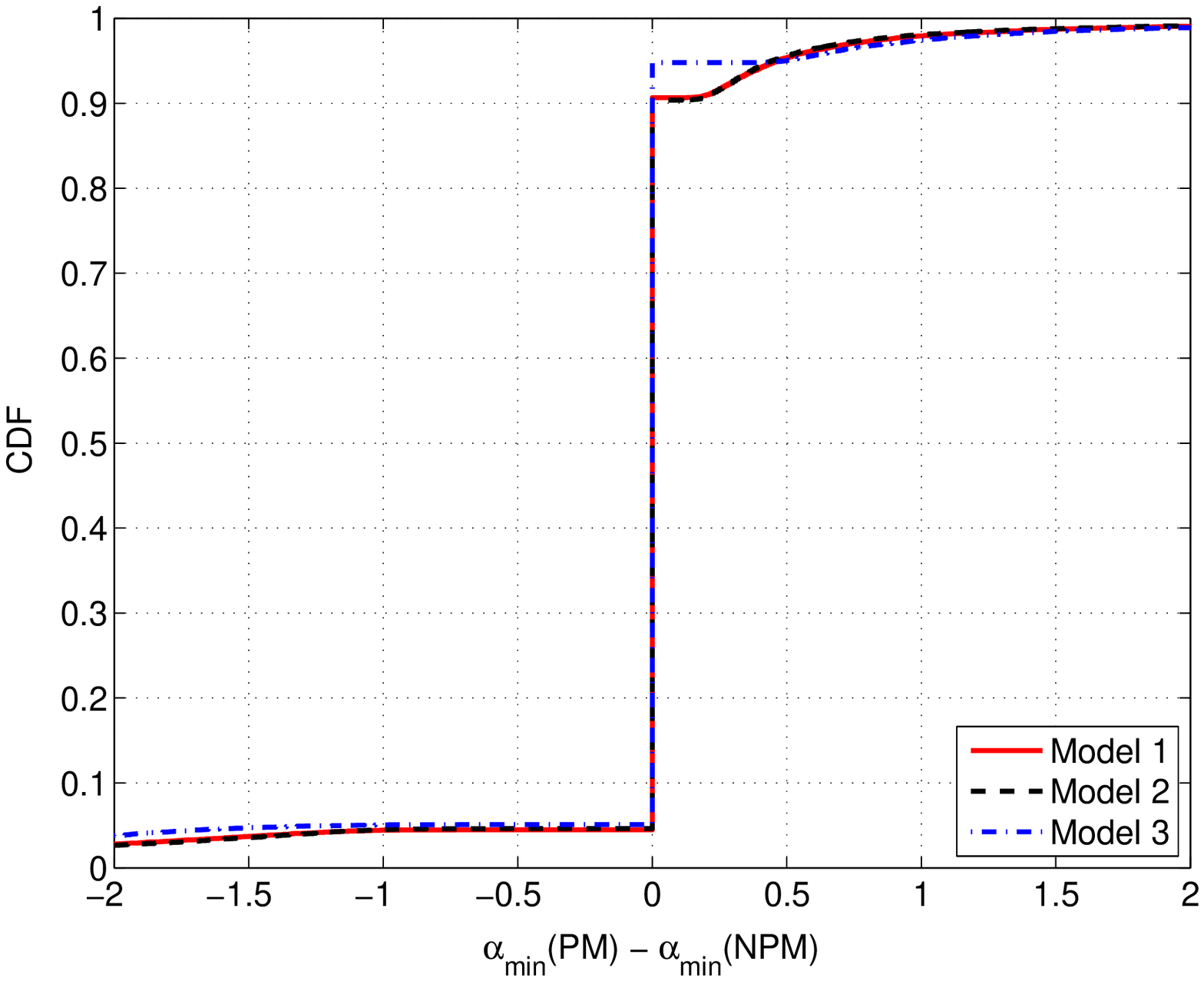}
%{figs/fig_catalytic_maj_theory_index_neg_diff_finder.eps}
\\
(c) & (d)
%}
%\end{minipage}
\end{tabular}
\caption{\label{fig_alpha_choose}
%(a)
CDF of: (a) $\alpha_{\sf max}({\sf PM})$ and $\alpha_{\sf max}({\sf NPM})$,
(b) $\alpha_{\sf max}({\sf PM}) - \alpha_{\sf max}({\sf NPM})$,
(c) $\alpha_{\sf min}({\sf PM})$ and $\alpha_{\sf min}({\sf NPM})$, and
(d) $\alpha_{\sf min}({\sf PM}) - \alpha_{\sf min}({\sf NPM})$ with
$\delta = 7$ and $K = 10$ for three random models generating
$\underline{\boldsymbol P}$ and $\underline{\boldsymbol Q}$.}
%in (a) and (b). }
\end{center}
\end{figure}

Specifically, if $\underline {\boldsymbol P} \prec \underline {\boldsymbol Q}$,
%and $\underline {\boldsymbol P} \neq \underline {\boldsymbol Q}$,
(from Prop.\ 1 and Corollary 1 of Supplementary B) we have
$\{ \alpha_{\sf max}({\sf PM}), \hsppp  \alpha_{\sf max}({\sf NPM}), \hsppp
\alpha_{\sf min}({\sf PM}) \} = 1$. In general, if $\underline {\boldsymbol P} \nprec
\underline {\boldsymbol Q}$, we can have the following possibilities:
$\{ \alpha_{\sf max}({\sf PM}), \hsppp \alpha_{\sf max}({\sf NPM}) \} > 1$ and
$\{ \alpha_{\sf min}({\sf PM}), \hsppp \alpha_{\sf min}({\sf NPM}) \} < 1$.
To understand the precise behavior of these quantities and the relative frequency of
the event where $\{ \alpha_{\sf max}({\sf PM}), \hsppp \alpha_{\sf max}({\sf NPM}) \} > 1$
and $\{ \alpha_{\sf min}({\sf PM}), \hsppp \alpha_{\sf min}({\sf NPM}) \} < 1$,
we now study the cumulative distribution function (CDF) %and the density function
of $\alpha_{\sf max}(\cdot)$ and $\alpha_{\sf min}(\cdot)$
%for %different choices of $\delta$:
 %, 10$ and $14$ days
with three random models for generating $\underline {\boldsymbol P}$ and
$\underline {\boldsymbol Q}$. In all the models, we have $\underline {\boldsymbol P} \in
{\mathbb P}_{\delta_1}$ and $\underline {\boldsymbol Q} \in {\mathbb P}_{\delta_2}$ where
$\delta_1$ and $\delta_2$ are (discrete) uniformly distributed as $\{ \delta_1, \hsppp \delta_2 \}
\sim {\cal U}( [1, \hsppp \delta] )$.
%{\cal U}( \{ 1, 2, \cdots, \delta \})$.
The non-zero entries of $\underline {\boldsymbol P}$ and $\underline {\boldsymbol Q}$
(which are yet to be permuted in non-increasing order) are generated as
\begin{eqnarray}
{\mb P}(i) = \frac{p_i}{ \sum_{j = 1}^{\delta_1} p_j },
\hspp i = 1, \cdots, \delta_1, & &
{\mb Q}(i) = \frac{q_i}{ \sum_{j = 1}^{\delta_2} q_j },
\hspp i = 1, \cdots, \delta_2
\nonumber
\end{eqnarray}
where $\{ p_i, q_i \}$ are from one of the three models below (continuous uniform,
%magnitude of a standard Gaussian
folded normal and discrete uniform):
\begin{eqnarray}
{\sf Model \hspp 1} & : & \{ p_i , \hsppp q_i \} \sim {\cal U}([0,1]) \nonumber \\
{\sf Model \hspp 2} & : & \{ p_i , \hsppp q_i \} \sim |{\cal N}| \nonumber \\
{\sf Model \hspp 3} & : & \{ p_i , \hsppp q_i \} \sim %{\cal U}( \{1, 2, \cdots, K \}  ).
{\cal U}( [1, \hsppp K]).
\nonumber
\end{eqnarray}

Fig.~\ref{fig_alpha_choose}(a) plots the CDF of $\alpha_{\sf max}({\sf PM})$ and
$\alpha_{\sf max}({\sf NPM})$ for the three models with $\delta = 7$ days and $K = 10$.
Clearly, $\alpha_{\sf max}({\sf PM})$ and $\alpha_{\sf max}({\sf NPM})$ are $1$ for a
large fraction (over $90 \%$) of the realizations of $\{ \underline {\boldsymbol P},
\hsppp \underline {\boldsymbol Q} \}$ suggesting\footnote{This is not conclusively established
since Conditions ii) and iii) of Theorem~\ref{main_theorem_main_text} have not been verified.}
that these realizations can be compared in a catalytic majorization relationship. We also
observe that, for a fixed $\alpha$, the normalized power mean can be used to discern
$\underline {\boldsymbol P}$ and $\underline {\boldsymbol Q}$ better in the sense that this
metric allows a comparative relationship for a larger fraction of the realizations than the
power mean. Fig.~\ref{fig_alpha_choose}(b)
plots the CDF of $\alpha_{\sf max}({\sf PM}) - \alpha_{\sf max}({\sf NPM})$ for the three
models which also shows that while $\alpha_{\sf max}({\sf PM}) = 1 =
\alpha_{\sf max}({\sf NPM})$ in a large fraction (over $95 \%$) of the cases, the
normalized power mean is more effective than the power mean reflected by the higher
probability for the event where $\alpha_{\sf max}({\sf PM})$ is larger than
$\alpha_{\sf max}({\sf NPM})$ than for the event where $\alpha_{\sf max}({\sf PM})$ is
smaller than $\alpha_{\sf max}({\sf NPM})$. Similar trends can be seen in
Figs.~\ref{fig_alpha_choose}(c)-(d) for $\{ \alpha_{\sf min}({\sf PM}), \hsppp
\alpha_{\sf min}({\sf NPM}) \}$ and $\alpha_{\sf min}({\sf PM}) -
\alpha_{\sf min}({\sf NPM})$.

While Fig.~\ref{fig_alpha_choose} does not conclusively establish the efficacy of the
normalized power mean (relative to the power mean) for a single function search as proxy,
a further study illustrated in Table~\ref{table1_efficacy} provides this evidence. In
this table, we list the conditional probability ${\sf P}_{\sf PM}(\alpha^{\star})$ (the
definition for ${\sf P}_{\sf NPM}(\alpha^{\star})$ follows analogously) that all the
inequality relationships in Theorem~\ref{main_theorem_main_text} are satisfied provided
the corresponding inequality relationship with the proxy function is satisfied:
\begin{eqnarray}
{\sf P}_{\sf PM}(\alpha^{\star}) & = & \left\{
\begin{array}{c}
{\sf P} \Big(
{\sf PM}( \underline {\boldsymbol P}, \hsppp \alpha ) \leq
{\sf PM}( \underline {\boldsymbol Q}, \hsppp \alpha )
\hspp {\sf for} \hspp {\sf all} \hspp \alpha \geq 1
\hspp {\sf and} \hspp
{\sf PM}( \underline {\boldsymbol P}, \hsppp \alpha ) \geq
{\sf PM}( \underline {\boldsymbol Q}, \hsppp \alpha )
%\nonumber \\  & &
\\
{\hspace{0.7in}}  \left.
\hspp {\sf for} \hspp {\sf all} \hspp \alpha \leq 1
\Big| {\sf PM}( \underline {\boldsymbol P}, \hsppp \alpha^{\star} ) \leq
{\sf PM}( \underline {\boldsymbol Q}, \hsppp \alpha^{\star} )
\right) \hspp
{\sf if} \hspp \alpha^{\star} \geq 1
\\
%%%
%%%
{\sf P} \Big(
{\sf PM}( \underline {\boldsymbol P}, \hsppp \alpha ) \leq
{\sf PM}( \underline {\boldsymbol Q}, \hsppp \alpha )
\hspp {\sf for} \hspp {\sf all} \hspp \alpha \geq 1
\hspp {\sf and} \hspp
{\sf PM}( \underline {\boldsymbol P}, \hsppp \alpha ) \geq
{\sf PM}( \underline {\boldsymbol Q}, \hsppp \alpha )
%\nonumber \\  & &
\\
{\hspace{0.7in}}  \left.
\hspp {\sf for} \hspp {\sf all} \hspp \alpha \leq 1
\Big| {\sf PM}( \underline {\boldsymbol P}, \hsppp \alpha^{\star} ) \geq
{\sf PM}( \underline {\boldsymbol Q}, \hsppp \alpha^{\star} )
\right) \hspp
{\sf if} \hspp \alpha^{\star} \leq 1.
\end{array}
\right.
\nonumber
\end{eqnarray}
Recall that the standard error for a probability estimate $\widehat{\sf p}$ with
${\sf n}$ samples used to estimate this probability is given as
$\sqrt{ \frac{ \widehat{\sf p} \cdot ( 1 - \widehat{\sf p}) } { {\sf n} } }$.
While Table~\ref{table1_efficacy} corresponds to $\delta = 7$, $K = 10$ and $K = 15$,
similar patterns are seen with $\delta = 10$ and $\delta = 14$ as well (data not
provided here for space reasons). From
Table~\ref{table1_efficacy}, we note that the efficacy of a single function search with
a larger $\alpha^{\star}$ is in general better than a check with a smaller $\alpha^{\star}$.
Further, the normalized power mean discerns the comparability between two vectors more
effectively than the power mean with $\alpha^{\star} = 1$ and $\alpha^{\star} = 2$ allowing
the comparison of over $90 \%$ of the realizations of $\{ \underline {\boldsymbol P},
\hsppp \underline {\boldsymbol Q} \}$ suggesting their utility here. %in the context of this work.
We use the choice $\alpha^{\star} = 2$ %($\alpha^{\star} = 1$ not used for space reasons)
in the rest of the sequel for illustrative purposes.
%While $\alpha^{\star} = 1$ can also be used, it is not illustrated here for space
%reasons.

%While pre-change distributional
%knowledge can be assumed if the in-control state lasts sufficiently long to learn the model
%parameters, post-change distributional knowledge is complicated by the need for its quick
%detection, which conflicts with a simultaneous model learning latency. More fundamentally,

\begin{table*}[htb!]
\caption{Conditional probability with their standard errors in parentheses
capturing efficacy of a single function search. Largest conditional probability
values are highlighted in bold-face.}
\label{table1_efficacy}
\begin{tabular}{ccccc||cccc}
%\hline
%\multicolumn{9}{c}{ $\delta = 7$ days } \\
\hline
& \multicolumn{4}{c||}{ ${\sf P}_{\sf PM}(\alpha^{\star})$ } &
\multicolumn{4}{c}{ ${\sf P}_{\sf NPM}(\alpha^{\star})$ }  \\
\hline
$\alpha^{\star}$ & $-1$ & $0$ & $1$ & $2$
& $-1$ & $0$ & $1$ & $2$
\\
\hline
{\sf Model 1} &
$0.7016$  & $0.7412$ &  $0.8797$   &  $0.8895$ &
$0.7939$ & $0.8364$  &  $0.9498$ &  ${\bf 0.9527}$
\\
& ($0.0030$) & ($0.0029$) & ($0.0024$) & ($0.0023$) &
($0.0026$) & ($0.0025$) & ($0.0015$)  & ($0.0015$)
\\
\hline
{\sf Model 2} &
$0.6689$ & $0.7075$  &  $0.8461$  &  $0.8606$ &
$0.7753$  & $0.8205$  &  $0.9359$  &  ${\bf 0.9360}$
\\
& ($0.0031$) & ($0.0030$) & ($0.0026$) & ($0.0026$) &
($0.0027$) & ($0.0026$) & ($0.0017$) & ($0.0017$)
\\
\hline
{\sf Model 3,} &
$0.7598$ & $0.7994$  & $0.9139$   &  $0.9002$  &
$0.8465$  & $0.8678$  & ${\bf 0.9720}$  &  $0.9563$
\\
$K = 10$
& ($0.0028$) & ($0.0027$) & ($0.0020$) & ($0.0021$) &
($0.0024$) & ($0.0022$) & ($0.0012$) &  ($0.0014$)
\\
\hline
{\sf Model 3,} &
$0.7386$ & $0.7807$  & $0.9072$   &  $0.8939$  &
$0.8318$  & $0.8600$  & ${\bf 0.9711}$  &  $0.9543$
\\
$K = 15$
& ($0.0029$) & ($0.0028$) & ($0.0021$) & ($0.0022$) &
($0.0024$) & ($0.0023$) & ($0.0012$) &  ($0.0015$)
\\
\hline
%\hline
%%%
%%%
%%%
\ignore{
\multicolumn{9}{c}{ $\delta = 10$ days } \\
\hline
& \multicolumn{4}{c||}{ ${\sf P}_{\sf PM}(\alpha^{\star})$ } &
\multicolumn{4}{c}{ ${\sf P}_{\sf NPM}(\alpha^{\star})$ }  \\
\hline
$\alpha^{\star}$ & $-1$ & $0$ & $1$ & $2$
& $-1$ & $0$ & $1$ & $2$
\\
\hline
{\sf Model 1} &
$0.7100$  & $0.7602$ &  $0.8792$   &  $0.8953$ &
$0.8034$ & $0.8580$  &  $0.9457$ &  $0.9514$
\\
& ($0.0029$) & ($0.0029$) & ($0.0024$) & ($0.0022$) &
($0.0026$) & ($0.0023$) & ($0.0016$)  & ($0.0015$)
\\
\hline
{\sf Model 2} &
$0.6779$ & $0.7162$  &  $0.8487$  &  $0.8625$ &
$0.7821$  & $0.8245$  &  $0.9362$  &  $0.9368$
\\
& ($0.0035$) & ($0.0035$) & ($0.0030$) & ($0.0029$) &
($0.0031$) & ($0.0029$) & ($0.0020$) & ($0.0020$)
\\
\hline
{\sf Model 3,} &
$0.7598$ & $0.7994$  & $0.9139$   &  $0.9002$  &
$0.8465$  & $0.8678$  & $0.9720$  &  $0.9563$
\\
$K = 10$
& ($0.0028$) & ($0.0027$) & ($0.0020$) & ($0.0021$) &
($0.0024$) & ($0.0022$) & ($0.0012$) &  ($0.0014$)
\\
\hline
{\sf Model 3,} &
$0.7386$ & $0.7807$  & $0.9072$   &  $0.8939$  &
$0.8318$  & $0.8600$  & $0.9711$  &  $0.9543$
\\
$K = 15$
& ($0.0029$) & ($0.0028$) & ($0.0021$) & ($0.0022$) &
($0.0024$) & ($0.0023$) & ($0.0012$) &  ($0.0015$)
\\
\hline
\hline
}
\end{tabular}
%\end{center}
\end{table*}

\subsection{Appropriate Function for Resilience and Coordination Monitoring}
\label{subsec_4c}
With Shannon entropy and normalized power mean corresponding to a fixed index
$\alpha^{\star}$ as candidate functions, we now illustrate the
importance of one function (over the other) depending on whether the goal is to
measure changes in resilience or coordination.

In the first case study, we consider the scenario where $K$ attacks are spread over
$\Delta_n$ in different ways: $K-k$ attacks on one day and $k$ days with one attack
on each day (where $k \in [0, \delta-1]$) %k_{\sf max}]$)
corresponding to an {\em attack frequency vector}
\begin{eqnarray}
\underline{\boldsymbol P}_k = \frac{1}{K} \cdot
\left[ %\frac{ K-k %}{K}
K-k , \hspp \underbrace{
%\frac{1}{K} ,  \hspp  \cdots, \hspp \frac{1}{K}
1, \hspp \cdots, \hspp 1}_{k \hspp {\sf times}}
, \hspp \underbrace{0 , \hspp \cdots, \hspp 0}_{ \delta - (k + 1) \hspp {\sf times}}
\right]  .
\nonumber
\end{eqnarray}
The attack frequency vector captures the distribution of frequency of attacks over
$\Delta_n$ and by definition, $\underline{\boldsymbol P}_k \in {\mathbb{P}}_{\delta}$
provided that there is at least one attack over $\Delta_n$.
With $k = 0$, the group is highly coordinating since all the attacks are clustered on
one day. With %$K = \delta$ and
$k = \delta - 1$, the group is %highly
resilient since the attacks are well-spread %equitously distributed
over $\Delta_n$. In general, as $k$
increases, the coordination in the group decreases and its resilience increases.

We consider two benchmark attack frequency vectors $\underline{\boldsymbol P}_{\sf c}$
and $\underline{\boldsymbol P}_{\sf r}$ to compare the efficacy of the two majorization
metrics (Shannon entropy and normalized power mean corresponding to an index
$\alpha^{\star}$) in measuring changes in the group dynamics with $\underline{\boldsymbol P}_k$:
\begin{eqnarray}
\underline{\boldsymbol P}_{\sf c} %& = &
= \Big[ 1, \hspp \underbrace{0, \hspp \cdots,
\hspp 0}_{\delta - 1 \hspp {\sf times}} \Big],
%\nonumber \\
\hspp \hspp \hspp \hspp \hspp
\underline{\boldsymbol P}_{\sf r} %& = &
= \Big[
\underbrace{ 1/\delta , \hspp \cdots, \hspp 1/\delta
}_{ \delta \hspp {\sf times}} \Big].
\nonumber
\end{eqnarray}
The majorization metrics considered for comparison are the differences in
Shannon entropy and normalized power mean with $\underline{\boldsymbol P}_k$ and
the benchmark vectors:
\begin{eqnarray}
\Delta{\sf SE}_{\sf r}(k) & \triangleq &
%\triangleq
\Big| {\sf SE}( \underline{\boldsymbol P}_{\sf r} )
- {\sf SE}( \underline{\boldsymbol P}_{k} ) \Big|
\nonumber \\
%, \hspp \hspp \hspp \hspp
\Delta{\sf SE}_{\sf c}(k) & \triangleq &
%\triangleq
\Big| {\sf SE}( \underline{\boldsymbol P}_{\sf c} )
- {\sf SE}( \underline{\boldsymbol P}_{k} ) \Big|
\nonumber \\
\Delta{\sf NPM}_{\sf r}(k) & \triangleq &
%\triangleq
\Big| {\sf NPM}( \underline{\boldsymbol P}_{\sf r}, \hsppp \alpha^{\star} )
- {\sf NPM}( \underline{\boldsymbol P}_{k} , \hsppp \alpha^{\star} ) \Big|
\nonumber \\
%, \hspp \hspp \hspp \hspp
\Delta{\sf NPM}_{\sf c}(k) & \triangleq &
%\triangleq
\Big| {\sf NPM}( \underline{\boldsymbol P}_{\sf c}, \hsppp \alpha^{\star} )
- {\sf NPM}( \underline{\boldsymbol P}_{k} , \hsppp \alpha^{\star} ) \Big|.
\nonumber
\end{eqnarray}
In Fig.~\ref{fig_majmetrics_select}(a), $\Delta{\sf SE}_{\sf r}(k)$ and
$\Delta{\sf NPM}_{\sf r}(k)$ are plotted with $\delta = 7$, $\alpha^{\star} = 2$,
and $K = 10, 15$ and $20$. From this study, we observe that the Shannon entropy
captures a change in a resilient group (large $k$ regime) better with a higher
$\Delta{\sf SE}$ value than $\Delta {\sf NPM}$. Similarly, in
Fig.~\ref{fig_majmetrics_select}(b), $\Delta{\sf SE}_{\sf c}(k)$ and
$\Delta{\sf NPM}_{\sf c}(k)$ are plotted and %we note that
the normalized power mean
captures a change in a coordinating group (small $k$ regime) better with a higher
$\Delta{\sf NPM}$ value than $\Delta {\sf SE}$.

We then consider a second case study where $k$ attacks are equally spread over $k$
days in the $\Delta_n$ time-period corresponding to an attack frequency vector
\begin{eqnarray}
\underline{\boldsymbol P}_k =
\left[
\underbrace{
\frac{1}{k} ,  \hspp  \cdots, \hspp \frac{1}{k} }_{k \hspp {\sf times}}
, \hspp \underbrace{0 , \hspp \cdots, \hspp 0}_{ \delta - k \hspp {\sf times}}
\right]  .
\nonumber
\end{eqnarray}
Clearly, ${\sf SE}( \underline{\boldsymbol P}_k ) = \log(k)$ and
${\sf NPM}( \underline{\boldsymbol P}_k , \alpha^{\star} ) =
\frac{1}{k^2}$ for any $\alpha^{\star}$. Further, it can be seen that
\begin{eqnarray}
\left| \frac{\partial}{ \partial k} {\sf SE} ( \underline{\boldsymbol P}_k )
 \right| = \frac{1}{k} & < & \frac{2}{k^3} =
\left| \frac{\partial}{ \partial k}
{\sf NPM} ( \underline{\boldsymbol P}_k, \alpha^{\star} ) \right|
\hspp {\sf for} \hspp {\sf small} \hspp k
\nonumber \\
\left| \frac{\partial}{ \partial k} {\sf SE} ( \underline{\boldsymbol P}_k ) \right|
= \frac{1}{k} & > & \frac{2}{k^3} =
\left| \frac{\partial}{ \partial k} {\sf NPM} ( \underline{\boldsymbol P}_k, \alpha^{\star} )
\right| \hspp {\sf for} \hspp {\sf large} \hspp k,
\nonumber
\end{eqnarray}
suggesting the relevance of Shannon entropy for capturing changes in resilience
(large $k$ regime) and normalized power mean %corresponding to an index $\alpha^{\star}$
for capturing changes in coordination (small $k$ regime). %These two studies motivate
%the use of Shannon entropy in tracking changes in resilience and the normalized power
%mean in tracking changes in coordination.

\subsection{Application to Spurt/Downfall Detection}
\label{subsec_4d}
We now apply the framework developed in Sec.~\ref{subsec_4a} to Sec.~\ref{subsec_4c} to
track changes in resilience and coordination in the group. For this, let
${\underline{\boldsymbol P}}_n =
\left[  {\boldsymbol P}_n(1), \hspp \cdots, \hspp {\boldsymbol P}_n(\delta) \right]$
%\Big|_{\Delta_n}$
%= \left[ {\sf M}_1, \cdots, {\sf M}_{\delta} \right]$
where
\begin{eqnarray}
{\boldsymbol P}_n(i) %\Big|_{\Delta_n}
= \left\{ \begin{array}{cl}
\frac{ {\boldsymbol M}_{ (n-1)\delta + i} }
{ \sum_{j \in \Delta_n} {\boldsymbol M}_j }
& {\sf if} \hspp \sum_{j \in \Delta_n} {\boldsymbol M}_j > 0,
\\
0 & {\sf otherwise}
\end{array}
\right.
\nonumber
\end{eqnarray}
denote the attack frequency vector over the time-window $\Delta_n$. In the
non-trivial setting of at least one attack over $\Delta_n$, the
Shannon entropy and normalized power mean reduce to
\begin{eqnarray}
{\sf SE} ( {\underline {\boldsymbol P}}_n ) & = &
\log \left( \sum_{i \in \Delta_n} {\boldsymbol M}_i \right) -
\frac{ \sum_{i \in \Delta_n} {\boldsymbol M}_i \log( {\boldsymbol M}_i) }
{ \sum_{i \in \Delta_n} {\boldsymbol M}_i }
\nonumber \\
{\sf NPM}( \underline {\boldsymbol P}_n, \hsppp \alpha^{\star})
& = &
\frac{ \left( \sum_{i \in \Delta_n } \left( {\boldsymbol M}_i
\right)^{\alpha^{\star} } \right)^{1/\alpha^{\star} } }
{ \left( \sum_{i \in \Delta_n} {\boldsymbol M}_i  \right) \cdot
\left( \sum_{i \in \Delta_n}  \indic \left( {\boldsymbol M}_i > 0 \right)
\right)^{1 + 1/\alpha^{\star} } },
%\cdot \indic \left( \sum \nolimits_{i \in \Delta_n} M_i > 0 \right)
\nonumber
\end{eqnarray}
respectively. In the trivial setting of no attacks over $\Delta_n$, we set
${\underline {\boldsymbol P}}_n = [0, \cdots, 0]$ and
${\sf SE} ( {\underline {\boldsymbol P}}_n ) = 0 =
{\sf NPM}( \underline {\boldsymbol P}_n, \hsppp \alpha^{\star})$. With these
metrics, resilience and coordination in the group over the
$n$th time-window are declared based on the satisfaction of the following conditions:
\begin{eqnarray}
{\sf Resilient} & \Longleftrightarrow &
{\sf SE} ( {\underline {\boldsymbol P}}_n ) > \underline{\sf SE}
\hspp \hsppp {\sf and} \hspp \hsppp
X_n > \widetilde{ \eta}_{\sf X}
%\nonumber
\label{Res_class_maj_theory}
\\
{\sf Coordinating} & \Longleftrightarrow &
{\sf NPM}( \underline {\boldsymbol P}_n, \hsppp \alpha^{\star}) > \underline{\sf NPM}
\hspp \hsppp {\sf and} \hspp \hsppp
Y_n > \widetilde{ \eta}_{\sf Y}
%\nonumber
\label{Coord_class_maj_theory}
\end{eqnarray}
corresponding to appropriate choices of $\widetilde{\eta}_{\sf X}$,
$\widetilde{\eta}_{\sf Y}$, $\underline{\sf SE}$ and $\underline{\sf NPM}$.

\begin{figure}[tbh!]
\begin{center}
\begin{tabular}{cc}
%\begin{minipage}{5.5in}
%\centerline{
\includegraphics[height=2.45in,width=2.7in] {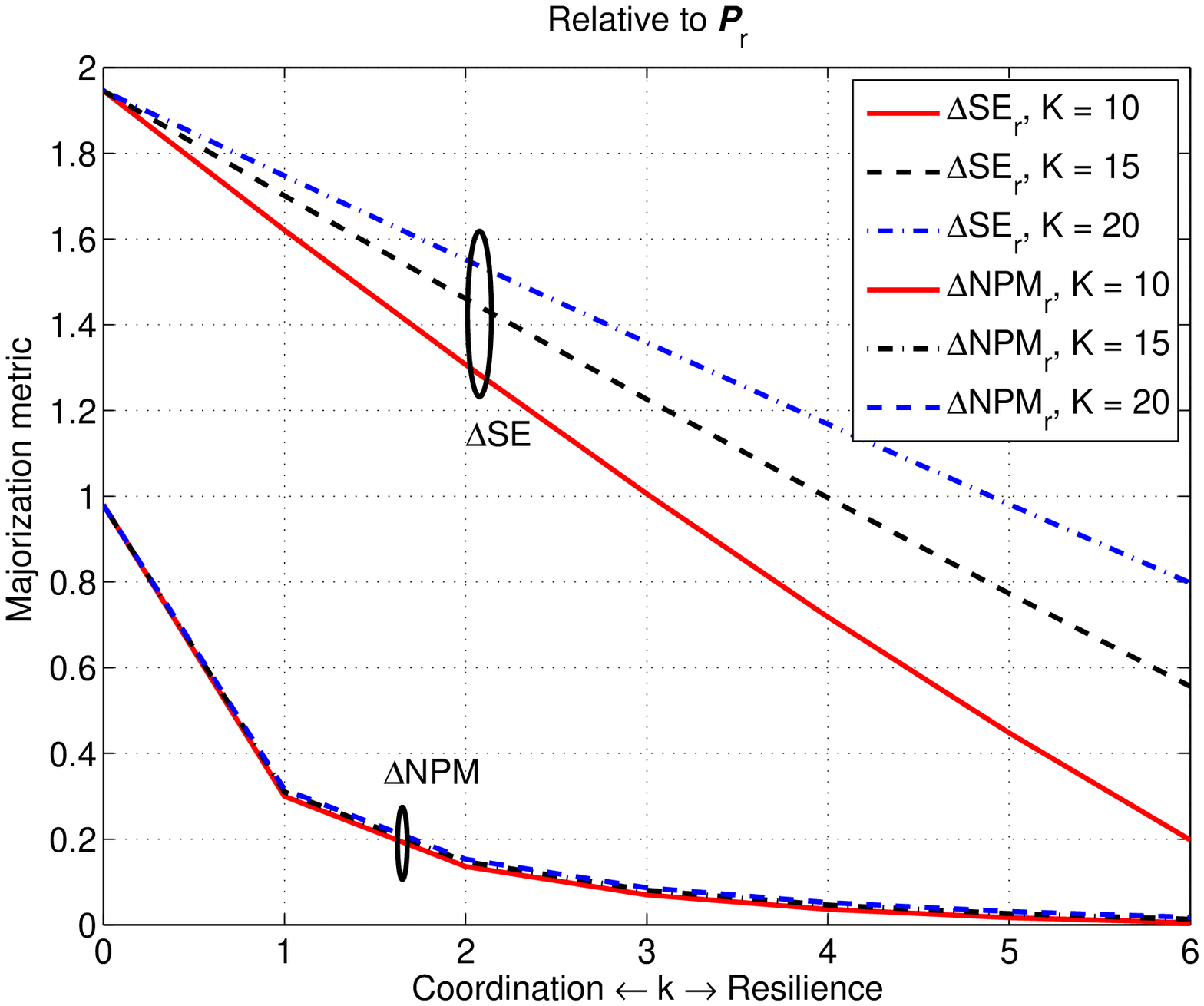}
%{figs/fig_select_majmetric_relativetores_study1.eps}
&
\includegraphics[height=2.45in,width=2.7in] {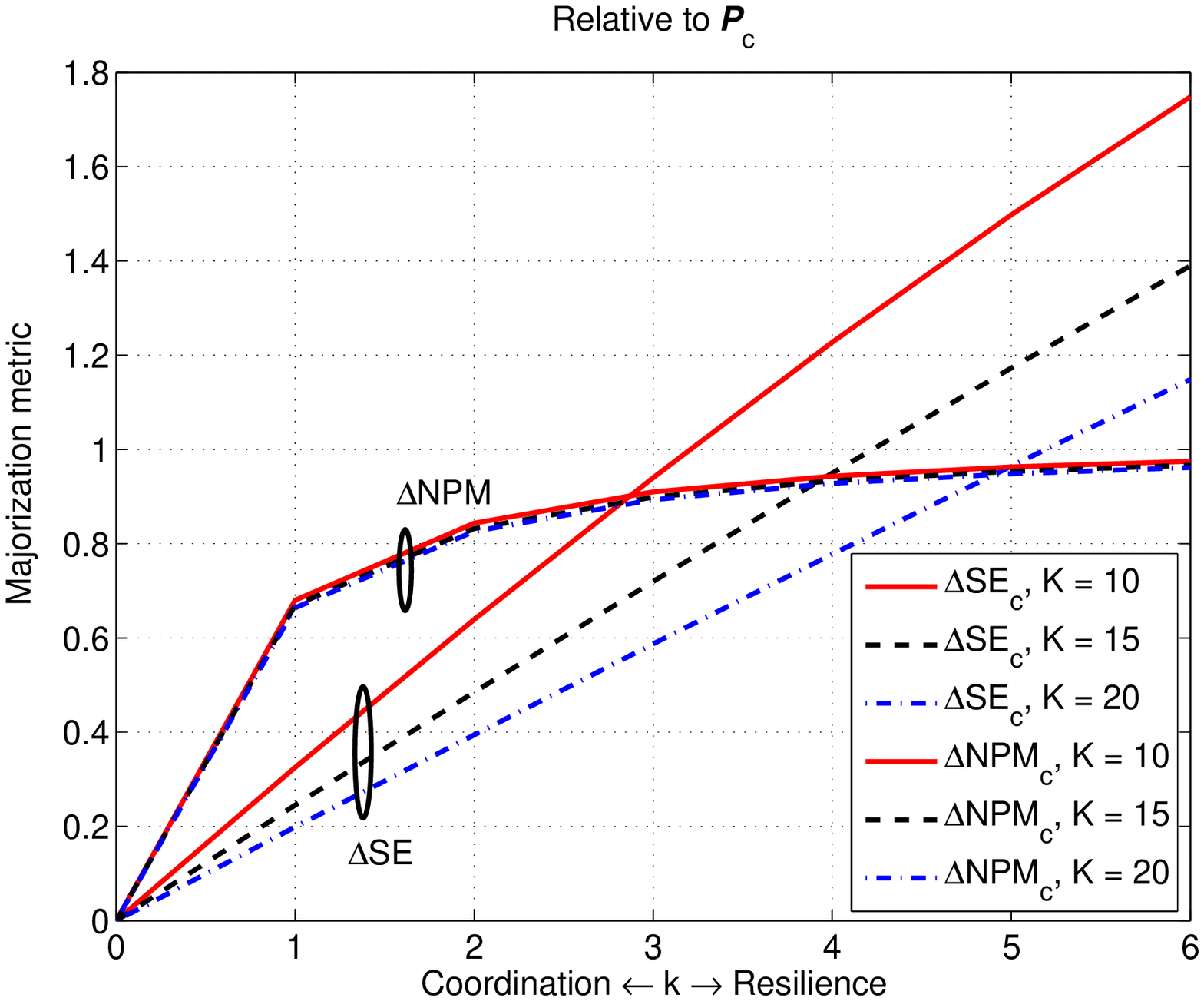}
%{figs/fig_select_majmetric_relativetocoord_study2.eps}
\\
(a) & (b)
%}
%\end{minipage}
\end{tabular}
\caption{\label{fig_majmetrics_select}
%(a)
Majorization metrics plotted relative to (a) $\underline{\boldsymbol P}_{\sf r}$
and (b) $\underline{\boldsymbol P}_{\sf c}$ as the attack frequency vector
$\underline{\boldsymbol P}_{k}$ transitions from more coordinating to more
resilient.}
\end{center}
\end{figure}

In addition to classification, we are interested in tracking the resilience and coordination
in the group over time. For this, we propose two tracking functions (${\sf Res}(n)$ and
${\sf Coord}(n)$, $n \geq 1$) that are updated as follows:
%\begin{eqnarray}
%{\sf Res}(n) & = & \sum_{n' = 1}^n \Big( {\sf SE} ( {\underline {\boldsymbol P}}_{n'} )
%+ X_{n'} \Big)
%- n \cdot \frac{ \sum_{n' = 1}^{N_{\sf max}}
%\Big( {\sf SE} ( {\underline {\boldsymbol P}}_{n'} ) + X_{n'} \Big) } {N_{\sf max}}
%\nonumber \\
%%
%{\sf Coord}(n) & = &
%\sum_{n' = 1}^n \Big( {\sf NPM}( \underline {\boldsymbol P}_{n'}, \hsppp \alpha^{\star})
%+ Y_{n'} \Big)
%- n \cdot \frac{ \sum_{n' = 1}^{N_{\sf max}}
%\Big( {\sf NPM}( \underline {\boldsymbol P}_{n'}, \hsppp \alpha^{\star}) + Y_{n'} \Big) }
%{N_{\sf max}}.
%\nonumber
%\end{eqnarray}
\begin{eqnarray}
{\sf Res}(n) & = & {\sf Res}(n-1) +
{\sf SE} ( {\underline {\boldsymbol P}}_{n} ) + X_{n}
- \frac{ \sum_{n' = 1}^{N_{\sf max}}
\Big( {\sf SE} ( {\underline {\boldsymbol P}}_{n'} ) + X_{n'} \Big) } {N_{\sf max}}
\nonumber \\
{\sf Coord}(n) & = & {\sf Coord}(n-1) +
{\sf NPM}( \underline {\boldsymbol P}_{n}, \hsppp \alpha^{\star}) + Y_{n}
- \frac{ \sum_{n' = 1}^{N_{\sf max}}
\Big( {\sf NPM}( \underline {\boldsymbol P}_{n'}, \hsppp \alpha^{\star}) + Y_{n'} \Big) }
{N_{\sf max}}
\nonumber
\end{eqnarray}
with ${\sf Res}(0) = 0 = {\sf Coord}(0)$. It can be checked that
${\sf Res} ( N_{\sf max}) = 0 = {\sf Coord}( N_{\sf max} )$ and thus these choices of
tracking functions allow comparison of resilience and coordination in the group relative to
long-term trends (captured by the time-index $N_{\sf max}$) by choosing $N_{\sf max}$
sufficiently large and appropriately.

\section{Numerical Studies}
\label{sec5}
We now illustrate the efficacy of the proposed theoretical framework with various
numerical studies. In the first numerical study, activity data from a two-state HMM
framework with a hurdle-based geometric model corresponding to ${\sf p}_0 = 0.4,
{\sf q}_0 = 0.6, \gamma_0 = 0.1, \mu_0 = 0.3, \gamma_1 = 0.2$ and $\mu_1 = 0.4$ is
generated over $N = 1500 \cdot \delta$ days where $\delta = 7$ (or approximately
$29$ years). The number of days of activity and number of attacks
%as well as the states in the HMM framework
over %the first $500$
%this entire activity period
the first $500$ time-windows of this period are plotted in Fig.~\ref{fig_hbg_expt1}(a).
Based on correlations between $X_n$ and resilience (and $Y_n$ and coordination),
we declare\footnote{The choices of the parameter settings are motivated by a
comparative analysis of the trends of many terrorist groups. The scope of this
discussion is left for a separate work; see~\citep{vasanth_compstudy2015}.} that the
group is resilient or coordinating provided the following conditions are satisfied:
\begin{eqnarray}
{\sf Resilient} & \Longleftrightarrow &
X_n > { \eta}_{\sf X} = 3 \hspp {\sf and} \hspp
{\boldsymbol S}_i \Big|_{i \in \Delta_n} = 1
\nonumber \\
{\sf Coordinating} & \Longleftrightarrow &
Y_n > { \eta}_{\sf Y} = 6 \hspp {\sf and} \hspp
{\boldsymbol S}_i \Big|_{i \in \Delta_n} = 1
\nonumber \\
{\sf Both} \hspp
{\sf resilient} \hspp {\sf and} \hspp {\sf coordinating}
& \Longleftrightarrow &
X_n > { \eta}_{\sf X} = 3, \hspp
Y_n > { \eta}_{\sf Y} = 6 \hspp {\sf and} \hspp
{\boldsymbol S}_i \Big|_{i \in \Delta_n} = 1.
\nonumber
\end{eqnarray}
The above classification leads to $24$, $26$ and $10$ %and $606$
time-windows over which the activity data is resilient, coordinating, and both resilient and
coordinating, %and {\em Active},
respectively.
%On the other hand, over the first $500$ time-windows, there are $9$, $5$, $3$ and $207$
%resilient, coordinating, both resilient and coordinating, and {\em Active} time-windows,
%respectively.
%In general, the resilience and coordination classification coincides with spikes in
%$X_n$ and $Y_n$, respectively.
This classification serves as ``the ground truth'' against which we compare the
performance of different algorithms subsequently. The nature of the states (resilient,
coordinating, or both, {\em Active} or {\em Inactive})
%as well as the number of days of attacks and number of attacks
over the first $500$ time-windows are also plotted in Fig.~\ref{fig_hbg_expt1}(a).

With $\{ {\boldsymbol M}_i \}$ as observations, model parameters are learned with the
Baum-Welch algorithm (see parameter estimates in Table~\ref{table2_HMMclassifications})
and with these estimates, states (over each day) are classified as {\em Active} or
{\em Inactive} using the Viterbi algorithm. Time-windows are then classified as resilient
or coordinating with the following classification parameter settings\footnote{Optimizing
the choice of the classification parameter settings is a task of importance that is left
for a follow-up work. The main purpose of this work is to illustrate the utility of the
proposed ideas rather than to fine-tune the algorithms extensively over the parameter space.}
%in~(\ref{eq_dec1})-(\ref{eq_dec3}):
in~Supplementary A: $\widehat{\eta}_{\sf X} = 3$, $\widehat{\eta}_{\sf Y} = 5$
and $\widehat{\eta} = 3$. Similarly, $\{ X_n \}$, $\{ Y_n \}$ and $\{ (X_n, Y_n) \}$ are
used as observations for parameter learning (see parameter estimates in Table~\ref{table2_HMMclassifications})
and the resultant parameter estimates are fed into the Viterbi algorithm for classifying
time-windows as resilient, coordinating, or both resilient and coordinating, respectively.
On the other hand, the majorization theory-based approach classifies time-windows as
resilient and/or coordinating with parameter settings in~(\ref{Res_class_maj_theory})-(\ref{Coord_class_maj_theory})
for classification chosen as $\widetilde{\eta}_{\sf X} = 3$, $\widetilde{\eta}_{\sf Y} = 6$,
$\underline{\sf SE} = 1$ and $\underline{\sf NPM} = 0.0625$.

To provide metrics on the performance of the different state classification algorithms, we
define a {\em missed detection} event as a time-window in the resilient/coordinating state
(as per ``the ground truth'') that is not declared to be so by the state classification
algorithm. Similarly, a {\em false alarm} event corresponds to a time-window that is declared
as resilient/coordinating when it is not one (as per ``the ground truth''). The probability
of missed detection (${\sf P}_{\sf MD}$) and probability of false alarm (${\sf P}_{\sf FA}$)
are defined as the fraction of true events not declared to be so by the classification
algorithm and fraction of classified events that are not true events, respectively.
%Note that
%in an application context, a smaller ${\sf P}_{\sf MD}$ is of more importance than a larger
%${\sf P}_{\sf FA}$.

Table~\ref{table2_HMMclassifications} provides a summary statistics of missed detection
and false alarm events with the different state classification algorithms as well as the
parameter estimates. From Table~\ref{table2_HMMclassifications}, we note that the
learned model parameters are similar with different sets of observations. Learning with
$\{ {\boldsymbol M}_i \}$ results in good resilience decisions with no miss of true events
and a low false alarm probability. However, $\{ {\boldsymbol M}_i \}$ leads to poor
coordination decisions with a number of misses and false alarms. On the other hand,
learning with $\{ X_n \}$ or $\{ Y_n \}$ or $\{ (X_n, Y_n) \}$ results in no misses of a
true resilience/coordination event, but these approaches lead to a considerably large number
of false alarms. Further, all these decisions come at the cost of model learning latencies.
On the
other hand, the majorization theory-based approach leads to a low miss and false alarm
probability by tuning to the appropriate signature for monitoring. The corresponding state
classifications with these approaches for resilience and coordination over the first $500$
time-windows are presented in Figs.~\ref{fig_hbg_expt1}(b)-(c), respectively.

In addition, Fig.~\ref{fig_hbg_expt1}(d) plots the resilience and coordination tracking
functions of the group over the first $500$ time-windows with $N_{\sf max} = 1500$. From
this plot, while resilience and
coordination follow similar macroscopic patterns, we can also obtain general indicators
as to whether resilience increases first (or coordination increases first). Notable
indicators of the group's behavior include: resilience in the group increases first while
coordination remains stable (around $20$-$30$ time-windows), coordination increases while
resilience diminishes (around $60$-$70$ time-windows), both features diminish ($90$-$130$
time-windows) followed by a spurt and further diminishing ($140$-$220$ time-windows), an increase in
both resilience and coordination ($220$-$250$ time-windows) followed by general diminishing
($250$-$300$ time-windows) and a general spurt ($300$-$450$ time-windows), and a spurt in
coordination at around the $450$th time-window followed by a spurt in resilience. Thus,
the tracking functions provide a close glimpse in terms of the group's {\em Capabilities}
over time. While similar classifications and tracking have been proposed
by~\citet{bakker_dark_network} and~\citet{rand}, our work differs from these prior works
by developing a theoretical (and automated) framework to classify and track resilience and
coordination in the group instead of depending on subject-matter experts for classification.

%From
%this plot, we observe that the group is initially better at coordination than
%in resilience and this trend continues (barring a few spurts) for $\approx 150$
%time-windows. From the $150$ to $200$ time-window period, both resilience and
%coordination in the group diminish and coordination picks up steam around the $200$th
%time-window before resilience can catch up. After another decay in the group's
%behavior, we note that resilience picks up around $300$ time-windows trailed by
%its coordination capability. As resilience continues to motor ahead, coordination
%falls from the $400$ to $450$ time-window before stabilizing.

\begin{figure}[htb!]
\begin{center}
\begin{tabular}{cc}
%\begin{minipage}{5.5in}
%\centerline{
%
\includegraphics[height=2.55in,width=2.9in] {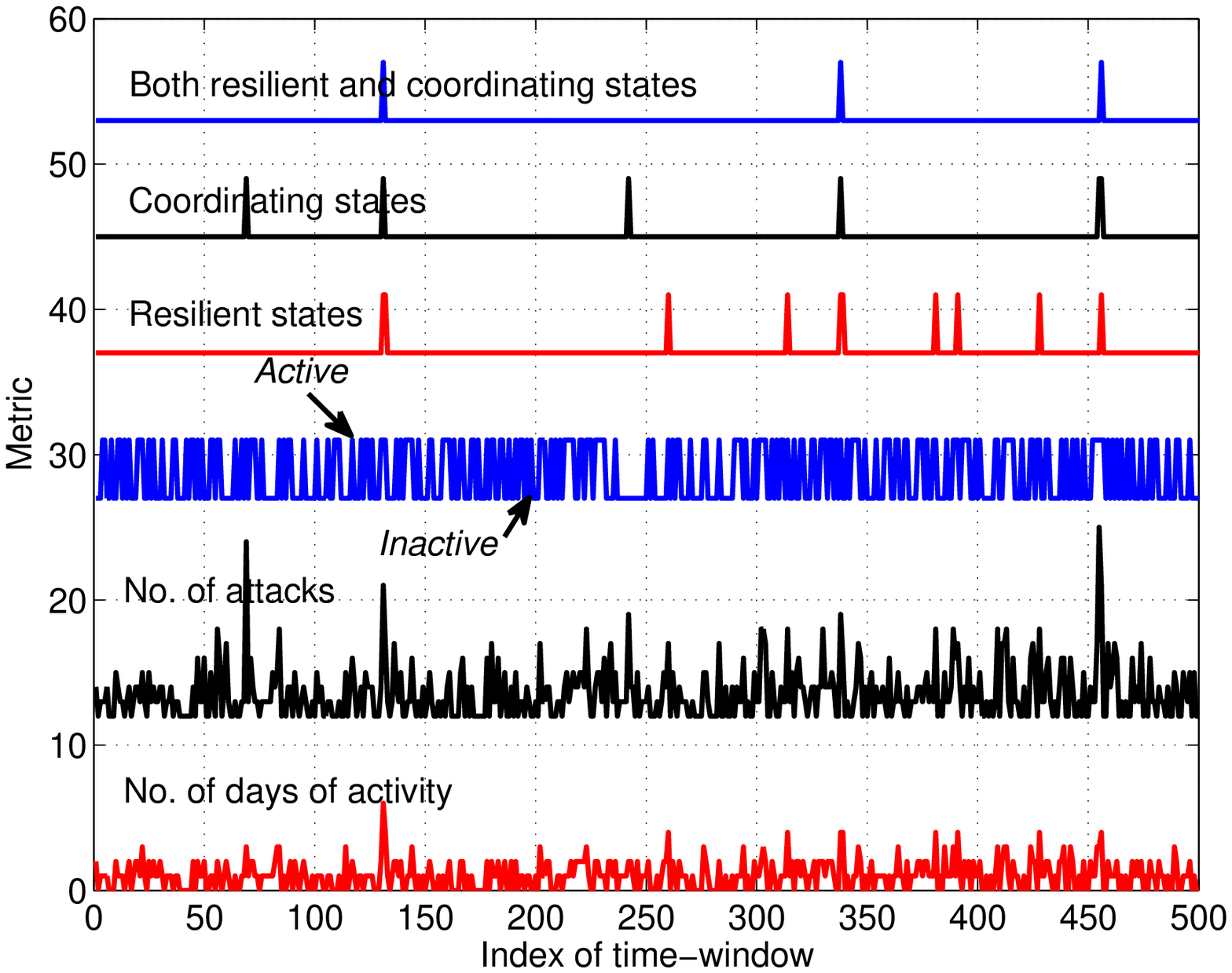}
%{figs/newfig_HBGgenerated_noofdaysactivity_noattacks_v4.eps}
&
\includegraphics[height=2.45in,width=2.7in] {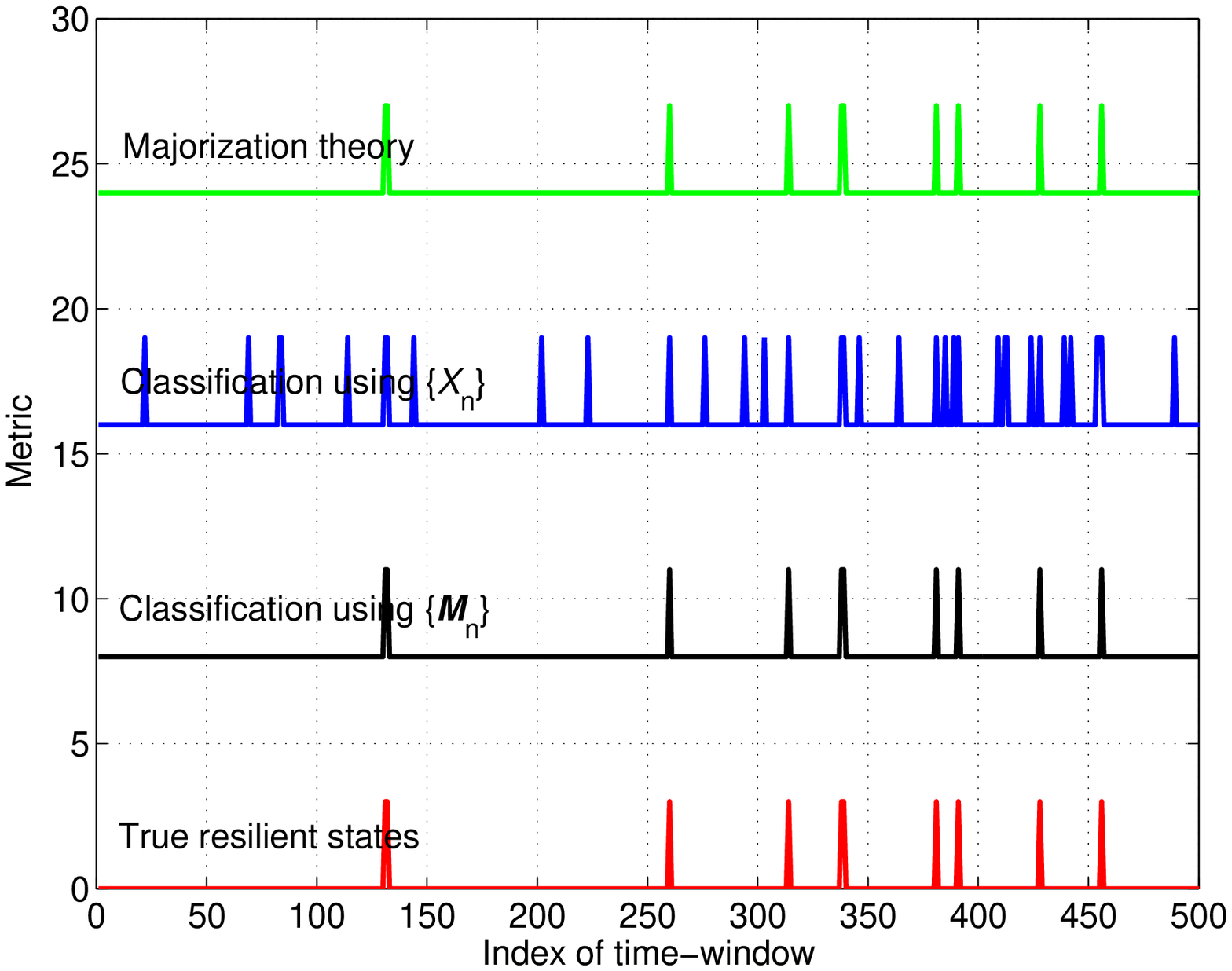}
%{figs/newfig_HBGgenerated_res_coord_both_v3.eps}
\\
(a) & (b)
\\
\includegraphics[height=2.45in,width=2.7in] {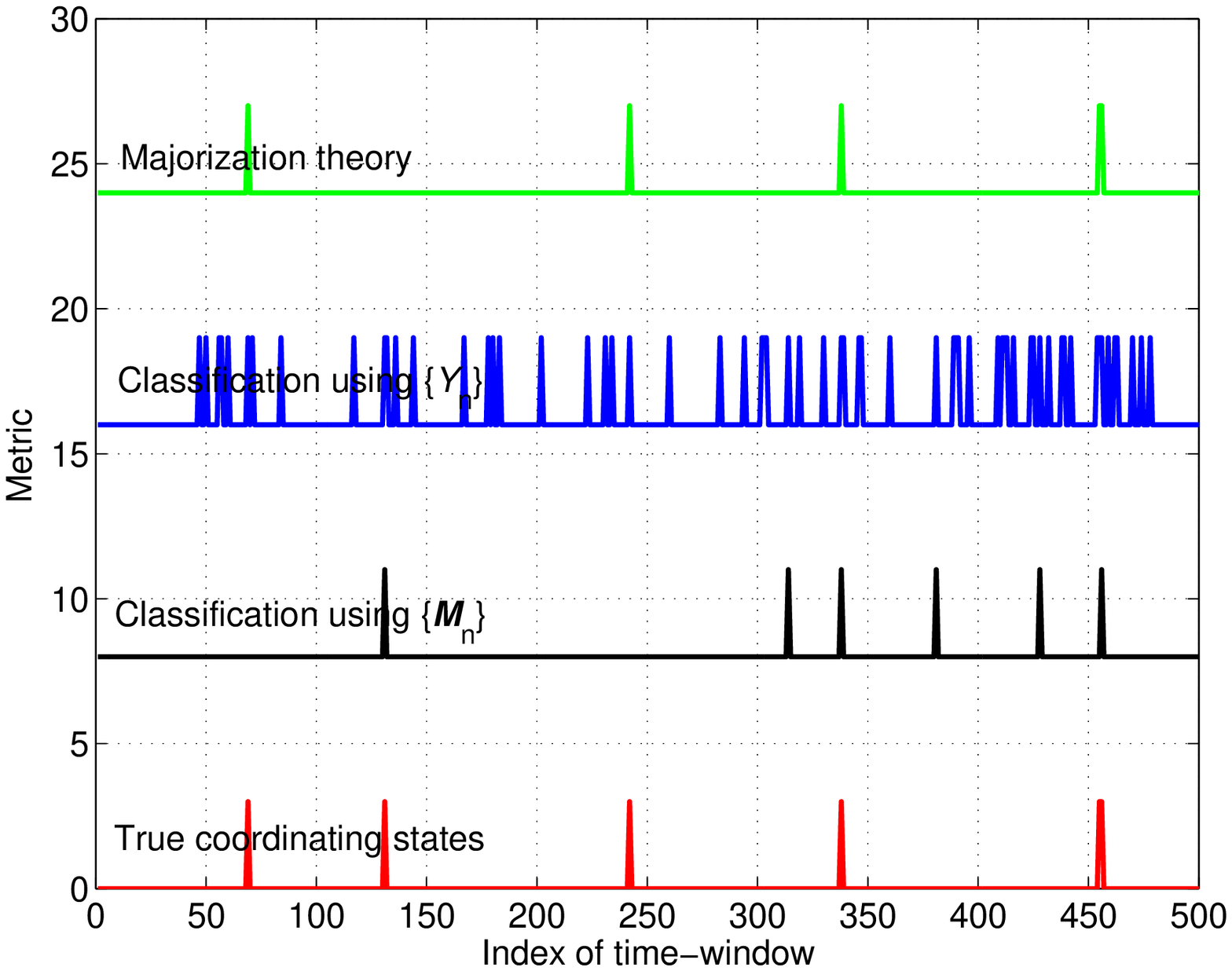}
&
\includegraphics[height=2.45in,width=2.7in] {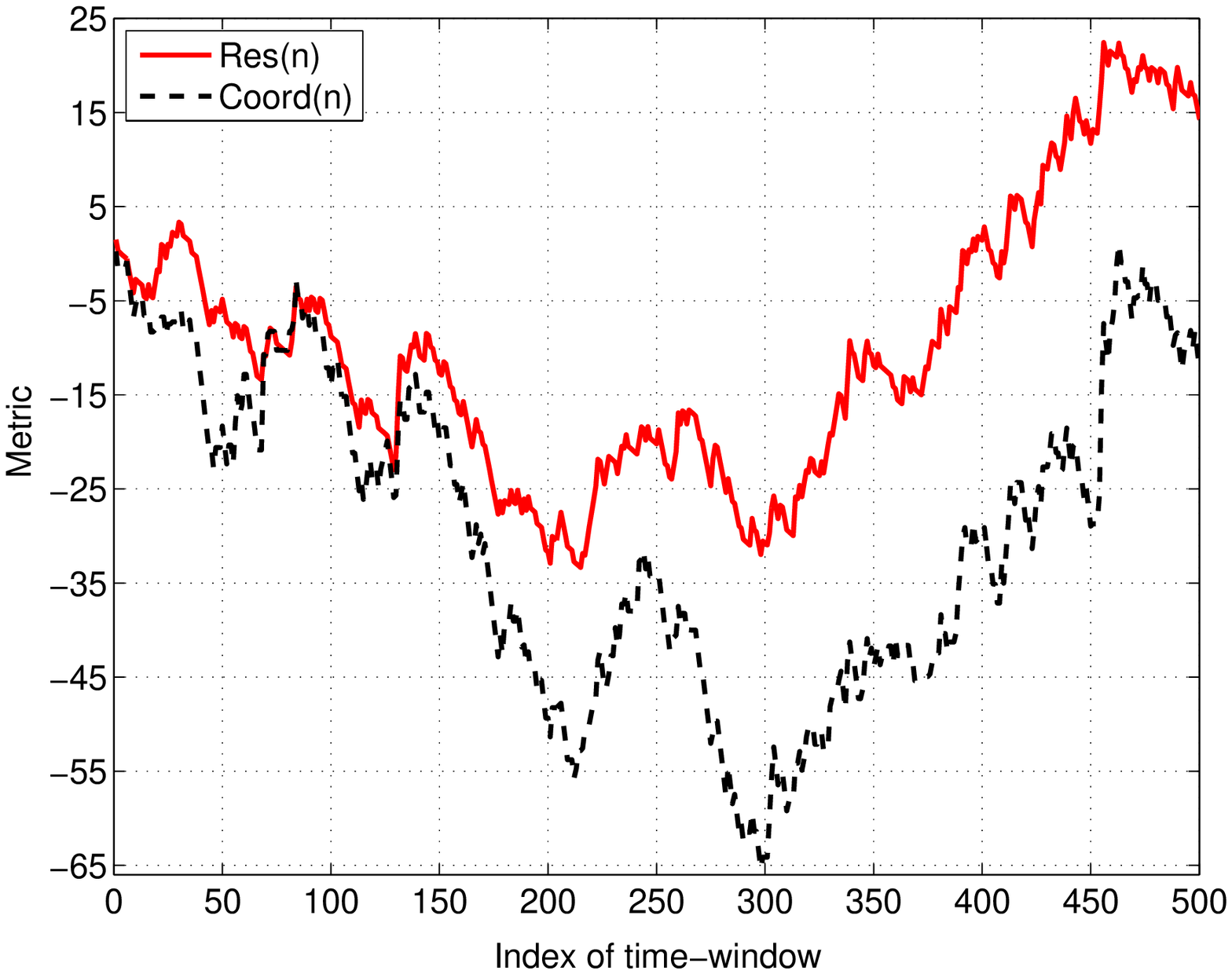}
\\
(c) & (d)
%}
%\end{minipage}
\end{tabular}
\caption{\label{fig_hbg_expt1}
%(a)
(a) Number of days of activity, number of attacks, true {\em Active}, resilient
and coordinating states, (b) Resilient and (c) Coordinating state classification
with different algorithms, and (d) Resilience and coordination tracking functions
over the first $500$ time-windows.}
\end{center}
\end{figure}

\begin{table*}[htb!]
\caption{Probabilities of missed detection and false alarm with different state
classification algorithms for the data generated from a hurdle-based geometric model
and for FARC data.}
\label{table2_HMMclassifications}
\begin{tabular}{c|c|c|c|c}
\hline
\multicolumn{5}{c}{{\sf Data from hurdle-based geometric model}} \\
\hline
%& \multicolumn{4}{c||}{ ${\sf P}_{\sf PM}(\alpha^{\star})$ } &
%\multicolumn{4}{c}{ ${\sf P}_{\sf NPM}(\alpha^{\star})$ }  \\
%\hline
{\sf Setting} &{\sf Parameters} & %\multicolumn{3}{|c|}{ {\sf Thresholds} } &
\multicolumn{3}{c}{ {\sf Number of states classified and}
(${\sf P}_{\sf MD}, {\sf P}_{\sf FA}$) }
\\
\cline{3-5}
%& & \multicolumn{3}{c}{ \hline } \\
& & {\sf Resilient } & {\sf Coordinating} & {\sf Both}
\\ \hline %\hline
%{\sf True} & {\sf Resilient } & {\sf Coordinating} & {\sf Both} &
%{\sf Resilient } & {\sf Coordinating} & {\sf Both}
%\\ \hline
{\sf True} &
$\gamma_0 = 0.1, \mu_0 = 0.3$ & %$\widetilde{\eta}_{\sf X} = 3$ &
%$\widetilde{\eta}_{\sf Y} = 6$ & ($\widetilde{\eta}_{\sf X},
%\widetilde{\eta}_{\sf Y}) = (3,6)$ &
$24$ & $26$ & $10$ \\
{\sf Observations} & $\gamma_1 = 0.2, \mu_1 = 0.4$ %& & &
& & & \\
\hline
{\sf Learning} & $\widehat{\gamma}_0 = 0.1058, \widehat{\mu}_0 = 0.3042$
& $26$ & $15$ & $15$
\\
{\sf with} $\{ {\boldsymbol M}_i \}$ &
$\widehat{\gamma}_1 = 0.3861, \widehat{\mu}_1 = 0.4433$
& ($0, 0.0769$) & ($0.6154, 0.3333$) & ($0, 0.3333$) \\
\hline
{\sf Learning} & $\widehat{\gamma}_0 = 0.1295, \widehat{\mu}_0 = 0.4776$
& $101$ & $-$ & $-$
\\
{\sf with} $\{ X_n \}$ &
$\widehat{\gamma}_1 = 0.2719, \widehat{\mu}_1 = 0.4776$
& ($0, 0.7624$) &  & \\ \hline
{\sf Learning} & $\widehat{\gamma}_0 = 0.1515, \widehat{\mu}_0 = 0.2019$ & $-$ &  $192$ & $-$
\\
{\sf with} $\{ Y_n \}$ & $\widehat{\gamma}_1 = 0.2528, \widehat{\mu}_1 = 0.5287$
& & ($0, 0.8646$)  & \\ \hline
{\sf Learning} & $\widehat{\gamma}_0 = 0.1306, \widehat{\mu}_0 = 0.3286$ & $-$ & $-$ & $127$
\\
{\sf with} $\{ (X_n, Y_n) \}$ &
$\widehat{\gamma}_1 = 0.2553, \widehat{\mu}_1 = 0.4797$
& & & ($0, 0.9213$) \\ \hline
{\sf Majorization} & $-$ & $26$ & $27$ & $8$ \\
{\sf theory} &  & ($0, 0.0769$) & ($0.0769, 0.1111$) & ($0.2000, 0$) \\
\hline \hline
%%%
%%%
%%%
\multicolumn{5}{c}{{\sf FARC data}} \\
\hline
{\sf Setting} &{\sf Parameters} & %\multicolumn{3}{|c|}{ {\sf Thresholds} } &
\multicolumn{3}{c}{ {\sf Number of states classified and}
(${\sf P}_{\sf MD}, {\sf P}_{\sf FA}$) }
\\
\cline{3-5}
& & {\sf Resilient} & {\sf Coordinating} & {\sf Both}
\\ \hline
{\sf True} &
$-$ &
$37$ & $18$ & $14$ \\
{\sf Observations} &
& & & \\
\hline
{\sf Learning} & $\widehat{\gamma}_0 = 0.0953, \widehat{\mu}_0 = 0.0762$
& $27$ & $13$ & $13$
\\
{\sf with} $\{ {\boldsymbol M}_i \}$ &
$\widehat{\gamma}_1 = 0.3988, \widehat{\mu}_1 = 0.3087$
& ($0.2703, 0$) & ($0.3889, 0.1538$) & ($0.2143, 0.1538$) \\
\hline
{\sf Learning} & $\widehat{\gamma}_0 = 0.0933, \widehat{\mu}_0 = 0.3505$
& $125$ & $-$ & $-$
\\
{\sf with} $\{ X_n \}$ &
$\widehat{\gamma}_1 = 0.3921, \widehat{\mu}_1 = 0.3505$
& ($0, 0.7040$) &  & \\ \hline
{\sf Learning} & $\widehat{\gamma}_0 = 0.0951, \widehat{\mu}_0 = 0.1232$ & $-$ &  $73$ & $-$
\\
{\sf with} $\{ Y_n \}$ & $\widehat{\gamma}_1 = 0.2500, \widehat{\mu}_1 = 0.5745$
& & ($0, 0.7534$)  & \\ \hline
{\sf Learning} & $\widehat{\gamma}_0 = 0.0949, \widehat{\mu}_0 = 0.0752$ & $-$ & $-$ & $73$
\\
{\sf with} $\{ (X_n, Y_n) \}$ &
$\widehat{\gamma}_1 = 0.3958, \widehat{\mu}_1 = 0.3082$
& & & ($0, 0.8082$) \\ \hline
{\sf Majorization} & $-$ & $27$ & $15$ & $13$ \\
{\sf theory} &  & ($0.2703, 0$) & ($0.2778, 0.1333$) & ($0.2143, 0.1538$) \\
\hline
\end{tabular}
\end{table*}

In the %third
next numerical study, real terrorism data from~\citet{rdwti} on the FARC
terrorist group over a time-period of $3640$ days (or $520$ time-windows with $\delta = 7$
days) over the $1998$-$2007$ period is considered. States are classified as resilient,
coordinating, or both resilient and coordinating according to a certain manual
classification procedure that correlates with spurts in $X_n$, $Y_n$ and $(X_n, Y_n)$,
respectively. Model parameters are learned with $\{ {\boldsymbol M}_i \}$, $\{ X_n \}$,
$\{ Y_n \}$ and $\{ (X_n, Y_n) \}$ as observations
and states are classified with the following classification parameter settings:
$\widehat{\eta}_{\sf X} = 3$, $\widehat{\eta}_{\sf Y} = 5$ and $\widehat{\eta} = 3$. Similarly,
states are classified with the majorization theory-based approach using
$\widetilde{\eta}_{\sf X} = 3$, $\widetilde{\eta}_{\sf Y} = 5$, $\underline{\sf SE} = 1$ and
$\underline{\sf NPM} = 0.0204$. Analogous to Fig.~\ref{fig_hbg_expt1}, Figs.~\ref{fig_farc}(a)-(c) plot
the resilient and coordinating state classifications, whereas Fig.~\ref{fig_farc}(d) plots the
resilience and coordination tracking functions corresponding to $N_{\sf max} = 520$ for FARC.
Table~\ref{table2_HMMclassifications} also provides a summary statistics of missed detection
and false alarm events with the different state classification algorithms as well as the parameter
estimates. Clearly, we see that the majorization theory-based approach performs as well as
(or better than) the parametric approach based on $\{ {\boldsymbol M}_i \}$,
$\{ X_n \}$, $\{ Y_n \}$, or $\{ (X_n, Y_n) \}$. Further, the tracking functions
capture the two major spurts in resilience and coordination in the group, as well as the relative
growth/decay in these attributes over time; see~\citep[Supplementary B]{vasanth_aoas2013}
and~\citep{rand} for an explanation.

\ignore{
\begin{figure}[htb!]
\begin{center}
\begin{tabular}{cc}
%\begin{minipage}{5.5in}
%\centerline{
%
\includegraphics[height=2.55in,width=2.8in] {figs/newfig_resilience_scatter.eps}
%{figs/newfig_HBGgenerated_noofdaysactivity_noattacks_v4.eps}
&
\includegraphics[height=2.55in,width=3.0in] {figs/newfig_coordination_scatter.eps}
%{figs/newfig_HBGgenerated_res_coord_both_v3.eps}
\\
(a) & (b)
\\
\includegraphics[height=2.45in,width=2.7in] {figs/newfig_both_scatter.eps}
&
\includegraphics[height=2.45in,width=2.7in] {figs/newfig_parameters_scatter.eps}
\\
(c) & (d)
%}
%\end{minipage}
\end{tabular}
\caption{\label{fig_scatter}
%(a)
Scatter plot of ${\sf P}_{\sf MD}$ vs.\ ${\sf P}_{\sf FA}$ for (a) resilience,
(b) coordination, and (c) both resilience and coordination classification with
different algorithms for $30$ different models. (d) Scatter plot of model
parameters.}
\end{center}
\end{figure}
}

\begin{figure}[htb!]
\begin{center}
\begin{tabular}{cc}
%\begin{minipage}{5.5in}
%\centerline{
%
\includegraphics[height=2.55in,width=2.9in] {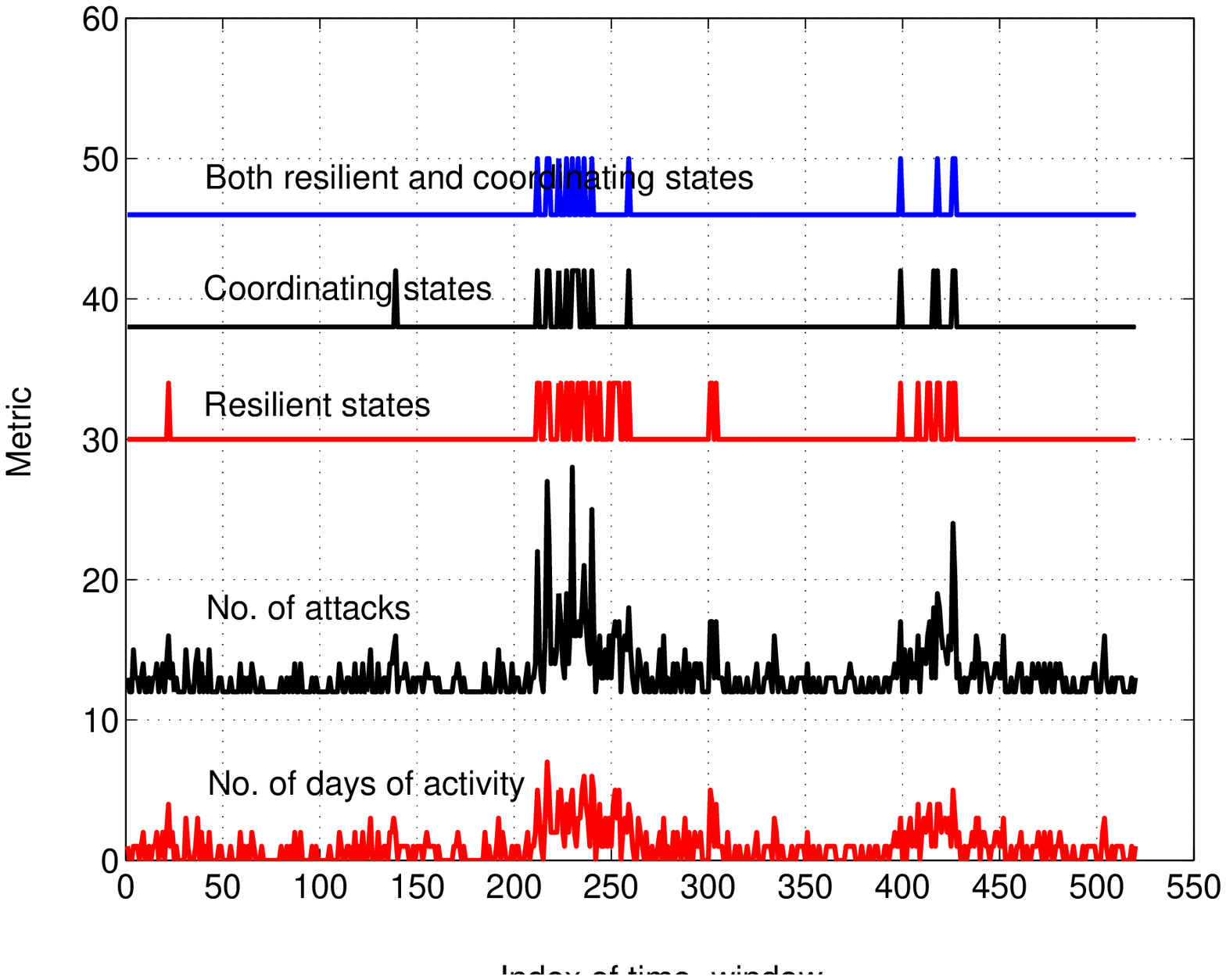}
%{figs/newfig_HBGgenerated_noofdaysactivity_noattacks_v4.eps}
&
\includegraphics[height=2.45in,width=2.7in] {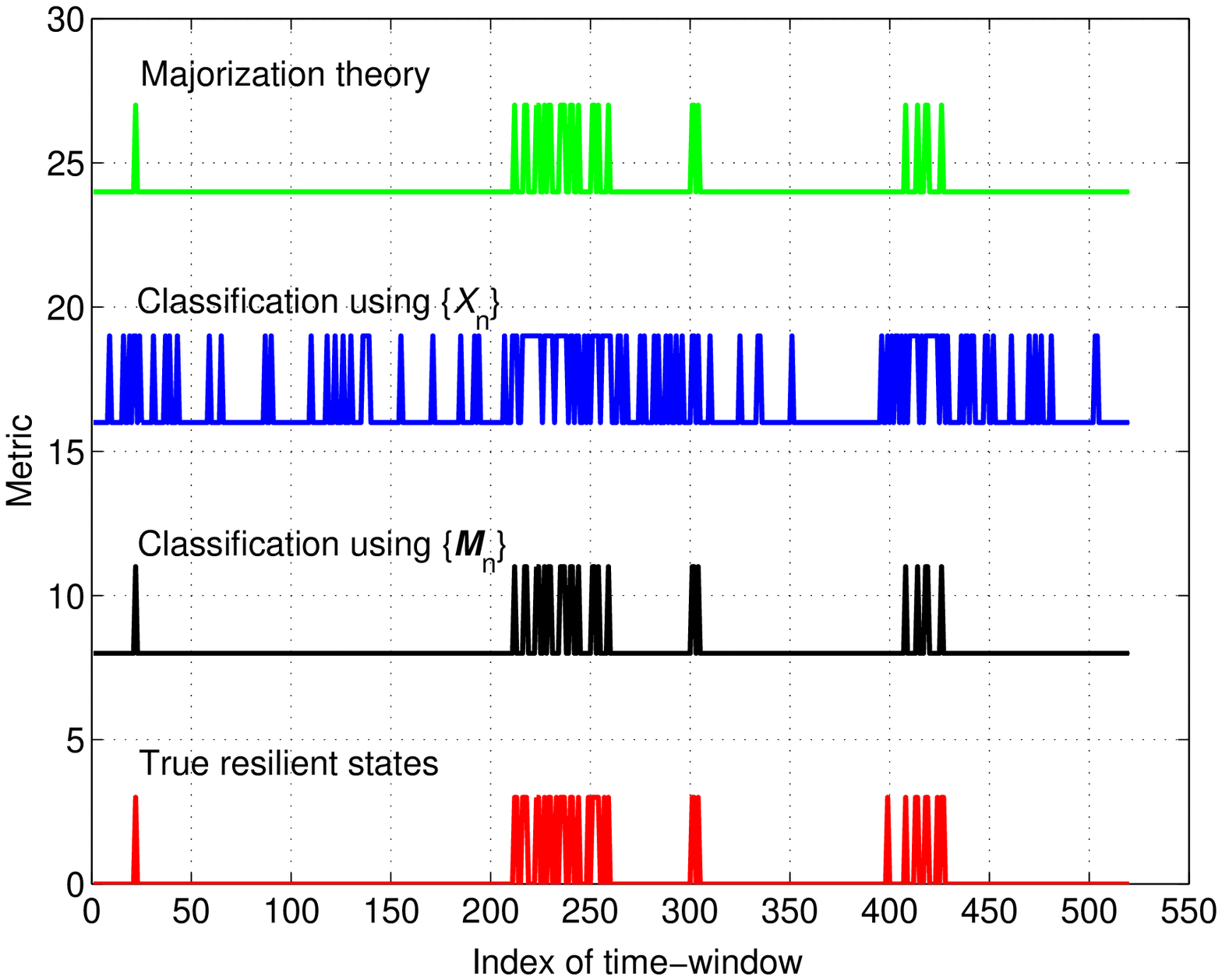}
%{figs/newfig_HBGgenerated_res_coord_both_v3.eps}
\\
(a) & (b)
\\
\includegraphics[height=2.45in,width=2.7in] {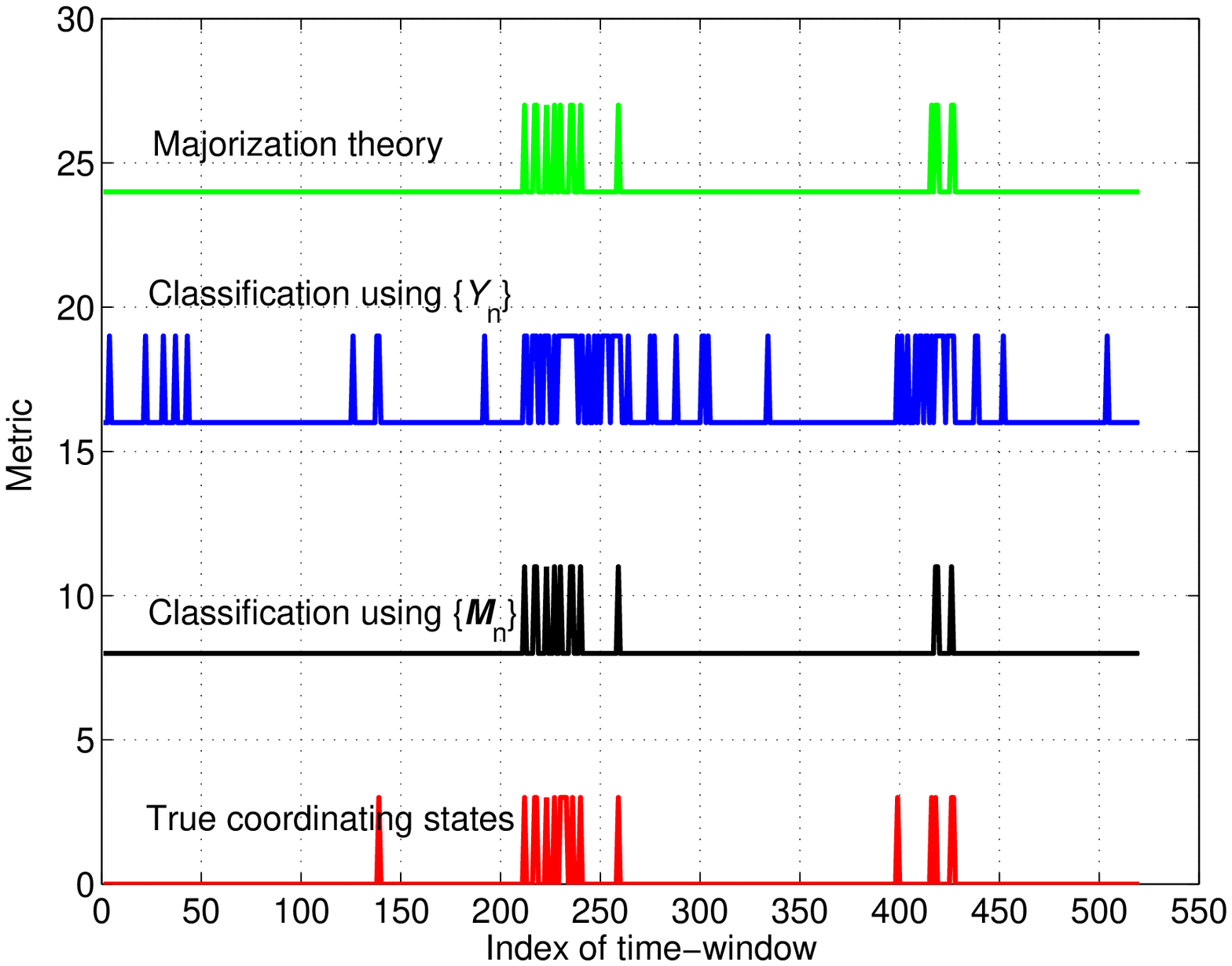}
&
\includegraphics[height=2.45in,width=2.7in] {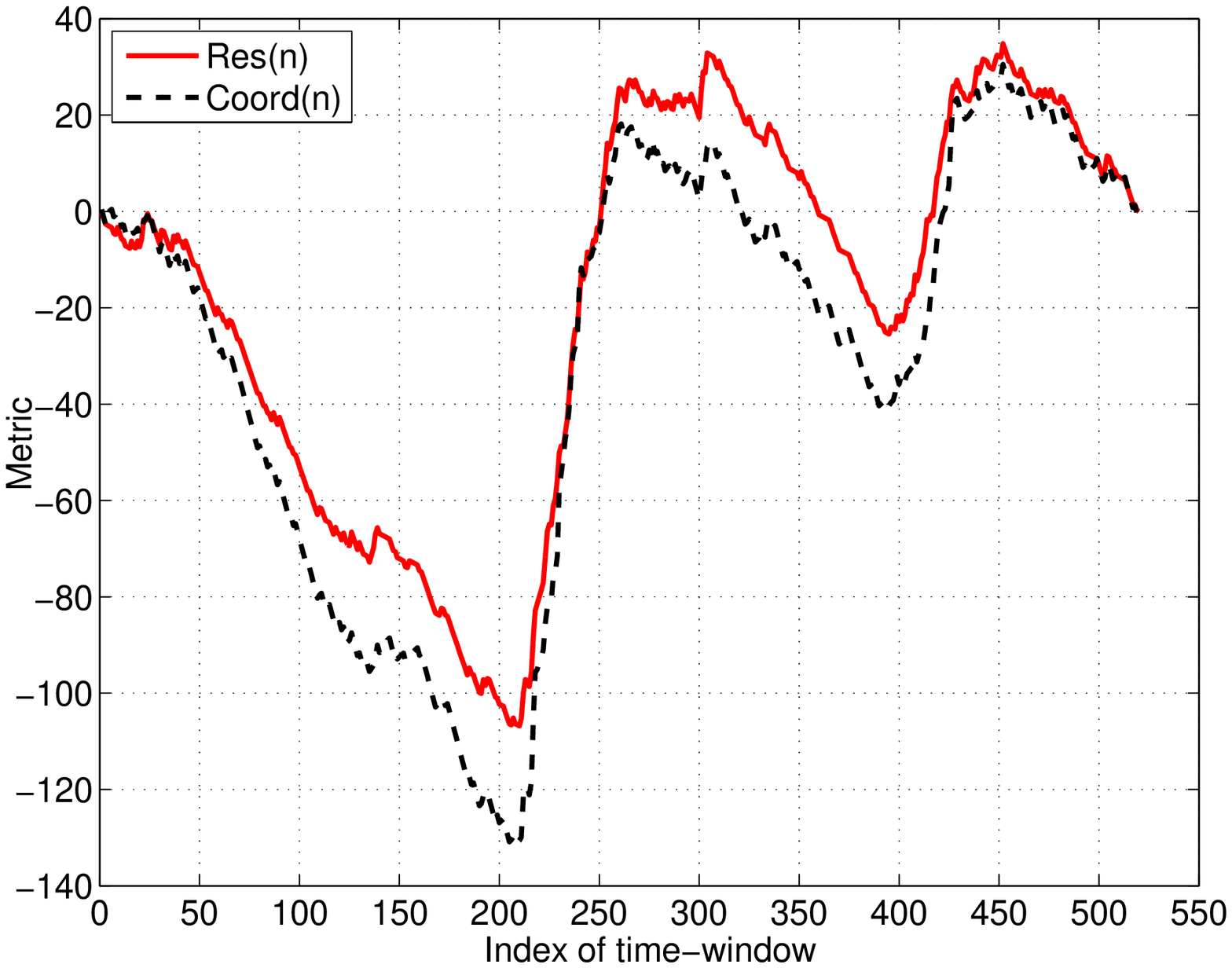}
\\
(c) & (d)
%}
%\end{minipage}
\end{tabular}
\caption{\label{fig_farc}
(a) Number of days of activity, number of attacks, and true states of FARC,
(b) Resilient and (c) Coordinating state classification with different algorithms, and
(d) Resilience and coordination tracking functions over the time-period of interest.}
\end{center}
\end{figure}

\section{Concluding Remarks}
In the light of recent interest in modeling and monitoring of terrorist activity,
this work focussed on detecting sudden spurts in the activity profile of terrorist
groups. Most work in this area are parametric in nature, which renders their
real-life application difficult. In particular, parametric approaches to spurt
detection often rely on past behavior for prediction, but terrorists' behavior
changes quickly enough to make some of this analysis useless. To overcome this
fundamental difficulty, we proposed a non-parametric approach based on majorization
theory to detect sudden and abrupt changes in the {\em Capabilities} of the group.
Leveraging the notion of catalytic majorization, we developed a simple approach to
increment/decrement an appropriate statistic that captures different facets of the
terrorist group (such as resilience and level of coordination) in this work. Future
work will consider the application of this approach to a broad swathe of terrorist
groups' activity profiles as well as applications in social network
settings.

\begin{supplement} [id=suppA]
\label{sec_suppa}
\stitle{Supplementary A: Update equations for the Baum-Welch algorithm}
\slink[doi]{}%{10.1214/10-AOAS395SUPPA}
\sdatatype{.pdf}
\sdescription{This section derives the update equations for model parameter
learning and mechanism association with different types of approaches within the HMM framework.}
\end{supplement}

%\thesuppdoi{suppA}
%\ref{suppA}

We now develop update equations for the observation density parameters
($\widehat{\gamma}_j$ and $\widehat{\mu}_j$) when Baum-Welch algorithm~\citep{rabiner}
is applied to a training-set of $N$ observations ${\cal O} = \{ O_n, \hsppp n = 1,
\cdots , N \}$. Let the corresponding hidden states be given as ${\cal S} = \{ S_n,
\hsppp n = 1, \cdots, N \}$ with $S_0$ %= \varnothing$.
initialized according to an initial probability density $\{ \pi_j, \hsppp j = 0,1\}$.
The Baum auxiliary function with
current/initial estimate of HMM parameters ${\boldsymbol{\bar{\lambda}}}$ as a function
of the optimization variable $\boldsymbol\lambda$, denoted as
$Q(\boldsymbol\lambda, \hsppp {\boldsymbol{\bar{\lambda}}})$, is given as:
\begin{eqnarray}
Q(\boldsymbol\lambda, \hsppp {\boldsymbol{\bar{\lambda}}}) \triangleq
\sum_{ {\cal S}} \log \Big( {\sf P}( {\cal O}, \hsppp {\cal S} | \boldsymbol\lambda)
\Big) \cdot {\sf P}( {\cal O}, \hsppp {\cal S} | {\boldsymbol{\bar{\lambda}}} ). \nonumber
\end{eqnarray}
We proceed via the same approach elucidated by~\citet{bilmes} leading to
%\begin{align} &
\begin{eqnarray}
\frac{
Q(\boldsymbol\lambda, \hsppp {\boldsymbol{\bar{\lambda}}})
 %\Big|_{ {\sf Obs.} \hspp {\sf density}}
}  { {\sf P}( {\cal O} | {\boldsymbol{\bar{\lambda}}} ) }
& = &
\frac{ \sum _{n = 1}^N  \sum _ {S_0 }
\log \Big( {\sf P} \big( S_0 | \boldsymbol\lambda \big) \Big)
\cdot {\sf P} ( {\cal O}, \hsppp S_0 | {\boldsymbol{\bar{\lambda}}} )
} { {\sf P}( {\cal O} | {\boldsymbol{\bar{\lambda}}} ) } +
\nonumber \\
& &
\frac{ \sum %\limits
_{n = 1}^N  \sum %\limits
_ {S_n, \hsppp S_{n-1} }
\log \Big( {\sf P} \big( O_n, \hsppp S_n | S_{n-1}, \hsppp \boldsymbol\lambda \big)
\Big) \cdot {\sf P} ( {\cal O}, \hsppp S_n, \hsppp S_{n-1} | {\boldsymbol{\bar{\lambda}}} )
}{ {\sf P}( {\cal O} | {\boldsymbol{\bar{\lambda}}} ) }.
\nonumber
\end{eqnarray}
The component of the auxiliary function corresponding to the initial probability
parameters can be written as:
\begin{eqnarray}
\frac{
Q(\boldsymbol\lambda, \hsppp {\boldsymbol{\bar{\lambda}}})
\Big|_{ {\sf Init.} \hspp {\sf probability}}
}  { {\sf P} ( {\cal O} | {\boldsymbol{\bar{\lambda}}} ) }
& = &
\sum_{j = 0}^1 \log \big( \pi_j \big) \cdot
\underbrace{ {\sf P}(S_0 = j | {\cal O}, \hsppp {\boldsymbol{\bar{\lambda}}} )
}_{ \triangleq \gamma_0(j) },
\nonumber
\end{eqnarray}
where the iterative update for $\gamma_0(j)$ follows from the forward and backward
algorithms (see~\citet[Sec.\ IIIA and B]{rabiner}). It can be easily seen that
\begin{eqnarray}
\widehat{\pi}_j = \frac{ \gamma_0(j) }{ \sum_{k = 0}^1 \gamma_0(k)}.
\nonumber
\end{eqnarray}
Similarly, the component of the auxiliary function corresponding to the transition
probability parameters can be written as:
\begin{eqnarray}
& & \frac{
Q(\boldsymbol\lambda, \hsppp {\boldsymbol{\bar{\lambda}}})
\Big|_{ {\sf Trans.} \hspp {\sf probability}}
}  { {\sf P} ( {\cal O} | {\boldsymbol{\bar{\lambda}}} ) }
\nonumber \\
& & {\hspace{0.3in}}
=
\sum_{n = 1}^N \sum_{i=0}^1 \sum_{j = 0}^1
\log \Big( {\sf P}(S_n = j|S_{n-1} = i, \boldsymbol\lambda) \Big) \cdot
\underbrace{ {\sf P}(S_n = j, S_{n-1} = i | {\cal O}, \hsppp {\boldsymbol{\bar{\lambda}}} )
}_{ \triangleq \zeta_{n-1}(i,j) } ,
\nonumber
\end{eqnarray}
where the iterative update for $\zeta_{n-1}(i,j)$ follows from~\citet[Sec.\ IIIC]{rabiner}. A
simple constrained optimization problem results in
\begin{eqnarray}
\widehat{\sf p}_0 & = & \frac{ \sum_{n = 1}^N \zeta_{n-1}(0,1) }
{ \sum_{n = 1}^N \sum_{j = 0}^1 \zeta_{n-1}(0,j) }, \nonumber \\
\widehat{\sf q}_0 & = & \frac{ \sum_{n = 1}^N \zeta_{n-1}(1,0) }
{ \sum_{n = 1}^N \sum_{j = 0}^1 \zeta_{n-1}(1,j) }.
\nonumber
\end{eqnarray}
On the other hand, the component of the auxiliary function corresponding to the
optimization of observation density parameters can be expressed as:
\begin{eqnarray}
\frac{
Q(\boldsymbol\lambda, \hsppp {\boldsymbol{\bar{\lambda}}})
\Big|_{ {\sf Obs.} \hspp {\sf density}}
}  { {\sf P} ( {\cal O} | {\boldsymbol{\bar{\lambda}}} ) }
& = &
\sum _{n = 1}^N  \sum_ {j = 0 }^1
\log \Big( {\sf P} \big( O_n | S_{n} = j, \hsppp \boldsymbol\lambda \big) \Big)
\cdot
\underbrace{ {\sf P}(S_n = j | {\cal O}, \hsppp {\boldsymbol{\bar{\lambda}}} )
}_{ \triangleq \gamma_n(j) },
\nonumber
\end{eqnarray}
where the iterative update for $\gamma_n(j)$ also follows from the forward and backward
algorithms (see~\citep[Sec.\ IIIA and B]{rabiner}). We now develop update equations for
$\gamma_j$ and $\mu_j$, specialized based on the observations in the HMM.

\section{$\{ {\mb M}_i \}$ as Observations}
%\noindent {\bf \em \underline{$\left\{ M_n \right\}$ as Observations:}}
With ${\cal O} =
\left\{ {\mb M}_i \right\}$, using the hurdle-based geometric model, it is straightforward to
check the following update equations for $\widehat{\gamma}_j$ and $\widehat{\mu}_j$:
\begin{eqnarray}
\widehat{\gamma}_j & = & \frac{ \sum_{n= 1}^N \indic \left( {\mb M}_n > 0 \right) \cdot
\gamma_n(j) } { \sum_{n = 1}^N \gamma_n(j)}, \nonumber \\
\widehat{\mu}_j & = & \frac{ \sum_{n = 1}^N ({\mb M}_n - 1) \cdot \indic \left({\mb M}_n > 0 \right)
\cdot \gamma_n(j) } { \sum_{n = 1}^N {\mb M}_n \cdot \indic \left({\mb M}_n > 0 \right) \cdot
\gamma_n(j) }.
\nonumber
\end{eqnarray}
By definition, it is clear that $0 \leq \{ \widehat{\gamma}_j, \hsppp \widehat{\mu}_j \}
\leq 1$.

\section{Associating State Changes with {\em Tactics} in Inferencing with $\{ {\mb M}_i \}$}
The proposed approach in Sec.~3.1 of the main paper requires %an explicit
a mapping of
%$\widehat{\Sb}_i \Big|_{ i \in \Delta_n}$
the state estimates $\widehat{\Sb}_i \Big|_{ i \in \Delta_n}$ to appropriate resilience
and/or coordination metric(s). We now provide details on the specific choices of $f(\cdot)$
and $g(\cdot)$ in~(3.2) of the main paper to associate state changes with specific changes in {\em Tactics}.

Motivated by the discussion on resilience and coordination,
we propose the following mappings for inferencing corresponding to certain choices of
$\widehat{\eta}$, $\widehat{\eta}_{\sf X}$ and $\widehat{\eta}_{\sf Y}$:
\begin{eqnarray}
\label{eq_dec1}
\sum_{ i \in \Delta_n} \widehat{\Sb}_i > \widehat{\eta} \hspp {\sf and} \hspp
X_n > \widehat{\eta}_{\sf X} & \Longleftrightarrow &
{\sf Group} \hspp {\sf is} \hspp {\sf resilient} \hspp {\sf in} \hspp \Delta_n
%\nonumber
\\
\label{eq_dec2}
\sum_{ i \in \Delta_n} \widehat{\Sb}_i > \widehat{\eta} \hspp {\sf and} \hspp
Y_n > \widehat{\eta}_{\sf Y} & \Longleftrightarrow &
{\sf Group} \hspp {\sf is} \hspp {\sf coordinating} \hspp {\sf in} \hspp \Delta_n
%\nonumber
\\
\label{eq_dec21}
\sum_{ i \in \Delta_n} \widehat{\Sb}_i > \widehat{\eta}, \hspp
X_n > \widehat{\eta}_{\sf X} \hspp {\sf and} \hspp
Y_n > \widehat{\eta}_{\sf Y} & \Longleftrightarrow &
{\sf Group} \hspp {\sf is} \hspp {\sf resilient}
\hspp %{\sf and}
\& \hspp {\sf coordinating} \hspp {\sf in} \hspp \Delta_n
\\
\label{eq_dec3}
\sum_{ i \in \Delta_n} \widehat{\Sb}_i > \widehat{\eta}
& \Longleftrightarrow &
{\sf Group} \hspp {\sf is} \hspp {\emph Active} \hspp {\sf in} \hspp \Delta_n.
%\nonumber
\end{eqnarray}

\section{$\left\{ (X_n, \hsppp Y_n) \right\}$ as Observations}
%\noindent {\bf \em \underline{$\left\{ (X_n, \hsppp Y_n) \right\}$ as Observations:}}
With the joint sequence
\begin{eqnarray}
{\cal O} = \left\{ (X_n, \hsppp Y_n), \hsppp n = 1, \cdots, K \right\}
\nonumber
\end{eqnarray}
as observations, using the density function
\begin{align}
%& {\hspace{0.0in}}
& {\hspace{-0.2in}}
{\sf P} \left( X_n = k, \hsppp Y_n = r \Big|  \Sb_{ i} |_{ i \in \Delta_n } = j
\right)
\nonumber \\ & {\hspace{0.3in}}
=
{\delta \choose k} { r -1 \choose r - k} \cdot
%\nonumber \\ & {\hspace{2in}}
(1 - \gamma_j)^{\delta  - k} (\gamma_j)^{k} \cdot (1 - \mu_j)^k
(\mu_j)^{r - k}, \hspp r \geq k,
\nonumber
%\label{eq_jointseq_density}
\end{align}
we have
\begin{eqnarray}
\widehat{\gamma}_j & = & \frac { \sum _{n = 1}^K X_n \gamma_n(j)}
{ \delta \cdot \sum _{n = 1}^N \gamma_n(j)}, %\nonumber
\label{eq_gammaj_XnYn} \\
\widehat{\mu}_j & = & \frac{ \sum_{n = 1}^K (Y_n - X_n) \gamma_n(j)}
{\sum_{n = 1}^N Y_n \gamma_n(j)} =
1 - \frac{ \sum_{n = 1}^K  X_n \gamma_n(j)} {\sum_{n = 1}^N Y_n \gamma_n(j)} .
\label{eq_muj_XnYn}
%\nonumber
\end{eqnarray}
The expressions for $\widehat{\gamma}_j$ and $\widehat{\mu}_j$ have also been
derived by~\citet{vasanth_aoas2013}. Note that $X_n \in [0, \hsppp \delta]$ and
$X_n \leq Y_n$ which imply that $0 \leq  \{ \widehat{\gamma}_j , \hsppp
\widehat{\mu}_j \} \leq 1$.

\section{$\left\{ X_n \right\}$ as Observations}
%\noindent {\bf \em \underline{$\left\{ X_n \right\}$ as Observations:}}
With ${\cal O} = \left\{ X_n , \hsppp n = 1, \cdots, K \right\}$, using the binomial density
function, we have
\begin{eqnarray}
\widehat{\gamma}_j & = & \frac{ \sum_{n = 1}^K X_n \gamma_n(j) }{ \delta \cdot
\sum_{n = 1}^K \gamma_n(j)}. %\nonumber
\label{eq_gammaj_Xn}
\end{eqnarray}
Note that the update equation~(\ref{eq_gammaj_Xn}) has the same structure as the
corresponding expression in~(\ref{eq_gammaj_XnYn}). However, the difference between the
two cases is that the update equation for $\gamma_n(j)$ depends on the choice of ${\cal O}$
and the associated density functions.

On the other hand, since $\{ X_n \}$ does not allow inferencing on $\mu_j$, we recall the
facts that
\begin{eqnarray}
{\sf E} [ {\mb M}_i | {\cal H}_j ] & = & \frac{ \gamma_j}{ 1 - \mu_j} \nonumber \\
{\sf Var} ({\mb M}_i | {\cal H}_j ) & = & \frac{ \gamma_j \cdot (1 + \mu_j - \gamma_j) }
{ (1  - \mu_j)^2}, \nonumber
\end{eqnarray}
which allows $\mu_j$ to be rewritten as
\begin{eqnarray}
\mu_j =
\frac{  {\sf Var} ( {\mb M}_i | {\cal H}_j ) + \Big( {\sf E} [{\mb M}_i | {\cal H}_j ] \Big)^2 -
{\sf E} [{\mb M}_i | {\cal H}_j ] }
{ {\sf Var} ({\mb M}_i | {\cal H}_j ) + \Big( {\sf E} [{\mb M}_i | {\cal H}_j ] \Big)^2 +
{\sf E} [{\mb M}_i | {\cal H}_j ] }. \nonumber
\end{eqnarray}
To allow for simple inferencing, we supplant $\widehat{\mu}_j$ (for both $j = 0,1$) with
the sample estimate of $\mu_j$ using all of $\{ {\mb M}_i \}$ (implicitly ignoring the two-state HMM
framework):
\begin{eqnarray}
\widehat{\mu}_j = \frac{ \sum_{i = 1}^N {\mb M}_i({\mb M}_i - 1) }
{ \sum_{i = 1}^N {\mb M}_i({\mb M}_i + 1) }. %\nonumber
\label{eq_muj_Xn}
\end{eqnarray}
It can be seen that $\widehat{\mu}_j$ in~(\ref{eq_muj_Xn}) can be approximated as
\begin{eqnarray}
\widehat{\mu}_j & = & \frac{ \sum_{n = 1}^K \sum_{i \in \Delta_n} {\mb M}_i({\mb M}_i - 1) }
{ \sum_{n = 1}^K \sum_{i \in \Delta_n} {\mb M}_i ({\mb M}_i + 1) }
\nonumber \\
& \stackrel{(a)}{\approx} &
\frac{ \sum_{n = 1}^K \sum_{i \in \Delta_n} {\mb M}_i \cdot \sum_{i \in \Delta_n} \big( {\mb M}_i - 1 \big) }
{ \sum_{n = 1}^K \sum_{i \in \Delta_n} {\mb M}_i \cdot \sum_{i \in \Delta_n} \big( {\mb M}_i + 1 \big) }
\nonumber \\
& \stackrel{(b)}{\approx} &
\frac{ \sum_{n = 1}^K \sum_{i \in \Delta_n} {\mb M}_i \cdot
\sum_{i \in \Delta_n} \big( {\mb M}_i - \indic({\mb M}_i > 0) \big) }
{ \sum_{n = 1}^K \sum_{i \in \Delta_n} {\mb M}_i \cdot
\sum_{i \in \Delta_n} {\mb M}_i}
\nonumber \\
& = & \frac{ \sum_{n = 1}^K  %\sum_{i \in \Delta_n} {\mb M}_i
Y_n \cdot (Y_n - X_n) }
{ \sum_{n = 1}^K %\sum_{i \in \Delta_n} {\mb M}_i \cdot
Y_n^2}
%\label{eq_muj_approx_Xn}
\nonumber \\
& = & 1 - \frac{ \sum_{n = 1}^K Y_n \cdot X_n}
{ \sum_{n = 1}^K Y_n^2}
\nonumber
\end{eqnarray}
where the approximation in (a) assumes that $\sum_k a_k b_k \approx \beta \cdot \hsppp
\sum_k a_k \cdot \hsppp \sum_k b_k$ for an appropriate choice of $\beta$, and the
approximation in (b) replaces $1$ with $\indic({\mb M}_i > 0)$ in the numerator and discards
this factor in the denominator. From the above series of expressions, it can be seen that
the expression for $\widehat{\mu}_j$ in~(\ref{eq_muj_Xn}) approximates~(\ref{eq_muj_XnYn})
with %$\sum_{i \in \Delta_n} M_i$
$Y_n$ capturing $\gamma_n(j)$ (in both states) up to a scaling factor. Also, note that
$\widehat{\mu}_j$ is not an iterative expression depending on the data. It is also important
to note that the price to pay for a lack of estimate of $\mu_j$ from $\{ X_n \}$ is to use
$\{ Y_n \}$ for it in a non-iterative sense.

%\noindent {\bf \em \underline{$\left\{ Y_n \right\}$ as Observations:}}
\section{$\left\{ Y_n \right\}$ as Observations}
With
${\cal O} = \left\{ Y_n , \hsppp n = 1, \cdots, K \right\}$, we now obtain simplified
expressions for $\widehat{\gamma}_j$ and $\widehat{\mu}_j$ under the assumptions that
$\delta \gg 1$ and $\mu_j > \gamma_j > \frac{1}{\delta}$.

Recall that the density function of $Y_n$ is given as
\begin{eqnarray}
{\sf P} \left(Y_n = r \Big|  \Sb_{ i} |_{ i \in \Delta_n } = j \right) =
(1 - \gamma_j)^{\delta} \cdot (\mu_j)^{r }\cdot
\sum_{ k = 1}^{\min(r, \hsppp \delta)} {\delta \choose k} { r -1 \choose r - k} \cdot
{\sf A}^k
%\label{eq_Yn}
\nonumber
\end{eqnarray}
with ${\sf A} = \frac{ (1 - \mu_j) \hspp \gamma_j } { (1 - \gamma_j) \hspp \mu_j }$.
Using the moment generating function of a hypergeometric distribution\footnote{A
hypergeometric distribution with parameters $(N_1, \hsppp K_1, \hsppp n_1)$ captures
the number of successes in $n_1$ trials of an experiment from a population of size
$N_1$ with $K_1$ elements of one type and $N_1 - K_1$ of another type.} with parameters
$(N_1, \hsppp K_1, \hsppp n_1)$ where $N_1 = \delta + r, K_1 = \delta, n_1 = r$, we have
\begin{eqnarray}
\sum_{k = 0}^{ \min(r, \hsppp \delta) }
{\delta \choose k} { r \choose r - k} \cdot {\sf A}^k =
{}_2 F_1\left(-r, \hsppp -\delta \hsppp; \hsppp 1 \hsppp; \hsppp {\sf A} \right),
%\nonumber
\label{eq_interm1}
\end{eqnarray}
with ${}_2 F_1\left(a, \hsppp b \hsppp; \hsppp c \hsppp; \hsppp z \right)$ denoting the
Gauss hypergeometric function~(see the definition of this function in~\citet[15.1.1, p.\ 556]{abramowitz}).
Note that the right-hand
side of~(\ref{eq_interm1}) is well-defined as a power series since $\mu_j > \gamma_j$ which
implies that ${\sf A} < 1$ (and hence within the radius of convergence). Similarly, with
$N_1 = r + \delta - 1, K_1 = \delta, n_1 = r - 1$, we have
\begin{eqnarray}
\sum_{k = 0}^{ \min(r-1, \hsppp \delta) }
{\delta \choose k} { r - 1 \choose r - 1 - k} \cdot {\sf A}^k =
{}_2 F_1\left(- (r-1), \hsppp -\delta \hsppp; \hsppp 1 \hsppp; \hsppp {\sf A} \right).
%\nonumber
\label{eq_interm2}
\end{eqnarray}
Combining~(\ref{eq_interm1}) and~(\ref{eq_interm2}), we have
\begin{eqnarray}
\frac{ {\sf P} \left(Y_n = r \Big|  \Sb_{ i} |_{ i \in \Delta_n } = j \right)}
{ (1 - \gamma_j)^{\delta} \cdot (\mu_j)^{r } }
& = &
%\Big(
{}_2 F_1\left(-r, \hsppp -\delta \hsppp ; \hsppp 1 \hsppp; \hsppp {\sf A} \right)
- {}_2 F_1\left(- (r-1), \hsppp -\delta \hsppp; \hsppp 1 \hsppp; \hsppp {\sf A} \right)
%\Big)
\nonumber \\
& \stackrel{(c)}{=} & \delta \hsppp {\sf A} \cdot
{}_2 F_1\left(- (r-1), \hsppp - (\delta-1) \hsppp; \hsppp 2 \hsppp; \hsppp {\sf A} \right)
\nonumber
\end{eqnarray}
where (c) follows from a straightforward application of the definition of the hypergeometric
function that allows the following simplification:
\begin{eqnarray}
{}_2 F_1\left(a+1, \hsppp b \hsppp ; \hsppp
c \hsppp ; \hsppp z \right) - {}_2 F_1\left(a, \hsppp b \hsppp ; \hsppp
c \hsppp ; \hsppp z \right) = \frac{bz}{c} \cdot
{}_2 F_1\left(a+1, \hsppp b+1 \hsppp ; \hsppp c +1 \hsppp ; \hsppp z \right).
\nonumber
\end{eqnarray}

Plugging in the expression for the density function of $Y_n$ in the auxiliary function, we have
\begin{eqnarray}
\frac{
Q(\boldsymbol\lambda, \hsppp {\boldsymbol{\bar{\lambda}}})
\Big|_{ {\sf Obs.} \hspp {\sf density}}
}  { {\sf P}( {\cal O} | {\boldsymbol{\bar{\lambda}}} ) }
& = &
\sum _{n = 1}^K \sum_ {j = 0 }^1
\Big[ \delta \log(1 - \gamma_j) + Y_n \log( \mu_j) + \log(\delta) + \log({\sf A})
\nonumber \\
& & {\hspace{0.3in}} + \log \Big( {}_2 F_1\left( -(Y_n - 1), \hsppp - (\delta - 1) \hsppp ; \hsppp 2 \hsppp ;
\hsppp {\sf A} \right) \Big) \Big] \cdot \gamma_n(j).
\nonumber
\end{eqnarray}
Setting the derivative of the auxiliary function (with respect to $\gamma_j$) to zero, we have
\begin{eqnarray}
& & {\hspace{-0.3in}}
\frac{\delta \sum_{n = 1}^K \gamma_n(j)}{ 1 - \gamma_j}
\nonumber \\
& =  &
\frac{d {\sf A}}{d \gamma_j} \cdot
\left[ \frac{ \sum_{n = 1}^K \gamma_n(j) }{ {\sf A} } +
\sum_{n = 1}^K \frac{ \gamma_n(j) \cdot
\frac{d}{d{\sf A}} {}_2 F_1\left( -(Y_n - 1), \hsppp - (\delta - 1) \hsppp ; \hsppp 2 \hsppp ;
\hsppp {\sf A} \right)
}
{ {}_2 F_1\left( -(Y_n - 1), \hsppp - (\delta - 1) \hsppp ; \hsppp 2 \hsppp ; \hsppp {\sf A} \right) }
 \right]
\nonumber \\
&  \stackrel{(d)}{=} &
\frac{d {\sf A}}{d \gamma_j} \cdot \left[
\frac{ \sum_{n = 1}^K \gamma_n(j) }{ {\sf A} } +
\sum_{n = 1}^K \frac{ \gamma_n(j)
\cdot {}_2 F_1\left( -(Y_n - 2), \hsppp - (\delta - 2) \hsppp ; \hsppp 3 \hsppp ; \hsppp {\sf A} \right)
}
{ {}_2 F_1\left( -(Y_n - 1), \hsppp - (\delta - 1) \hsppp ; \hsppp 2 \hsppp ; \hsppp {\sf A} \right) }
 \right]
\nonumber \\
& \stackrel{(e)}{=} &
\frac{d {\sf A}}{d \gamma_j} \cdot \left[
\frac{ \sum_{n = 1}^K \gamma_n(j) }{ {\sf A} } +
\frac{1}{ 1 - {\sf A}} \cdot
\sum_{n = 1}^K \frac{ \gamma_n(j)
\cdot
{}_2 F_1\left( Y_n + 1, \hsppp \delta + 1 \hsppp ; \hsppp 3 \hsppp ; \hsppp {\sf A} \right)
}
{ {}_2 F_1\left( Y_n + 1, \hsppp \delta + 1 \hsppp ; \hsppp 2 \hsppp ; \hsppp {\sf A} \right) }
\right]
\label{eq_hyper1}
\end{eqnarray}
where (d) follows from the fact in~\citet[15.2.1, p.\ 557]{abramowitz} that
%~\cite[2.1.2(7), p.\ 58]{magnus_oberhettinger} that
\begin{eqnarray}
\frac{d}{dz}{}_2 F_1\left( a, \hsppp b \hsppp ; \hsppp c \hsppp ; \hsppp z \right) =
\frac{ab}{c} \cdot
{}_2 F_1\left( a+1, \hsppp b+1 \hsppp ; \hsppp c+1 \hsppp ; \hsppp z \right), \hspp z < 1
\nonumber
\end{eqnarray}
and (e) from the fact in~\citet[15.3.3, p.\ 559]{abramowitz} that
\begin{eqnarray}
{}_2 F_1\left( a, \hsppp b \hsppp ; \hsppp c \hsppp ; \hsppp z \right)
= (1 - z)^{c-a-b} \cdot
{}_2 F_1\left( c-a, \hsppp c-b \hsppp ; \hsppp c \hsppp ; \hsppp z \right), \hspp z < 1 .
\nonumber
\end{eqnarray}
Similarly, setting the derivative of the auxiliary function (with respect to $\mu_j$)
to zero, we have
\begin{equation}
\begin{split}
\frac{ \sum \limits_{n = 1}^K Y_n \gamma_n(j) }{ \mu_j} =
-\frac{d {\sf A}}{d \mu_j} \cdot \left[
\frac{ \sum_{n = 1}^K \gamma_n(j) }{ {\sf A} } +
%\sum_{n = 1}^K \frac{ \gamma_n(j)
%\cdot {}_2 F_1\left( -(Y_n - 2), \hsppp - (\delta - 2) \hsppp ; \hsppp 3 \hsppp ; \hsppp {\sf A} \right)
%}
%{ {}_2 F_1\left( -(Y_n - 1), \hsppp - (\delta - 1) \hsppp ; \hsppp 2 \hsppp ; \hsppp {\sf A} \right) }
\frac{1}{ 1 - {\sf A}} \cdot
\sum_{n = 1}^K \frac{ \gamma_n(j)
\cdot
{}_2 F_1\left( Y_n + 1, \hsppp \delta + 1 \hsppp ; \hsppp 3 \hsppp ; \hsppp {\sf A} \right)
}
{ {}_2 F_1\left( Y_n + 1, \hsppp \delta + 1 \hsppp ; \hsppp 2 \hsppp ; \hsppp {\sf A} \right) }
\right]. %\nonumber
\label{eq_hyper2}
\end{split}
\end{equation}
From the two derivative expressions in~(\ref{eq_hyper1}) and~(\ref{eq_hyper2}), we
clearly have
\begin{eqnarray}
{\mu}_j = 1 - \frac{ {\gamma}_j \cdot \delta \cdot \sum_{n = 1}^K \gamma_n(j) }
{ \sum_{n = 1}^K Y_n \gamma_n(j) }.
%\nonumber
\label{eq_hyper3}
\end{eqnarray}

We now simplify the expressions for $\gamma_j$ and $\mu_j$ under the assumptions
that $\delta \gg 1$ and $\gamma_j > \frac{1}{\delta}$. For this\footnote{We use the
notation $f(x) \stackrel{x \gg 1} {\sim} g(x)$ to denote that $\lim_{x \rightarrow
\infty} \frac{f(x)}{g(x)} = 1.$ Further, we use the notation
$\lim_{x \rightarrow a} f(x) = O(g(x))$ if there exists $\epsilon$ and $M$
such that $|f(x)| \leq M |g(x)|$ for all $x$ such that $|x - a| < \epsilon$. If a
choice of $a$ is not specified, it is implicitly assumed to be $+\infty$.}, we use
the following fact from~\citet[2.3.2(13), p.\ 77]{magnus_oberhettinger}
\begin{eqnarray}
{}_2 F_1\left( a, \hsppp b \hsppp ; \hsppp c \hsppp ; \hsppp z \right)
& \stackrel%{ |b| \rightarrow \infty}
{ |b| \gg 1} {\sim } &
{}_1 F_1\left( a \hsppp ; \hsppp c \hsppp ; \hsppp bz \right)
\cdot \left[ 1 + O \left( |b|^{-1} \right) \right]
\nonumber \\
%& \stackrel{ |b| \rightarrow \infty}{\sim } &
& \stackrel%{ |b| \rightarrow \infty}
{ |b| \gg 1} {\sim } & \frac{ \Gamma(c) e^{bz} (bz)^{a-c}}
{ \Gamma(a)} \cdot \left[
1 + \frac{1 - a}{bz} + \frac{ (1-a) (2-a) (c-a) (c-a+1)}{2  b^2z^2}
\right]
\nonumber
\end{eqnarray}
where ${}_1 F_1\left( \cdot ; \cdot ; \cdot \right)$ is the confluent hypergeometric
function of the first kind (see the definition in~\citet[13.1.2, p.\ 505]{abramowitz}) and the second step
follows from~\citet[Sec.\ 4, (9)-(11)]{macdonald_tech_rep}. Applying these facts
to~(\ref{eq_hyper1}), we have
\begin{eqnarray}
\frac{
{}_2 F_1\left( Y_n + 1, \hsppp \delta + 1 \hsppp ; \hsppp 3 \hsppp ; \hsppp {\sf A} \right)
}
{ {}_2 F_1\left( Y_n + 1, \hsppp \delta + 1 \hsppp ; \hsppp 2 \hsppp ; \hsppp {\sf A} \right) }
\stackrel{ \delta \gg 1}{\sim} \frac{2}{(\delta+1){\sf A}} \cdot
\left[ 1 +  O \left( \frac{1}{\delta} %+ 1}
\right) \right] \nonumber
\end{eqnarray}
and thus
\begin{eqnarray}
\frac{\delta \cdot \sum_{n = 1}^K \gamma_n(j)}{ 1 - \gamma_j}
& \stackrel{ \delta \gg 1}{\sim} &
\frac{d {\sf A}}{d \gamma_j} \cdot \sum_{n = 1}^K \gamma_n(j)
\cdot \left[ \frac{ 1 }{ {\sf A} } + \frac{2}{ (1 - {\sf A}) {\sf A} (\delta+1)}
\right]
\nonumber \\
\Longrightarrow \mu_j & \stackrel{ \delta \gg 1}{\sim} &
\frac{ \gamma_j(\delta+1) ( \delta \gamma_j -1) }
{ (\delta+1) ( \delta \gamma_j -1) - 2 (1 - \gamma_j) }
\nonumber \\
& = & \frac{ \gamma_j }{ 1 - \frac{ 2 (1 - \gamma_j)}{ (\delta \gamma_j - 1)(\delta + 1)}},
\label{eq_hyper4}
\end{eqnarray}
where the condition that $\gamma_j > \frac{1}{\delta}$ ensures that $\mu_j > \gamma_j$.

Combining~(\ref{eq_hyper3}) with~(\ref{eq_hyper4}), it can be seen that $\gamma_j$ is
a solution to the quadratic equation:
\begin{eqnarray}
& & \gamma_j^2 \cdot \left[ (\delta^2 + \delta + 2) \delta \sum_{n = 1}^K \gamma_{n,j}
+ (\delta^2 + \delta) \sum_{n = 1}^K Y_n \gamma_{n,j} \right]
\nonumber \\
& &  \hsp -
\gamma_j \cdot \left[ \delta(\delta + 3)\sum_{n = 1}^K \gamma_{n,j}
+ (\delta^2 + 2 \delta + 3) \sum_{n = 1}^K Y_n \gamma_{n,j} \right] +
(\delta + 3) \sum_{n = 1}^K Y_n \gamma_{n,j} = 0.
\nonumber
\end{eqnarray}
The two solutions to the above quadratic equation in the $\delta \gg 1$ regime are:
\begin{eqnarray}
{\sf Solution \hspp 1} &: &
\widehat{\gamma}_j =
\frac{1}{\delta} \left(1 + \frac{2}{\delta} \right) + O \left( \frac{1}{\delta^3} \right),
\hspp \hspp
\nonumber \\
&  &
\widehat{\mu}_j = 1 - \frac{ \sum_{n = 1}^K \gamma_{n,j} }
{ \sum_{n = 1}^K Y_n \gamma_{n,j} } \cdot \left( 1 + \frac{2}{\delta} \right)
+ O \left( \frac{1}{\delta^2} \right) ,
\nonumber \\
{\sf Solution \hspp 2} &: &
\widehat{\gamma}_j = \frac{ \sum_{n = 1}^K Y_n \gamma_{n,j} }
{ \delta \sum_{n = 1}^K \gamma_{n,j} } \cdot
\left(1 -  \frac{ \sum_{n = 1}^K Y_n \gamma_{n,j} } { \delta \sum_{n = 1}^K \gamma_{n,j} } \right)
+ O \left( \frac{1}{\delta^3} \right) , \hspp \hspp
\nonumber \\
& & \widehat{\mu}_j = \frac{ \sum_{n = 1}^K Y_n \gamma_{n,j} }
{ \delta \sum_{n = 1}^K \gamma_{n,j} } + O \left( \frac{1}{\delta^2} \right) .
\nonumber
\end{eqnarray}
To ensure that the three conditions ($\widehat{\gamma}_j > \frac{1}{\delta}$,
$\widehat{\mu}_j > \widehat{\gamma}_j$ and $\widehat{\mu}_j < 1$) are met, we note
that ${\sf Solution \hspp 1 \hsppp}$ needs to satisfy the condition that
$\frac{ \sum_{n = 1}^K \gamma_{n,j} } { \sum_{n = 1}^K Y_n \gamma_{n,j} } <
1 - \frac{1}{\delta}$, whereas ${\sf Solution \hspp 2 \hsppp}$ needs to satisfy
the condition $\frac{1}{\delta} < \frac{ \sum_{n = 1}^K \gamma_{n,j} }
{ \sum_{n = 1}^K Y_n \gamma_{n,j} } < 1$. While either condition does not appear
to suffer from any technical difficulties in the $\delta \gg 1$ regime, it is clear
that ${\sf Solution \hspp 2 \hsppp}$ converges to a geometric model $(\gamma_j =
\mu_j)$ since
\begin{eqnarray}
\frac{ \widehat{\gamma}_j } { \widehat{\mu}_j } \stackrel{ \delta \gg 1} {\sim} 1.
\nonumber
\end{eqnarray}
This solution particularizes the hurdle-based geometric model and thus reduces the
general model to a special case. Thus, we use ${\sf Solution \hspp 1 \hsppp}$ for
the update equations with ${\cal O} = \{ Y_n \}$ as observations. Note that
the structure of ${\sf Solution \hspp 1 \hsppp}$ follows the same general structure
as~(\ref{eq_gammaj_XnYn})-(\ref{eq_muj_XnYn}) with $X_n = 1 + \frac{2}{\delta}$. However,
as before, the update equation for $\gamma_n(j)$ depends on the choice of ${\cal O}$
and the associated density functions.

%It can be easily seen that
%\begin{eqnarray}
%\frac{
%{}_2 F_1\left( Y_n + 1, \hsppp \delta + 1 \hsppp ; \hsppp 3 \hsppp ; \hsppp {\sf A} \right)
%}
%{ {}_2 F_1\left( Y_n + 1, \hsppp \delta + 1 \hsppp ; \hsppp 2 \hsppp ; \hsppp {\sf A} \right) }
%=
%\frac{ \sum_{\ell = 0}^{\infty} \frac{2}{2+\ell} \cdot \frac{ {\sf A}^\ell}{\ell!}
%\cdot \frac{ \Gamma(Y_n + 1 + \ell) \cdot \Gamma(\delta + 1 + \ell) }
%{ \Gamma(2 + \ell)} }
%{
%\sum_{\ell = 0}^{\infty} \frac{ {\sf A}^\ell}{\ell!}
%\cdot \frac{ \Gamma(Y_n + 1 + \ell) \cdot \Gamma(\delta + 1 + \ell) }
%{ \Gamma(2 + \ell)} }
%< 1
%\nonumber
%\end{eqnarray}

%Clearly, the complicated dependency of ${\sf P} \left(Y_n = k \Big|  \Sb_{ i}
%|_{ i \in \Delta_n } = j \right)$ on the parameters $\mu_j$ and $\gamma_j$ that renders a
%product decomposition (for parameter learning) difficult, we are left

\begin{supplement} [id=suppB]
\label{sec_suppb}
\stitle{Supplementary B: Background on majorization theory}
\slink[doi]{}%{10.1214/10-AOAS395SUPPC}
\sdatatype{.pdf}
\sdescription{This section provides a brief primer on majorization theory and
reverse majorization.}
\end{supplement}

\section{Preliminaries}
We refer the readers to the seminal book by~\citet{olkin} for a comprehensive
background on majorization theory. Here, we provide a brief review of the main
theoretical underpinning needed to develop this paper.

Let %${\mathbb P}(\delta)$
${\mathbb P}_{\delta}$ denote the space of probability vectors of length $\delta$
with $\underline{\boldsymbol P} = \left[ {\mb P}(1), \cdots, {\mb P}(\delta) \right] \in
{\mathbb P}_{\delta} \Longrightarrow {\mb P}(i) \geq 0$ for all $i = 1, \cdots, \delta$ and
$\sum_i {\mb P}(i) = 1$. Without loss in generality, we can assume that the entries of
$\underline {\boldsymbol P}$ are arranged in non-increasing order (${\mb P}(1) \geq \cdots
\geq {\mb P}(\delta)$).
\begin{defn}[Majorization]
Let $\{ \underline {\boldsymbol P}, \hsppp \underline {\boldsymbol Q} \}
\in {\mathbb P}_{\delta}$. We say that $\underline {\boldsymbol P}$ is {\em majorized}
by $\underline {\boldsymbol Q}$ and denote it as $\underline{\boldsymbol P}
\prec \underline{\boldsymbol Q}$ if
\begin{eqnarray}
\sum_{i=1}^k {\mb P}(i) \leq \sum_{i=1}^k {\mb Q}(i), \hspp %1 \leq k \leq \delta.
\hspp k = 1, \cdots, \delta.
\label{equality}
\end{eqnarray}
Note that equality holds in~(\ref{equality}) for $k = \delta$ because
%of the assumption of probability vectors.
$\{ \underline {\boldsymbol P}, \hsppp \underline {\boldsymbol Q} \}
\in {\mathbb P}_{\delta}$, which implies that $\sum_i {\mb P}(i) = 1 = \sum_i {\mb Q}(i)$.
\flushright
\qedsymbol
%\endproof
%\qedhere
\end{defn}
The majorization relationship captures the notion that $\underline {\boldsymbol P}$ is
more well-spread out than $\underline {\boldsymbol Q}$. It also vaguely captures the notion
that $\underline {\boldsymbol P}$ is more unambiguously random/bursty than $\underline {\boldsymbol Q}$.
We now provide many illustrative examples of majorization. In the first example, as $k$ decreases
from $\delta$ to $1$, we have a progressive majorization relationship:
\begin{eqnarray}
\Big[
\underbrace{ 1/\delta , \cdots, 1/\delta } _{ \delta \hsppp
\hsppp {\rm times} } \Big]
%\prec \Big[ \underbrace{
%\frac{1}{\delta -1 } , \cdots, \frac{1}{\delta-1} ,} _{ (\delta - 1)
%\hsppp \hsppp {\rm times} } \hsppp 0  \Big]
\prec \cdots \prec
\Big[ \underbrace{
1/k , \cdots, 1/k ,} _{ k \hsppp \hsppp {\rm times} }
\hsppp \underbrace{ 0, \cdots, 0 }_{ (\delta - k) \hsppp
\hsppp {\rm times} } \Big] \prec \cdots \prec
\Big[  %\underbrace{ 1}_{ 1 \hsppp \hsppp {\rm time}}, \hsppp
1, \hsppp
\underbrace{  0, \cdots, 0}_{ (\delta - 1) \hsppp \hsppp {\rm times} }  \Big].
\nonumber
\end{eqnarray}
On the other hand, any $\underline {\boldsymbol P} \in {\mathbb P}_{\delta}$ satisfies:
\begin{eqnarray}
\Big[
\underbrace{ 1/\delta , \cdots, 1/\delta } _{ \delta \hsppp
\hsppp {\rm times} } \Big] \prec \underline {\boldsymbol P} \prec
\Big[  %\underbrace{ 1}_{ 1 \hsppp \hsppp {\rm time}}, \hsppp
1, \hsppp
\underbrace{  0, \cdots, 0}_{ (\delta - 1) \hsppp \hsppp {\rm times} }  \Big].
\nonumber
\end{eqnarray}
In this sense, any vector $\underline {\boldsymbol P}$ majorizes
$\Big[ \underbrace{ 1/\delta , \cdots, 1/\delta } _{ \delta \hsppp
\hsppp {\rm times} } \Big]$ and is majorized by $\Big[ 1, \hsppp
\underbrace{  0, \cdots, 0}_{ (\delta - 1) \hsppp \hsppp {\rm times} }  \Big]$.
An easy consequence of the above majorization relationship is that for any $c \geq 0$
and any $\underline {\boldsymbol P} \in {\mathbb P}_{\delta}$, we have
\begin{eqnarray}
\frac{ \left[ {\mb P}(1) + c, \hspp  \cdots , \hspp {\mb P}(\delta) + c
\right]}{1 + \delta c} \prec \underline {\boldsymbol P}.
\nonumber
\end{eqnarray}
In the $\delta = 2$ case, we have
\begin{eqnarray}
\Big[ {\mb P}(1) , \hspp 1 - {\mb P}(1) \Big] \prec
\Big[ {\mb Q}(1) , \hspp 1 - {\mb Q}(1) \Big]
\Longleftrightarrow
\frac{1}{2} \leq {\mb P}(1) \leq {\mb Q}(1).
%\nonumber
\label{eq_12}
\end{eqnarray}
While a similar set of equivalent inequalities on the entries of $\underline {\boldsymbol P}$
and $\underline {\boldsymbol Q}$ can be written for the $\delta \geq 3$ case, they quickly
get overwhelmingly complicated.

\begin{defn}[Schur-convex and -concave functions]
A function $f \hsppp : \hsppp \left( {\mathbb{R}}^+ \right) ^{\delta} \mapsto {\mathbb{R}}$
is said to be {\em Schur-convex} if for any $\underline{\boldsymbol P}$ and
$\underline{\boldsymbol Q}$ with $\underline{\boldsymbol P} \prec
\underline{\boldsymbol Q}$, we have $f(\underline{\boldsymbol P}) \leq
f(\underline{\boldsymbol Q})$. A function $f(\cdot)$ is Schur-concave if
$-f(\cdot)$ is Schur-convex. That is, $\underline{\boldsymbol P} \prec
\underline{\boldsymbol Q}$ implies that $f(\underline{\boldsymbol P}) \geq
f(\underline{\boldsymbol Q})$.
\flushright \qedsymbol
%\endproof
\end{defn}
We now provide some examples of Schur-convex and Schur-concave functions.
\begin{prop}
\label{prop_majorize}
The counting function of non-zero elements in $\underline{\boldsymbol P}$
(also called the rank function), defined as,
\begin{eqnarray}
{\sf NZ}( \underline{\boldsymbol P} )  \triangleq
\sum_i \indic \Big( {\mb P}(i) > 0 \Big)
\nonumber
\end{eqnarray}
is Schur-concave. If ${\mb P}(i) = 0$, let $-{\mb P}(i) \log( {\mb P}(i))$ be
extended continuously to $0$ and ${\mb P}(i)^{\alpha}$ be extended continuously to
$0$ if $\alpha > 0$ and to $+ \infty$ if $\alpha < 0$. Then, the Shannon entropy and
geometric mean functions, defined respectively as,
\begin{eqnarray}
{\sf SE}( \underline{\boldsymbol P} )  \triangleq  - \sum_i {\mb P}(i) \log( {\mb P}(i)) ,
%\nonumber \\
& &
{\sf GM}( \underline{\boldsymbol P} )  \triangleq  \left( \prod_i {\mb P}(i) \right)
^{ 1/\delta } \nonumber
\end{eqnarray}
are also Schur-concave. %on ${\cal P}(\delta)$.
The power mean function corresponding to an index $\alpha$, defined as,
\begin{eqnarray}
{\sf PM}( \underline{\boldsymbol P}, \hsppp \alpha) \triangleq
\left( \frac{ \sum_i {\mb P}(i) ^{\alpha} }{ \sum_i \indic \left( {\mb P}(i) > 0 \right) }
\right)^{1/\alpha}
\nonumber
\end{eqnarray}
is Schur-convex %in ${\cal P}(\delta)$
if $\alpha \geq 1$ and Schur-concave if $\alpha \leq 1, \hsppp \alpha \neq 0$.
\end{prop}
\begin{proof}
To see that ${\sf NZ}( \underline{\boldsymbol P} )$ is Schur-concave,
assume that $\underline{\boldsymbol P} \prec \underline{\boldsymbol Q}$ and let
\begin{eqnarray}
\underline {\boldsymbol Q} = \left[ {\mb Q}(1), \cdots, {\mb Q}(p), 0 , \cdots , 0 \right]
\nonumber
\end{eqnarray}
with ${\mb Q}(p) > 0$ for some $p$. A rewriting of the condition in~(\ref{equality}) is:
\begin{eqnarray}
\sum_{i = k}^{\delta} {\mb P}(i) \geq \sum_{i = k}^{\delta} {\mb Q}(i), \hspp
\hspp k = 1, \cdots, \delta.
%1 \leq k \leq \delta.
%\nonumber
\label{rewrite}
\end{eqnarray}
With $k = p$ in~(\ref{rewrite}), we have $\sum_{i = p}^{\delta} {\mb P}(i)
\geq {\mb Q}(p) > 0$. We have a contradiction if ${\mb P}(p) = 0$ since $\{ {\mb P}(p),
\cdots, {\mb P}(\delta) \}$ are arranged in non-increasing order and all
of them have to be $0$. Thus, ${\mb P}(p) > 0$ and this implies that
\begin{eqnarray}
\sum_{i=1}^{\delta} \indic \big( {\mb P}(i) > 0 \big) \geq
\sum_{i=1}^{\delta} \indic \big( {\mb Q}(i) > 0 \big). \nonumber
\end{eqnarray}

%By restricting attention to the subset of indices with non-zero entries,
%without loss in generality, we can assume that ${\sf M}_{\delta} > 0$.
%%$\underline{\boldsymbol M} \in {\cal P}(\delta)$.
%For $\underline{\boldsymbol N}$, we begin by considering the setting
%where %both $\{ \underline{\boldsymbol M},
%%\hspp \underline{\boldsymbol N}  \in {\cal P}(\delta)$.
%${\sf N}_{\delta} > 0$.
The proof of the Schur-convexity or -concavity of the different functional
structures in the statement of the proposition follow from the main result
from~\citet[Prop.\ 3.C.1, p.\ 64]{olkin} that if $g \hsppp : \hsppp (0, \hsppp \infty)
\mapsto {\mathbb{R}}$ is convex (or concave), then $\underline{\boldsymbol P}
\mapsto \sum_i g({\mb P}(i) )$ is Schur-convex (or Schur-concave). In the setting
where $\{ \underline{\boldsymbol P}, \hspp \underline{\boldsymbol Q} \}
\in {\mathbb P}_{\delta}$, but with some zero entries, all the inequality relations
corresponding to Schur-convexity and -concavity hold trivially with the appropriate
continuous extensions. %to the $N_i = 0$ case.
%\flushright \qedsymbol
\end{proof}

\begin{cor}
A straightforward consequence of Prop.~\ref{prop_majorize} is that the
{\em normalized power mean}, %corresponding to an index $\alpha$,
defined as,
\begin{eqnarray}
{\sf NPM} \left( \underline{\boldsymbol P}, \hsppp
\alpha \right) \triangleq \frac{ {\sf PM} \left( \underline{\boldsymbol P}
, \hsppp \alpha \right) }{ {\sf NZ}
\left( \underline{\boldsymbol P} \right) }
\nonumber
\end{eqnarray}
is Schur-convex if $\alpha > 1$. Note that Schur-concavity of the normalized
power mean may not hold if $\alpha < 1$ except in the trivial case where
${\sf NZ} \left( \underline{\boldsymbol P} \right) =
{\sf NZ} \left( \underline{\boldsymbol Q} \right)$.
\flushright \qedsymbol
\end{cor}

From the example in~(\ref{eq_12}), it is clear that majorization theory provides a
complete ordering (all vectors are comparable with each other) in the $\delta = 2$ case.
However, for $\delta \geq 3$, it is important to observe that majorization theory provides
only a partial ordering. For example, it can be seen that both
$\underline{\boldsymbol P} \nprec \underline{\boldsymbol Q}$
and $\underline{\boldsymbol Q} \nprec \underline{\boldsymbol P}$ are true
with the choice $\underline{\boldsymbol P} = \left[ 0.5, \hsppp 0.25, \hsppp 0.25 \right]$
and $\underline{\boldsymbol Q} = \left[ 0.4, \hsppp 0.4, \hsppp 0.2 \right]$. Another
such choice with $\delta = 4$ is
$\underline{\boldsymbol P} = \left[ 0.4, \hsppp 0.35, \hsppp 0.15, \hsppp 0.1 \right]$ and
$\underline{\boldsymbol Q} = \left[ 0.45, \hsppp 0.27, \hsppp 0.25,
\hsppp 0.03 \right]$.
%neither
%$\underline{\boldsymbol M} \prec \underline{\boldsymbol N}$ nor
%$\underline{\boldsymbol N} \prec \underline{\boldsymbol M}$ is true.
Thus, two arbitrary probability vectors in ${\mathbb P}_{\delta}$ cannot necessarily
be compared by a majorization relationship. Further, while Schur-convexity
and -concavity allow an ordering of vectors from ${\mathbb P}_{\delta}$ to ${\mathbb R}$,
we seek a reverse majorization theory where $f(\underline{\boldsymbol P})
\leq f(\underline{\boldsymbol Q})$ for an appropriate choice of $f(\cdot)$
implies that $\underline{\boldsymbol P} \prec \underline{\boldsymbol Q}$.

\section{Reverse Majorization}
The notion of reverse majorization is established over a bigger subset of ${\mathbb{P}}_{\delta}$
by extending (or ``lifting'') the majorization relationship to that of {\em catalytic
majorization}. The idea was first proposed by~\citet{jonathan_plenio} in the context of
quantum entanglement. Various terms such as {\em trumping}, {\em entanglement catalysis},
{\em entanglement-assisted local transformation}, etc., are used in the literature to
describe it.

\begin{defn}[Catalytic majorization]
Let $\{ \underline {\boldsymbol P}, \hsppp \underline {\boldsymbol Q} \}
\in {\mathbb P}_{\delta}$. We say that $\underline {\boldsymbol P}$ is
catalytically majorized by $\underline {\boldsymbol Q}$ if there exists some
$m \geq 1$ and some $\underline {\boldsymbol L} \in {\mathbb P}_m$ such that
\begin{eqnarray}
\underline {\boldsymbol P} \otimes \underline {\boldsymbol L} \prec
\underline {\boldsymbol Q} \otimes \underline {\boldsymbol L},
\label{catalytic}
%\nonumber
\end{eqnarray}
where $\otimes$ denotes the Kronecker product operation:
\begin{eqnarray}
\underline{\boldsymbol P} \otimes \underline{\boldsymbol L}
& \triangleq & \left[ {\mb P}(1) {\mb L}(1), \cdots, {\mb P}(1) {\mb L}(m), \hspp
{\mb P}(2) {\mb L}(1), \cdots, {\mb P}(2) {\mb L}(m), \hspp \cdots , \hspp
\right. \nonumber \\
& & {\hspace{1.0in}} \left.
{\mb P}(\delta) {\mb L}(1) , \cdots, {\mb P}(\delta) {\mb L}(m) \right].
\nonumber
\end{eqnarray}
Note that without loss in generality, $\underline{\boldsymbol L}$ can be assumed
to satisfy ${\mb L}(m) > 0$.
\flushright \qedsymbol
%\endproof
\end{defn}

Technically speaking, catalytic majorization is tensor product-induced majorization,
where $\underline{\boldsymbol L}$ can be seen as a ``resource that allows one to transform
$\underline{\boldsymbol P}$ (a certain state) into $\underline{\boldsymbol Q}$ (another
state) via local operations and classical information; this vector $\underline{\boldsymbol L}$
remains unaltered after being used, yet the transformation could not occur without
its presence'' (the above explanation is sourced {\em verbatim}
from~\citep[p.\ 113]{plosker_phd}). The distinction between majorization and
catalytic majorization comes from the fact that while most majorization results hold even if
the components of the vectors are negative, almost all of the catalytic majorization results
critically depend on the non-negativity of the vector components. Note
that the $\delta m$ inequality relations corresponding to~(\ref{equality})
need to be checked to verify $\underline{\boldsymbol P} \otimes
\underline{\boldsymbol L} \prec \underline{\boldsymbol Q} \otimes
\underline{\boldsymbol L}$ {\em after} reordering the entries of
$\underline{\boldsymbol P} \otimes \underline{\boldsymbol L}$ and
$\underline{\boldsymbol Q} \otimes \underline{\boldsymbol L}$ in non-increasing
order. %It can also be checked that $\sum_{i = 1}^{\delta} \sum_{j = 1}^m
%{\sf M}_i {\sf P}_j = \sum_i {\sf M}_i \cdot \sum_j {\sf P}_j = 1 = \sum_i {\sf N}_i
%\cdot \sum_j {\sf P}_j = \sum_{i = 1}^{\delta} \sum_{j = 1}^m {\sf N}_i {\sf P}_j$.
Further, no specific
conditions are imposed on the length $m$ of $\underline{\boldsymbol L}$
nor on the uniqueness of $\underline{\boldsymbol L}$. Without reference
to $\underline {\boldsymbol L}$, we denote the relationship in~(\ref{catalytic})
as $\underline {\boldsymbol P} \prec_T \underline {\boldsymbol Q}$, with
$T$ standing for ``trumping.''

The following result shows that $\prec_T$ is also not a complete ordering on
${\mathbb P}_{\delta}$. Nevertheless, the set of vectors that can be catalytically
majorized is strictly larger than the set that can be majorized. %~\cite{daftuar,kribs}.

\begin{prop}
\label{prop_bigger_set}
(a) Clearly, if $\underline{\boldsymbol P} \prec \underline{\boldsymbol Q}$,
then $\underline{\boldsymbol P} \prec_T \underline{\boldsymbol Q}$ since $m = 1$
and $\underline{\boldsymbol L} = \big[1  \big]$ can be used to establish catalytic
majorization. But more generally, $\underline{\boldsymbol P} \prec \underline{\boldsymbol Q}$
implies that $\underline{\boldsymbol P} \prec_T \underline{\boldsymbol Q}$
for any choice of $m$ and for any $\underline{\boldsymbol L} \in {\mathbb P}_{m}$.

(b) In converse, if $\underline{\boldsymbol P} \prec_T \underline{\boldsymbol Q}$
for some $\underline{\boldsymbol L} \in {\mathbb P}_{m}$ and $\delta \leq 3$, then
$\underline{\boldsymbol P} \prec \underline{\boldsymbol Q}$. In general, if $\delta
\geq 4$, there exists an $\underline{\boldsymbol P}$ and $\underline{\boldsymbol Q}$
such that $\underline{\boldsymbol P} \prec_T \underline{\boldsymbol Q}$, but
$\underline{\boldsymbol P} \nprec \underline{\boldsymbol Q}$. In the $\delta \geq 4$
case, the set of vectors majorized by $\underline{\boldsymbol Q}$ is a strict
subset of the set of vectors catalytically majorized by it provided that
$\underline{\boldsymbol Q}$ has at least four distinct components.
%\endproof
\end{prop}
\begin{proof}
For (a), we first note that given that there is no simple algorithm that captures
the order of non-increasing entries of $\underline{\boldsymbol P} \otimes \underline{\boldsymbol L}$,
verifying the $\delta m$ inequalities of~(\ref{equality}) is a difficult exercise
in general. To overcome this problem, we note that $\underline{\boldsymbol P} \prec
\underline{\boldsymbol Q}$ is equivalent to the fact from~\citet[Prop.\ 4.B.3, p.\ 109]{olkin}
that $\sum_{i = 1}^{\delta} {\mb P}(i) = \sum_{i = 1}^{\delta} {\mb Q}(i)$ and
$\sum_{i = 1}^{\delta} \left({\mb P}(i) - t \right)^+ \leq \sum_{i = 1}^{\delta} \left(
{\mb Q}(i) - t \right)^+$ for all $t \in {\mathbb R}$ where $(x)^+ = \max(x, \hsppp 0)$
is the positive part function. From this fact, assuming without loss in generality that
${\mb L}(m) > 0$, we can easily see that
\begin{eqnarray}
\sum_{i = 1}^{\delta} \sum_{j = 1}^m {\mb P}(i) {\mb L}(j) =
\sum_i {\mb P}(i) \cdot \sum_j {\mb L}(j) = 1 = \sum_i {\mb Q}(i) \cdot
\sum_j {\mb L}(j) = \sum_{i = 1}^{\delta} \sum_{j = 1}^m {\mb Q}(i) {\mb L}(j).
\nonumber
\end{eqnarray}
Further, for all $t \in {\mathbb R}$, we have
\begin{eqnarray}
& & \sum_{j = 1}^m \sum_{i = 1}^{\delta} \left( {\mb P}(i) {\mb L}(j) - t \right)^+
= \sum_{j = 1}^m {\mb L}(j) \cdot \sum_{i = 1}^{\delta} \left( {\mb P}(i) -
\frac{t}{ {\mb L}(j)} \right)^+ \nonumber \\
& & {\hspace{0.3in}}
\stackrel{(a)}{\leq}
\sum_{j = 1}^m {\mb L}(j) \cdot \sum_{i = 1}^{\delta} \left( {\mb Q}(i) -
\frac{t}{ {\mb L}(j)} \right)^+
= \sum_{j = 1}^m \sum_{i = 1}^{\delta} \left( {\mb Q}(i) {\mb L}(j) - t \right)^+
\nonumber
\end{eqnarray}
where (a) follows from the assumption that
$\underline{\boldsymbol P} \prec \underline{\boldsymbol Q}$. This implies that
$\underline{\boldsymbol P} \otimes \underline{\boldsymbol L} \prec
\underline{\boldsymbol Q} \otimes \underline{\boldsymbol L}$.

For (b), note that catalytic majorization $\underline{\boldsymbol P}
\prec_T \underline{\boldsymbol Q}$ implies that (see~\cite{jonathan_plenio})
${\mb P}(1) {\mb L}(1) \leq {\mb Q}(1) {\mb L}(1)$ (which is equivalent to
${\mb P}(1) \leq {\mb Q}(1)$) and ${\mb P}(\delta) {\mb L}(\delta) \geq
{\mb Q}(\delta) {\mb L}(\delta)$ (which is equivalent to ${\mb P}(\delta)
\geq {\mb Q}(\delta)$). If $\delta = 2$, these facts clearly imply that
$\frac{1}{2} \leq {\mb P}(1) \leq {\mb Q}(1)$, which from~(\ref{eq_12}) is
equivalent to the fact that $\underline{\boldsymbol P} \prec
\underline{\boldsymbol Q}$. If $\delta = 3$, combining
${\mb P}(1) \leq {\mb Q}(1)$ and ${\mb P}(3) \geq {\mb Q}(3)$ results in
$\underline{\boldsymbol P} \prec \underline{\boldsymbol Q}$. For the first
counterexample in the $\delta \geq 4$ case, while the previous discussion
showed that $\underline{\boldsymbol P} = \left[ 0.4, \hsppp 0.35, \hsppp 0.15,
\hsppp 0.1 \right]$ and $\underline{\boldsymbol Q} = \left[ 0.45, \hsppp 0.27,
\hsppp 0.25, \hsppp 0.03 \right]$ result in
$\underline{\boldsymbol P} \nprec \underline{\boldsymbol Q}$ and
$\underline{\boldsymbol Q} \nprec \underline{\boldsymbol P}$, we also have that
$\underline{\boldsymbol P} \otimes \underline{\boldsymbol L} \prec
\underline{\boldsymbol Q} \otimes \underline{\boldsymbol L}$ with the choice
$\underline{\boldsymbol L} = \left[ 0.6, \hsppp 0.4 \right]$. For the second
counterexample, the choice $\underline{\boldsymbol P} = \left[ 0.4, \hsppp 0.27,
\hsppp 0.27, \hsppp 0.06 \right]$ and $\underline{\boldsymbol Q} = \left[ 0.5,
\hsppp 0.2, \hsppp 0.2, \hsppp 0.1 \right]$ satisfies ${\mb P}(1) \leq
{\mb Q}(1)$ and ${\mb P}(4) < {\mb Q}(4)$ and hence,
$\underline{\boldsymbol P} \nprec_T \underline{\boldsymbol Q}$. The proof of the
last statement follows from Theorem 2.4.1 of~\citet{daftuar_phd} (also,
see~\citet{daftuar}).
%\flushright \qedsymbol
\end{proof}

\begin{figure}[tbh!]
\begin{center}
\begin{tabular}{c}
\includegraphics[height=1.5in,width=4.2in] {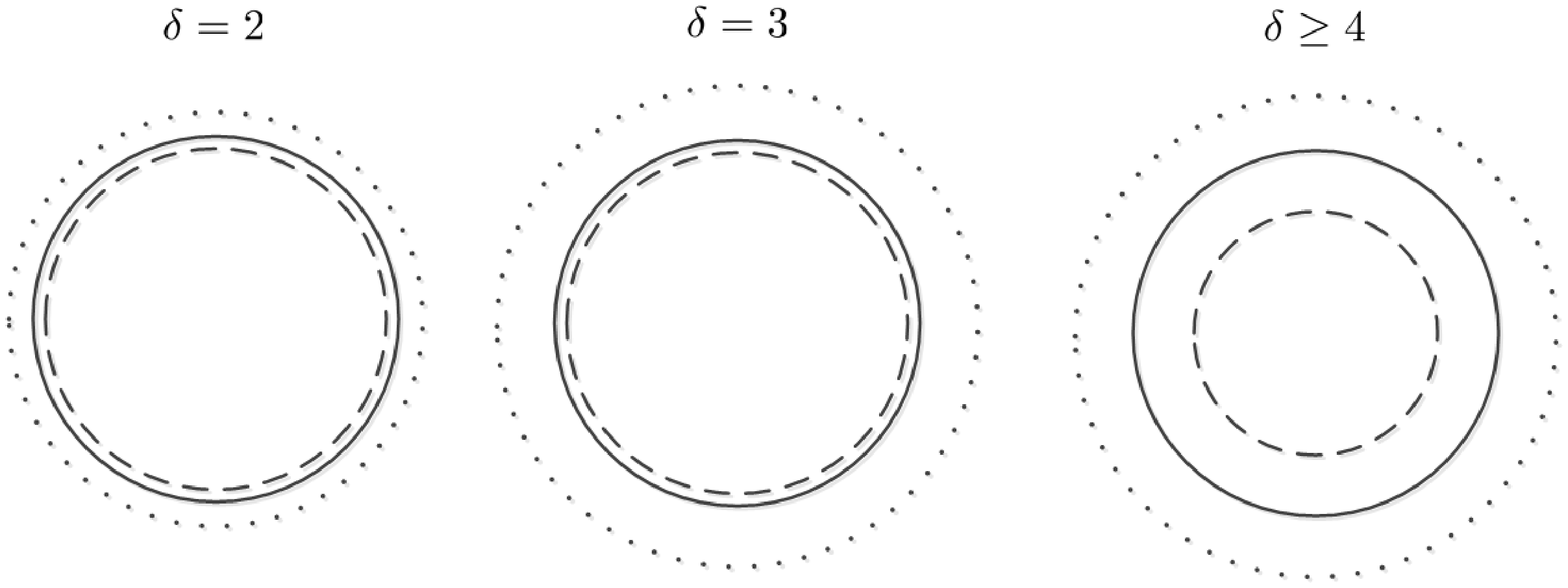}
\end{tabular}
\caption{\label{fig_comp}
Set inclusion relationships between the set of all vector pairs in ${\mathbb P}_{\delta}$
(denoted by outer circle with dots), the set of all catalytically majorizable vector
pairs (denoted by the middle circle with solid lines), and the set of all majorizable
vector pairs (denoted by the inner circle with dashed lines).
}
\end{center}
\end{figure}

The conclusions of Prop.~\ref{prop_bigger_set} in terms of the set inclusion
relationships between the set of all vector pairs in ${\mathbb P}_{\delta}$, the
set of all catalytically majorizable vector pairs, and the set of all majorizable
vector pairs are pictorially illustrated in Fig.~\ref{fig_comp} for the $\delta = 2$,
$\delta = 3$ and $\delta \geq 4$ cases. The main result
from~\citep{turgut,klimesh,klimesh_isit} on reverse catalytic majorization is provided
next.

\begin{thm}
\label{prop_trumping}
Let $\{ \underline{\boldsymbol P}, \hspp \underline{\boldsymbol Q} \}$ be
distinct elements of ${\mathbb P}_{\delta}$ with ${\mb P}(\delta) > 0$. We
have $\underline{\boldsymbol P} \prec_T \underline{\boldsymbol Q}$ if and
only if all the following conditions hold true:
\begin{eqnarray}
{\rm i)} & & {\sf PM}( \underline{\boldsymbol P}, \hsppp \alpha) <
{\sf PM}( \underline{\boldsymbol Q}, \hsppp \alpha) \hspp \hspp
{\rm if} \hspp \alpha > 1  %\hspp {\rm or} \hspp \alpha < 0
, \nonumber \\
{\rm ii)} & & {\sf PM}( \underline{\boldsymbol P}, \hsppp \alpha) >
{\sf PM}( \underline{\boldsymbol Q}, \hsppp \alpha) \hspp \hspp
{\rm if} \hspp \alpha < 1, \nonumber \\
{\rm iii)} & & {\sf SE}( \underline{\boldsymbol P}) >
{\sf SE}( \underline{\boldsymbol Q}). \nonumber
%\\
%{\rm iv)} & & {\sf GM}( \underline{\boldsymbol M}) >
%{\sf GM}( \underline{\boldsymbol N}). \nonumber
\end{eqnarray}
%On the other hand, if ${\sf N}_{\delta} = 0$, $\underline{\boldsymbol M}
%\prec_T \underline{\boldsymbol N}$ if and only if all the following
%conditions hold true:
%\begin{eqnarray}
%{\rm i)} & & {\sf PM}( \underline{\boldsymbol M}, \hsppp \alpha) <
%{\sf PM}( \underline{\boldsymbol N}, \hsppp \alpha) \hspp \hspp
%{\rm if} \hspp \alpha > 1, \nonumber \\
%{\rm ii)} & & {\sf PM}( \underline{\boldsymbol M}, \hsppp \alpha) >
%{\sf PM}( \underline{\boldsymbol N}, \hsppp \alpha) \hspp \hspp
%{\rm if} \hspp 0 < \alpha < 1, \nonumber \\
%{\rm iii)} & & {\sf SE}( \underline{\boldsymbol M}) >
%{\sf SE}( \underline{\boldsymbol N}). \nonumber
%\end{eqnarray}
\flushright \qedsymbol
%\endproof
\end{thm}
%Props.~\ref{prop_majorize} and~\ref{prop_trumping}

An alternate near-equivalent characterization of Theorem~\ref{prop_trumping}
is provided in the work by~\citet{aubrun_nechita} in terms of $\ell_p$ norms of infinite-dimensional
probability vectors with finitely many non-zero components, and an equivalent
characterization based on an alternate approach is provided in terms of general
Dirichlet polynomials and Mellin transforms by~\citet{pereira_plosker}
and in terms of completely monotone functions by~\citet{kribs}.

At this stage, it is important to %point out that the power mean function
%considered by~\citet{turgut} is a normalized version of the one in
note that
\begin{eqnarray}
\lim \limits_{ \alpha \hsppp \rightarrow \hsppp \infty}
{\sf PM}( \underline{\boldsymbol P}, \hsppp \alpha) & = &
\max \limits_{ i = 1, \hsppp \cdots, \hsppp \delta} {\mb P}(i), \nonumber \\
\lim \limits_{ \alpha \hsppp \rightarrow \hsppp -\infty}
{\sf PM}( \underline{\boldsymbol P}, \hsppp \alpha) & = &
\min \limits_{ i = 1, \hsppp \cdots, \hsppp \delta} {\mb P}(i), \hspp {\sf and}
\nonumber \\
\lim \limits_{ \alpha \hsppp \rightarrow \hsppp 0}
{\sf PM}( \underline{\boldsymbol P}, \hsppp \alpha) & = &
{\sf GM}( \underline{\boldsymbol P} ).
\nonumber
\end{eqnarray}
While we know from the proof of Prop.~\ref{prop_bigger_set} that
$\max_{i = 1, \hsppp \cdots, \hsppp \delta} {\mb P}(i) \leq
\max_{i = 1, \hsppp \cdots, \hsppp \delta} {\mb Q}(i)$ and
$\min_{i = 1, \hsppp \cdots, \hsppp \delta} {\mb P}(i) \geq
\min_{i = 1, \hsppp \cdots, \hsppp \delta} {\mb Q}(i)$ when
$\underline{\boldsymbol P} \prec_T \underline{\boldsymbol Q}$.
Theorem~\ref{prop_trumping} is along the right direction in the limiting settings of
$\alpha \rightarrow \infty$ and $\alpha \rightarrow -\infty$ except for the
modification of the strict inequality with an inclusive inequality in these
cases. Further, while the statement of Theorem~\ref{prop_trumping} has not made
any assumption on whether ${\mb Q}(\delta) > 0$ or ${\mb Q}(\delta) = 0$, it is
clear that the inequalities hold in the latter case since
\begin{eqnarray}
{\sf PM}( \underline{\boldsymbol P}, \hsppp \alpha) = 0
\hspp \hspp {\sf if} \hspp \alpha \leq 0,
\nonumber
\end{eqnarray}
with the equality seen as a limiting case of $\alpha \rightarrow 0^-$ at the
extreme point.

The importance of Theorem~\ref{prop_trumping} is in emphasizing the role of {\em only}
two specific candidate functionals (Shannon entropy and power mean) from a broad
class of functionals that might have potentially been of importance. For example
${\sf NZ}( \underline{\boldsymbol P})$ is a Schur-concave function that is not
important from the viewpoint of catalytic majorization. However, while
Theorem~\ref{prop_trumping} characterizes catalytic majorization in terms of two
functions, it is imperative to note that they correspond to an {\em uncountably}
infinite set of inequalities in the parameter $\alpha$. It is widely
conjectured by~\citet{klimesh,klimesh_isit} that a significant computational
reduction in this checking might not be possible. Nevertheless, in the following
special case, the checking of the infinitely many power mean functions is not
necessary and only the Shannon entropy function is seen to be important.

\begin{prop}
Let $\{ \underline{\boldsymbol P}, \hsppp \underline{\boldsymbol Q} \} \in
{\mathbb P}_{\delta}$  with ${\mb P}(k^{\star}) > 0$ where
$k^{\star} = \arg \max \limits_{k = 1, \hsppp \cdots, \hsppp \delta}
\left\{ {\mb Q}(k) > 0 \right\}$. If
\begin{eqnarray}
\frac{ {\mb Q}(1) }{ {\mb P}(1)} \geq \cdots \geq
\frac{ {\mb Q}(k^{\star}) }{ {\mb P}(k^{\star})},
\nonumber
\end{eqnarray}
then
\begin{eqnarray}
{\sf PM}( \underline{\boldsymbol P}, \hsppp \alpha) & < &
{\sf PM}( \underline{\boldsymbol Q}, \hsppp \alpha), \hspp
{\sf if} \hspp \alpha > 1 \hspp {\sf and}
\nonumber \\
{\sf PM}( \underline{\boldsymbol P}, \hsppp \alpha) & > &
{\sf PM}( \underline{\boldsymbol Q}, \hsppp \alpha), \hspp
{\sf if} \hspp \alpha > 1.
\nonumber
\end{eqnarray}
\end{prop}
\begin{proof}
The proof follows from the monotonicity of ratio of means in the $\alpha$
parameter~\citep[Prop.\ 5.B.3, p.\ 130]{olkin} under the assumptions made
in the statement of the proposition and by comparing it with the $\alpha = 1$
case where the ratio of the means of $\underline{\boldsymbol P}$ and
$\underline{\boldsymbol Q}$ is $1$.
\end{proof}

%\begin{supplement} [id=suppC]
%\label{sec_suppc}
%\stitle{Background on EWMA approach and set-up}
%\slink[doi]{}%{10.1214/10-AOAS395SUPPB}
%\sdatatype{.pdf}
%\sdescription{This section sets up a semi-parametric approach based on the
%EWMA control chart for tracking terrorist activity.}
%\end{supplement}

\bibliographystyle{imsart-nameyear}
\bibliography{dtra_bib2}
%\bibliographystyle{IEEEbib}
%\bibliography{newrefsx}

\end{document}